\documentclass[11pt,a4paper]{amsart}
\usepackage[foot]{amsaddr}
\usepackage{ifxetex}
\ifxetex
  \usepackage[no-math]{fontspec}
\else
\fi
\usepackage{amsmath}
\usepackage{amsfonts}
\usepackage{amssymb}
\usepackage{amsthm}
\usepackage{fullpage}
\usepackage{microtype}
\usepackage[libertine]{newtxmath}
\usepackage[tt=false]{libertine} 
\usepackage{caption}
\usepackage{bbm}
\usepackage{hyperref, color}
\hypersetup{colorlinks=true,citecolor=blue, linkcolor=blue, urlcolor=blue}
\usepackage[linesnumbered,boxed,ruled,vlined]{algorithm2e}
\usepackage{bm}
\usepackage{bbm}
\usepackage[numbers]{natbib}
\usepackage{xcolor}
\usepackage{enumerate} 
\usepackage{tabularx}
\usepackage{array}
\usepackage{cleveref}
\newcolumntype{L}[1]{>{\raggedright\arraybackslash}p{#1}}
\newcolumntype{C}[1]{>{\centering\arraybackslash}m{#1}}
\newcolumntype{R}[1]{>{\raggedleft\arraybackslash}p{#1}}

\usepackage{makecell}

\usepackage{cleveref}
\usepackage{color}
\usepackage{graphicx}
\usepackage{float}

\usepackage{tikz} 
\usetikzlibrary{positioning}
\tikzstyle{square}  = [rectangle, minimum size = 7.5mm, thick, draw = black]
\tikzstyle{cycle}  = [rectangle, minimum size = 7.5mm, thick, draw = black]
\usepackage{subfig} 
\usepackage[justification = centering, labelsep =period]{caption}

\usepackage{footnote}
\makesavenoteenv{tabular}

\theoremstyle{plain}
\newtheorem{thm}{Theorem}[section]
\newtheorem{lem}[thm]{Lemma}
\newtheorem{prop}[thm]{Proposition}

\newtheorem{corol}[thm]{Corollary}

\DeclareMathAlphabet{\mathcal}{OMS}{cmsy}{m}{n}

\theoremstyle{definition}
\newtheorem{lemma}{Lemma}[section]
\newtheorem{definition}{Definition}[section]

\newtheorem*{claim}{Claim}

\newcommand{\CInterval}{\mathcal{I}_c}
\newcommand{\VInterval}{\mathcal{I}_v}

\newcommand{\Boundary}{\partial}
\newcommand{\Interior}{\mathcal{I}}

\theoremstyle{remark}
\newtheorem*{rem}{Remark}
\newtheorem*{example}{Example}

\newcommand{\Neighbor}{\mathcal{N}}

\newcommand{\GDependencyGraph}{G_D}
\newcommand{\DependencyGraph}{G}

\newcommand{\GBipartiteGraph}{G_B}

\newcommand{\Subspace}{V}

\newcommand{\SubspaceSet}{\vvec{V}} 
\newcommand{\GDependencyGraphB}{G_D(G_B)}
\newcommand{\EventSet}{\vvec{A}}
\newcommand{\EventSetB}{\vvec{B}}
\newcommand{\Event}{A}

\newcommand{\Neg}[1]{\overline{#1}}
\renewcommand{\vec}[1]{\mathbf{#1}}
\renewcommand{\Pr}{\mathbb{P}}

\newcommand{\R}{\mathbb{R}}

\newcommand{\bigplus}{\sum}

\newcommand{\EventVariable}{event-variable graph}

\renewcommand{\vec}[1]{\boldsymbol{#1}}
\newcommand{\vvec}[1]{\vec{#1}}

\newcommand\abs[1]{\left|#1\right|}
\DeclareMathOperator*{\ddim}{\dim} 
\DeclareMathOperator*{\RR}{\mathbb{R}} 
\newcommand{\lloc}{\ast}   

\def\vecp{\vec{p}}

\begin{document}
\title{Quantum Lov{\'a}sz Local Lemma: Shearer's Bound is Tight}

\author{
Kun He$^{*}$, Qian Li $^{\dag}$, Xiaoming Sun$^{\dag}$ and Jiapeng Zhang $^{\ddag}$
}

\footnotetext[1]{The Key Lab of Data Engineering and Knowledge Engineering, MOE, Renmin University of China, Beijing, China. 
Email:\texttt{hekun.threebody@foxmail.com}}

\footnotetext[2]{Institute of Computing Technology, Chinese Academy of Sciences. University of Chinese Academy of Sciences. Beijing, China.
Email:\texttt{liqian,sunxiaoming@ict.ac.cn}}
\footnotetext[3]{University of Southern California. Email:\texttt{jiapengz@usc.edu}}
\date{}
\maketitle

\begin{abstract}
The Lov{\'a}sz Local Lemma (LLL) is a very powerful tool in combinatorics and probability theory to show the possibility of avoiding all bad events under some weakly dependent conditions. 
In a seminal paper, Ambainis, Kempe, and Sattath (JACM 2012) introduced a quantum version LLL (QLLL) which shows the possibility of avoiding all ``bad" Hamiltonians under some weakly dependent condition, and applied QLLL to the random k-QSAT problem. Sattath, Morampudi, Laumann, and Moessner (PNAS 2015) extended Ambainis, Kempe, and Sattath's result and showed that Shearer's bound is a sufficient condition for QLLL, and conjectured that Shearer's bound is indeed the tight condition for QLLL.

In this paper, we affirm this conjecture. Precisely, we prove that Shearer's bound is tight for QLLL, i.e., the relative dimension of the smallest satisfying subspace is completely characterized by the independent set polynomial. 
Our result implies the tightness of Gily{\'e}n and Sattath's algorithm (FOCS 2017), and also implies that the lattice gas partition function fully characterizes quantum satisfiability for almost all Hamiltonians with large enough qudits (Sattath, Morampudi, Laumann and Moessner, PNAS 2015).

The commuting LLL (CLLL), which focuses on commuting local Hamiltonians, is also investigated here. We prove that the tight regions of CLLL and QLLL are different in general. This result indicates that it is possible to design an algorithm for CLLL which is still efficient beyond Shearer's bound.

\end{abstract}

\newpage

\section{Introduction}
\noindent\textbf{Classical Lov{\'a}sz Local Lemma}\quad The \emph{Lov{\'a}sz Local Lemma} (or LLL) is a very powerful tool in combinatorics and probability theory to show the possibility of avoiding all ``bad" events under some ``weakly dependent" condition. 
Formally, given a set $\EventSet=\{A_1,\cdots,A_m\}$ of bad events in a probability space, the LLL is a condition on the probabilities of $\EventSet$ and the dependency among $\EventSet$  under which $\Pr(\cap_{\Event\in\EventSet}\Neg{\Event}) >0$. The dependency among events is characterized by an undirected graph, called dependency graph. Precisely, a dependency graph is an undirected graph $\GDependencyGraph=([m],E_D)$ such that for any vertex $i$, $\Event_i$ is independent of $\{\Event_j: j\notin \Gamma_i \cup \{i\}\}$, where $\Gamma_i$ stands for the set of neighbors of $i$ in $\GDependencyGraph$. In this setting, finding the conditions under which $\Pr(\cap_{\Event\in\EventSet}\Neg{\Event}) >0$ is reduced to the following problem: given a graph $\GDependencyGraph$, determine its abstract interior $\mathcal{I}(\GDependencyGraph)$ which is the set of vectors $\vec{p}$ such that $\Pr\left(\cap_{\Event\in \EventSet} \Neg{\Event} \right)>0$ for any event set $\EventSet$ with dependency graph $\GDependencyGraph$ and probability vector $\vec{p}$. Local solutions to this problem, including the first LLL proved in 1975 by Erd{\H{o}}s and Lov{\'a}sz ~\cite{erdos1975problems}, are referred as abstract LLLs. 

The most frequently used abstract LLL is as follows:
\begin{thm}[\cite{spencer1977asymptotic}]\label{thm:asymmetric}
Given a dependency graph $\DependencyGraph_D=([m],E_D)$ and a probability vector $\vec{p}\in [0,1]^n$, if there exist real numbers $x_1,...,x_n \in (0,1)$ such that $p_i \leq x_i \prod_{j\in  \Gamma_i} (1-x_j)$ for any $i\in[m]$, then $\vec{p}\in \mathcal{I}(\DependencyGraph_D)$.
\end{thm}

Shearer~\cite{shearer1985problem} provided the exact characterization of $\mathcal{I}(\DependencyGraph_D)$ with the independence polynomial defined as follows.

\begin{definition}[Multivariate independence polynomial]\label{def: indpoly}
Let $G_D=([m],E_D)$, $\vec{x}=(x_i:i\in [m])$, and $\mathsf{Ind}(G_D)$ be the set of all independent sets of $G_D$. 
For each vertex set $S\subseteq [m]$, define
$$I(G_D,\vec{x},S)\triangleq \sum_{T\subseteq S:T\in \mathsf{Ind}(G_D)}(-1)^{|T|}\prod_{i\in T}x_i.$$ 
We call $I(G_D,\vec{x}) \triangleq I(G_D,\vec{x},[m])$
the \emph{multivariate independence polynomial} of $G_D$.
\end{definition}

\begin{definition}[Shearer's bound]\label{def:shearerbound}
Given a dependency graph $G_D=([m],E_D)$,
a probability vector $\vecp=(p_1,\cdots,p_m)\in [0,1]^{m}$ is called \emph{beyond the Shearer's bound} of $G_D$ if there exists a vertex set $S\subseteq [m]$ such that $I(G_D,\vecp,S)\leq 0$. Otherwise we say $\vecp$ is \emph {in the Shearer's bound} of $G_D$.
\end{definition}

The tight criterion under which the abstract LLL holds provided by Shearer is as follows.
Given any set of events $\vec{A} = \{A_1,\cdots,A_m\}$, 
we use $\vec{A}\sim (G_D,(p_1,\cdots,p_m))$ to denote that
$\vec{A}$ has $G_D$ as dependency graph and satisfies $\Pr(A_i)= p_i$ for any $i\in [m]$.

\begin{thm}[\cite{shearer1985problem}]\label{thm:shearer1985problem}
For a dependency graph $G_D =([m],E_D)$ and a probability vector $\vecp=(p_1,\cdots,p_m)\in [0,1]^{m}$, the following conditions are equivalent:
\begin{enumerate}
  \item $\vecp$ is in the Shearer's bound of $G_D$.
  \item for any probability space $\Omega$, and any event set $\vec{A} = \{A_i\subseteq\Omega:i\in [m]\}$ where $\vec{A}\sim (G_D,\vecp)$, we have $\Pr(\cap_{A\in \vec{A}}\Neg{A})\geq I(G_D,\vecp)>0$.
\end{enumerate}
In addition, let $\delta=I(G_D,\vecp)$
if $\vec{p}\in \mathcal{I}(\DependencyGraph_D)$
and be $\delta=0$ otherwise.
Then there exists a set of events $\vec{A}\sim (G_D,\vecp)$
such that $\Pr(\cap_{A\in \vec{A}}\Neg{A}) = \delta$.
\end{thm}
\noindent In other words, $\vec{p}\in \mathcal{I}(\DependencyGraph_D)$ if and only if $\vecp$ is in the Shearer's bound of $G_D$.

\emph{Variable version Lov{\'a}sz Local Lemma} (or VLLL), which focuses on variable-generated events, is another important version of LLL. In this version, there is a set of mutually independent random variables $\mathcal{X}=\{X_1,\cdots,X_n\}$, and each event $A_i$ can be fully determined by some subset $\mathcal{X}_i\subseteq\mathcal{X}$ of variables. 
An event-variable graph is a bipartite graph $\GBipartiteGraph=([m],[n],E_B)$ such that for any $X_j \in \mathcal{X}_i$, there is an edge $(i,j) \in [m]\times [n]$. Similar to the abstract-LLL, the VLLL is for solving the following problem: given a bipartite graph $\GBipartiteGraph$, determine its variable interior $\mathcal{VI}(\GBipartiteGraph)$ which is the set of vectors $\vec{p}$ such that $\Pr\left(\cap_{\Event\in \EventSet} \Neg{\Event} \right)>0$ for any variable-generated event system $\EventSet$ with \EventVariable ~$\GBipartiteGraph$ and probability vector $\vec{p}$.

The VLLL is important because many problems in which LLL has applications naturally conform with the variable setting, including hypergraph coloring \cite{mcdiarmid1997hypergraph}, satisfiability \cite{gebauer2016local,gebauer2009lovasz}, counting solutions to CNF formulas \cite{moitra2016approximate}, acyclic edge coloring \cite{giotis2017acyclic}. Moreover, a major line of research on constructive LLLs is based on the variable model \cite{moser2010constructive,pegden2014extension,kolipaka2011moser,soda23mt}.

A key problem around the VLLL is whether Shearer's bound is tight for VLLL \cite{kolipaka2011moser}. Precisely, given a bipartite graph $\GBipartiteGraph=(U,V,E)$, its \emph{base graph} is defined as the graph $\GDependencyGraphB=(U,E')$ such that for any two nodes $u_i, u_j \in U$, there is an edge $(u_i, u_j) \in E'$ if and only if $u_i$ and $u_j$ share some common neighbor in $\GBipartiteGraph$. Observe that $\GDependencyGraphB$ is a dependency graph of the variable-generated event system with \EventVariable ~$\GBipartiteGraph$. Thus we have $\mathcal{I}(\GDependencyGraphB)\subseteq \mathcal{VI}(\GBipartiteGraph)$. If $\mathcal{I}(\GDependencyGraphB)\neq\mathcal{VI}(\GBipartiteGraph)$, we say that Shearer's bound is not tight for $\GBipartiteGraph$, or $\GBipartiteGraph$ has a gap. The first example of gap existence is a bipartite graph whose base graph is a cycle of length 4 \cite{kolipaka2011moser}. Then, He et al.~ \cite{he2017variable} showed that Shearer's bound is generally not tight for VLLL. More precisely, Shearer's bound is tight if the base graph $G_D(G_B)$ is a tree, while not tight if $G_D(G_B)$ has an induced cycle of length at least 4. The remaining case when $G_D(G_B)$ has only 3-cliques is partially solved.

\vspace{1em}
\noindent\textbf{Quantum Satisfiability and Quantum Lovasz Local Lemma}\quad Most systems of physical interest can be described by local Hamiltonians $H=\sum_i H_i$ where each $k$-local term $H_i$ acts nontrivially on at most $k$ qudits. We say $H$ is frustration free if the ground state of $H$ is also the ground state of every $H_i$. Let $\Pi_i$ be the projection operator on the excited states of $H_i$ and $\Pi=\sum\Pi_i$, then it is easy to see that the frustration freeness of $H$ and $\Pi$ are the same. Henceforth, we only care about the Hamiltonians that are projectors. Determining whether a given $\Pi$ is frustration free (or satisfiable, in computer science language), known as the quantum satisfiability problem, is a central pillar in quantum complexity theory, and has many applications in quantum many-body physics.

Unfortunately, the quantum satisfiability problem has been shown to be QMA$_1$-complete~\cite{bravyi2011efficient}, which is widely believed to be intractable in general even for quantum computing. This makes it highly desirable to search for efficient heuristics and algorithms in order to, at least,
partially answer this question.

In a seminal paper, by generalizing the notations of probability and independence as described in the following table, Ambainis, Kempe, and Sattath~\cite{ambainis2012quantum} introduced a quantum version LLL (or QLLL), which is a sufficient condition on the relative dimensions and the dependency graph under which the Hamiltonian is guaranteed to be frustration free. Here, the relative dimension of a Hamiltonian is defined as that of the subspace it projects. Utilizing the QLLL, they \cite{ambainis2012quantum} greatly improved the known critical density for random $k$-QSAT from $\Omega(1)$~\cite{Laumann2010On} to $\Omega(2^k/k^2)$, almost meeting the best known upper bound of $O(2^k)$~\cite{Laumann2010On}.

\vspace{1em}
\begin{tabular}{lll}
Probability space $\Omega$ & $\rightarrow$ &  Vector space $V$\\
Event $A$ & $\rightarrow$ & Subspace $A\subseteq V$\\
Complement $\overline{A}=\Omega\backslash A$ & $\rightarrow$ & Orthogonal completementary subspace $A^{\perp}$\\
Probability $\Pr(A)$ & $\rightarrow$ & Relative dimension $\R(A):=\frac{\dim(A)}{\dim(V)}$\\
Disjunction $A\vee B$ & $\rightarrow$ & $A+B=\{a+b|a\in A, b\in B\}$\\
Conjunction $A\wedge B$ & $\rightarrow$ & $A\cap B$\\
Independence $\Pr(A\wedge B) = P(A)\cdot P(B)$ & $\rightarrow$ & $\R(A\cap B)=\R(A)\cdot \R(B)$\\
Conditioning $\Pr(A|B)=\frac{\Pr(A\wedge B)}{\Pr(B)}$ & $\rightarrow$ & $\R(A|B):=\frac{\R(A\cap B)}{\R(B)}$
\end{tabular}
\vspace{1em}
\\

Then, Sattath et al.~\cite{pnas} leveraged Shearer's technique to the QLLL and showed that Shearer's bound is still a sufficient condition for QLLL. Here, the interaction bipartite graph, which is the quantum analog of the classical event-variable graph, is used to characterize the dependency between Hamiltonians and qudits. In an interaction  bipartite graph, the left vertices represent Hamiltonians, the right vertices represent qudits, and an edge between a left vertex and a right vertex means the corresponding Hamiltonian acts on the corresponding qudit. Remarkably, the probability threshold of Shearer's bound turns out to be the first negative
fugacity of  the hardcore lattice gas partition function, which has been extensively studied in classical statistical mechanics. Utilizing tools in classical statistical
 mechanics, they concretely apply QLLL to evaluate the critical threshold for various regular lattices. In contrast to the classical VLLL \cite{he2017variable} which generally goes beyond Shearer's bound, Sattath et al. \cite{pnas} conjectured that Shearer's bound is tight for QLLL, which, if true, would have important physical significance and several striking consequences ~\cite{pnas}.

In the past few years, as a special case of the quantum satisfiability problem, the commuting local Hamiltonian problem (CLH), where $[\Pi_i,\Pi_j]=0$ for all $i$ and $j$, has attracted considerable attention~\cite{bravyi2005commutative,aharonov2011complexity,schuch2011complexity,gottesman2009quantum,aharonov2018complexity}. Commuting Hamiltonians are somewhat ``halfway" between classical and quantum, and are capable of exhibiting intriguing multi-particle entanglement phenomena, such as the well-known toric code~\cite{kitaev2003fault}. CLH interests people not only because the commutation restriction is natural and often made in physics, but also because it may help us to understand the centrality of non-commutation in quantum mechanics. CLH can be viewed as a generalization of the classical SAT, thus CLH is at least NP-hard, and as a sufficient condition, the commuting version LLL (or CLLL) is desirable and would have various applications.

The QLLL and CLLL provide sufficient conditions for frustration freeness. A natural question is whether there is an efficient way to prepare a frustration-free state under the condition of QLLL or CLLL. A series of results showed that the answer is affirmative if all local Hamiltonians commute \cite{schwarz2013information, Cubitt2012A, sattath2015constructive}. Recently, Gily{\'e}n and Sattath \cite{gilyen2016preparing} improved the previous constructive results by designing an algorithm that works efficiently under Shearer's bound for non-commuting terms as well under the condition that the Hamiltonian has a uniform inverse polynomial gap. Here, a uniform gap is the minimum energy gap among the system and all its subsystems \cite{gilyen2016preparing}.

Therefore, the following two closely related problems beg answers:

\begin{enumerate}
\item Tight region for QLLL: complete characterization of the interior of QLLL, $\mathcal{QI}(G_B)$, for a given interaction bipartite graph $\GBipartiteGraph$. Here the interior $\mathcal{QI}(G_B)$ is the set of vectors $\vec{r}$ such that any local Hamiltonians with relative dimensions $\vec{r}$  and interaction bipartite graph $\GBipartiteGraph$ are frustration free. As Shearer's bound has been shown to be a sufficient condition for QLLL~\cite{pnas}, a fundamental open question here is whether Shearer's bound is tight. If it is tight, there are several striking consequences. First, the tightness implies that Gily{\'e}n and Sattath's algorithm \cite{gilyen2016preparing} converges up to the tight region assuming a uniform inverse-polynomial spectral gap of the Hamiltonian.
      Second, the geometrization theorem \cite{laumann2009phase} says that given the interaction bipartite graph, dimensions of qudits, and dimensions of local Hamiltonians, either all such Hamiltonian are frustration free, or almost all such Hamiltonians are not. If Shearer's bound turns out to be tight for the QLLL, by geometrization theorem we know that the quantum satisfiability for almost all Hamiltonians with large enough qudits can be completely characterized by the lattice gas partition function. The lattice gas critical exponents can be directly applied to count of the ground state entropy of almost all quantum Hamiltonians in the frustration free regime. Thus, the tightness means a lot for transferring insights from classical statistical mechanics into the quantum complexity domain ~\cite{pnas}.

\item Tight region for CLLL: complete characterization of the interior of CLLL, $\mathcal{CI}(G_B)$, for a given interaction bipartite graph $\GBipartiteGraph$. Here the interior $\mathcal{CI}(G_B)$ is the set of vectors $\vec{r}$ such that any \emph{commuting} Hamiltonians with relative dimensions $\vec{r}$ and interaction bipartite graph $\GBipartiteGraph$ are frustration free. Obviously, the interior of the CLLL is a superset of that of the QLLL for any $G_B$. An interesting question that remains is whether the containment is proper. There are a series of results on the algorithms for preparing a frustration-free state for commuting Hamiltonians under the conditions of QLLL ~\cite{schwarz2013information, Cubitt2012A, sattath2015constructive}. Thus if the containment turns out to be proper, it might be possible to design a more specialized algorithm for commuting Hamiltonians that is still efficient beyond the conditions of QLLL, e.g., Shearer's bound. The tight region for CLLL requires characterization not only due to the various applications in CLH, but also because it may help us to understand the role of non-commutation in the quantum world.


\end{enumerate}
\vspace{1em}

\subsection{Results and Discussion}
	We provide a complete answer to the first problem. Specifically, we show that Shearer's bound is tight for QLLL. We also study CLLL and partially answer the second problem. Precisely, we show that in contrast to QLLL, the interior of CLLL goes beyond Shearer's bound generally. The main results are listed and discussed as follows.

In this work, the interaction bipartite graph of Hamiltonians and the classical event-variable graph are both denoted by the bipartite graph $\GBipartiteGraph=([m],[n],E_B)$. We call the vertices in $[m]$ the left vertices and those in $[n]$ the right vertices. Usually, we will index the left vertices with  ``$i$" and the right vertices with ``$j$". In $\GBipartiteGraph$, there may be two vertices with the same index $k$: one is a left vertex and the other is a right vertex. In this paper, there will never be ambiguity in identifying which vertex is which from the context.

\subsubsection{Tight Region for QLLL} In this paper, we first prove the tightness of Shearer's bound for QLLL, which affirms the conjecture in \cite{pnas,Morampudi2018Many}. Precisely,

\begin{thm}[Informal]\label{Shearer'sboundistightforQLLL}
	Given an interaction bipartite graph $G_B = ([m],[n],E_B)$ and a rational vector $\vec{r} \in [0,1]^m$, consider the Hamiltonians $\Pi=\sum_i \Pi_i$ with relative dimensions $\vec{r}$ and interaction bipartite graph $G_B$.

	\begin{itemize}
		\item If $\vec{r}\in \mathcal{I}(\GDependencyGraphB)$, then  $\R(\emph{ker } \Pi)\geq I(G_D(G_B),\vec{r})>0$ \emph{~\cite{pnas}} for all such Hamiltonians. For qudits of proper dimensions, this lower bound can be achieved by almost all such Hamiltonians acting on these qudits.
		Moreover, there exists a $\vec{d}_0$ such that for all qudits with dimensions $\vec{d} \geq \vec{d}_0$, we have $\R(\emph{ker } \Pi) \leq I(G_D(G_B),\vec{r})+\epsilon$ for almost all such Hamiltonians, where $\epsilon>0$ can be arbitrarily small as $\vec{d}_0$ goes to infinity.
		\item Otherwise, for qudits of proper dimensions, almost all such Hamiltonians acting on these qudits are not frustration free. Furthermore, there exists a $\vec{d}_0$ such that for all qudits with dimensions $\vec{d} \geq \vec{d}_0$, we have $\R(\emph{ker } \Pi) \leq \epsilon$ for almost all such Hamiltonians, where $\epsilon>0$ can be arbitrarily small as $\vec{d}_0$ goes to infinity.
	\end{itemize}
\end{thm}
In contrast to the classical VLLL which goes beyond Shearer's bound generally, QLLL is another example of the difference between the classical world and the quantum world. As mentioned above, Theorem \ref{Shearer'sboundistightforQLLL} means that the position of the first negative fugacity zero of the lattice gas partition function is exactly the critical threshold of quantum satisfiability for almost all Hamiltonians with large enough qudits, and the relative dimension of the smallest satisfying subspace is exactly characterized by the independent set polynomial. Additionally, Theorem \ref{Shearer'sboundistightforQLLL} also implies the tightness of Gily{\'e}n and Sattath's algorithm \cite{gilyen2016preparing}, which efficiently prepares  a frustration free state in the Shearer's bound assuming a uniform inverse polynomial spectral gap.

Independently, Siddhardh Morampudi and Chris Laumann showed that Shearer's bound is tight for a large class of graphs \cite{Morampudi2018Many}. Our result shows that Shearer's bound is tight for any graph.

Finally, the $\vec{d}_0$ that we obtain is tremendously large (see the formal statement of Theorem \ref{Shearer'sboundistightforQLLL} in Section \ref{sec: qlll}). We are curious about how small $\vec{d}_0$ can be, and particularly whether $\vec{d}_0$ can be polynomially bounded by the vector $\vec{r}$. This open problem is important especially for the computational aspects of QLLL.

It seems \cite{Bravyi2010,Laumann2010On,ambainis2012quantum,pnas} that QLLL has three ranges: for sufficiently small relative dimensions, there is a classical (unentangled) satisfying state, and when the relative dimensions of Hamiltonians increase the states need to become entangled in order to satisfy all Hamiltonians, just before the system becomes unsatisfiable. As only two ranges are studied in Theorem \ref{Shearer'sboundistightforQLLL}, namely the satisfiable region and the unsatisfiable region, it is another important open problem to investigate when the satisfying state must be entangled.

\subsubsection{Tight Region for CLLL}

We partially depict the interior of CLLL. Precisely, we obtain the following results.

\vspace{1ex}
\noindent\underline{Solitary qudits are classical.} 
Given $G_B$, we call a right vertex $j\in[n]$ solitary if for any  $i_1,i_2\in\mathcal{N}(j)$ we have $\mathcal{N}(i_1)\cap \mathcal{N}(i_2)=\{j\}$. Here, $\mathcal{N}(j)$ stands for the set of neighbors of $j$ in $G_B$. We prove that all qudits $\mathcal{H}_j$ where $j$ is solitary can be restricted to be classical variables without changing the interior (Theorem \ref{thm:solitaryisclassical}). As a corollary, if all right vertices in $G_B$ are solitary, then the interior of CLLL equals to that of VLLL. In particular, when $G_B$ is a cycle of length at least 6, we have the tight region of CLLL goes beyond Shearer's bound, as that of VLLL does \cite{he2017variable}.  

\vspace{1ex}
\noindent\underline{Leveraging tools for VLLL to CLLL.} We leverage two tools developed in \cite{he2017variable} for deciding whether Shearer's bound is tight for VLLL on a given $G_B$ to CLLL. The first tool, namely Theorem \ref{Conj:GapGeom}, enables us to prove Shearer's bound is tight for CLLL on a given $G_B$ just by constructing a commuting local Hamiltonians and without computing the interior of CLLL or the Shearer's bound. The second tool is a  set of reduction rules with which we can infer whether Shearer's bound is tight for CLLL on a given interaction bipartite graph from known graphs.

\vspace{1ex}
\noindent\underline{An almost complete characterization of graphs on which Shearer's bound is tight for CLLL.}  
Based on the above results, we prove that Shearer's bound is not tight for CLLL on many interaction bipartite graphs. Precisely, given an interaction bipartite graph $G_B$, we show that Shearer's bound is tight for CLLL if its base graph is a tree, and not tight if its base graph has an induced cycle of length at least 4. This gives an almost complete characterization of bipartite graphs on which Shearer's bound is tight for CLLL except when the base graph has only 3-cliques.

\vspace{1em}

\noindent\textbf{Organization.}\quad Section \ref{sec:definitionsandnotations} provides some basic notations. In Section \ref{sec:ShearerBoundisTight}, we prove that Shearer's bound is tight for QLLL. In Section \ref{sec:tightregionforclll}, we investigate the tight region of CLLL.

\section{Notations for Hamiltonians, qudits and instances}\label{sec:definitionsandnotations}
Let $\mathcal{H}_1,\mathcal{H}_2,\cdots, \mathcal{H}_n$ be $n$ qudits. Then the Hilbert space of the quantum system is an $n$th-order tensor product  $\mathcal{H}_1\otimes\mathcal{H}_2\otimes\cdots\otimes \mathcal{H}_n$ over $\mathbb{C}$. For any $S\subseteq [m]$, let $\mathcal{H}_{S}\triangleq \bigotimes_{i\in S}\mathcal{H}_i$ denote the Hilbert space of the qudits in $S$. For example,  $\mathcal{H}_{\{1,2\}} = \mathcal{H}_{1}\otimes\mathcal{H}_{2}$.
For simplicity, we assume that $\mathcal{H}_{\emptyset}$ satisfies $\mathcal{H}_{\emptyset}\otimes\mathcal{H} = \mathcal{H}$ for each Hilbert Space $\mathcal{H}$.
For any $j \in [n]$,  let $\dim(\mathcal{H}_{j})$ be the dimension of $\mathcal{H}_{j}$ and $\dim(\mathcal{H}_{1},\cdots, \mathcal{H}_{n})$ be $(\dim(\mathcal{H}_{1}),\dim(\mathcal{H}_{2}),\cdots,\dim(\mathcal{H}_{n}))$.

In this paper, the terms ``subspaces", ``Hamiltonians", and ``projectors" will be used interchangeably. We will use $U,V,W$ to denote subspaces.
A vector space $V$ is said to be direct sum of its subspaces $V_{1},\dots, V_{k}$, denoted as $V=V_{1}\oplus V_{2}\oplus \dots\oplus V_{k}$, if $V=V_{1}+V_{2}+\dots+V_{k}$ and $V_i\cap \sum_{\ell\neq i}V_\ell=\{0\}$ for any $1\leq i\leq k$. 
Given a Hilbert space $\mathcal{H}$ and
a subspace $V\subseteq\mathcal{H}$, let $\Pi_V$ be the projector onto $V$. The relative dimension of $V$ to $\mathcal{H}$ is defined as 
$$\mathbb{R}(V,\mathcal{H})\triangleq\frac{\mathrm{tr}(\Pi_V)}{\ddim(\mathcal{H})}=\frac{\ddim(V)}{\ddim(\mathcal{H})}.$$ 
For simplicity, we will omit  ``to $\mathcal{H}$'' and denote $\mathbb{R}(V,\mathcal{H})$ as use $\mathbb{R}(V)$ if $\mathcal{H}$ is clear from the context.
Throughout this paper, we are only interested in finite-dimensional quantum systems and restrict our attention to rational relative dimensions.

Given a bipartite graph $G_B=([m],[n],E_B)$, let $\mathcal{N}(G_B,i)$ (or $\mathcal{N}(i)$ if $G_B$ is implicit) denote the neighbors of vertex $i$ in $G_B$ if which side this vertex belongs to is clear from the context. We say two left vertices $i_1,i_2\in [m]$ are neighboring or adjacent if $\mathcal{N}(i_1)\cap \mathcal{N}(i_2) \not = \emptyset$. We say a left vertex $i\in [m]$ and a right vertex $j\in [n]$ are neighboring or adjacent if $j \in \mathcal{N}(i)$.
Given a set $S$, let $\mathcal{N}(S)$ denote $\bigcup_{i\in S}\mathcal{N}(i)$.
We say a set of local Hamiltonians $\vvec{V}=\{V_1,\cdots,V_m\}$ conforms with $G_B$, denoted by $\vvec{V}\sim G_B$, if for any $i\in[m]$, $\Pi_{V_i}$ acts trivially on the qudits $\mathcal{H}_{[n]\setminus\mathcal{N}(i)}$. Thus, we can write $V_i$ as $V_i^{\ast}\otimes \mathcal{H}_{[n]\backslash \mathcal{N}(i)}$ where $V^{\ast}_i$ is some subspace of $\mathcal{H}_{\mathcal{N}(i)}$.
Similarly, we can also define a set of events $\EventSet$ conforms with $G_B$, denoted by $\EventSet\sim G_B$.
In this paper, we usually call $G_B$ the interaction graph.

Given $\mathcal{H}_1,\mathcal{H}_2,\cdots, \mathcal{H}_n$ and a set of subspaces $\vvec{V}=\{V_1,\cdots,V_m\}$ in $\mathcal{H}_{[n]}$,
we will use $\mathbb{R}(\vvec{V})$ to represent the vector $(\mathbb{R}(V_1),\cdots,\mathbb{R}(V_m))$
and use $\dim(\vvec{V})$ to represent the vector $(\dim(V_1),\cdots,\dim(V_m))$. 
We say a set of subspaces $\vvec{V}$ is frustration free if $V_1,\cdots,V_m$ do \emph{not} span $\mathcal{H}_{[n]}$.
We will use boldface type for vectors.
For example, $\vec{V}$ stands for a set of subspaces, $\vec{r}$ stands for a relative dimension vector, $\vec{p}$ stands for a probability vector and $\vec{d}$ stands for a dimension vector. 
For any two vectors of numbers $\vec{q}=(q_1,q_2,\cdots,q_m)$ and $\vec{q}'=(q'_1,q'_2,\cdots,q'_m)$ of the same length, we say $\vec{q} \geq \vec{q}'$ if $q_i\geq q'_i$ holds for each $i\in [m]$. 
We say $\vec{q} > \vec{q}'$ if $\vec{q} \geq \vec{q}'$ and  $q_i > q'_i$ holds for some $i\in [m]$. 

Now we can define the instances and the random instance of some given $G_B,\vec{r}$ and $\vec{d}$. 

\begin{definition}[Instance and random instance]\label{def-instance}
Given an interaction graph $G_B = ([m],[n],E_B)$, a rational vector $\vec{r}=(r_1,\cdots,r_m)\in [0,1]^m$ and an integer vector $\vec{d}=(d_1,\cdots,d_n)$, we say a subspace set $\vec{V}$ of $\mathcal{H}_{[n]}$ is an instance of the setting $(G_B,\vec{r},\vec{d})$, denoted as $\vvec{V} \sim (G_B,\vec{r},\vec{d})$, if $\vvec{V} \sim G_B$, $\mathbb{R}(\vvec{V},\mathcal{H}_{[n]}) = \vec{r}$ and $\ddim(\mathcal{H}_{1},\cdots,\mathcal{H}_{n}) = \vec{d}$.
When we talk about the instances of the setting $(G_B,\vec{r},\vec{d})$,
we always assume that the Hilbert space $\mathcal{H}_{[n]}$ with the dimension vector $\vec{d}$ can admit an instance with the interaction graph $G_B = ([m],[n],E_B)$ and the relative dimension vector $\vec{r}$, i.e., each $r_i$ where $i\in [m]$ satisfies that $r_i\ddim(\mathcal{H}_{\mathcal{N}(i)})$ is an integer. 

We say a subspace set $\vec{V} = \left(V_1,\cdots,V_m\right)$ of $\mathcal{H}_{[n]}$ is a \emph{random} instance of the setting $(G_B,\vec{r},\vec{d})$ if $\vec{V}\sim (G_B,\vec{r},\vec{d})$ and
$V_i=V_i^{\ast}\otimes \mathcal{H}_{[n]\backslash \mathcal{N}(i)}$ for each left vertex $i$, where $V_i^{\ast}$ is a random subspace of $\mathcal{H}_{\mathcal{N}(i)}$ according to the Haar measure with $\mathbb{R}(V_i^{\ast},\mathcal{H}_{\mathcal{N}(i)}) = r_i$ (in particular, $V_i^{\ast} = \mathcal{H}_{\mathcal{N}(i)}$ if $r_i=1$).
Given a random instance $\vec{V}\sim(G_B,\vec{r},\vec{d})$ in $\mathcal{H}_{[n]}$, we say $\vec{V}$ spans the whole space if $\Pr[\mathbb{R}(\bigplus_{V\in\vec{V}} V,\mathcal{H}_{[n]})=1]=1$.
\end{definition}

\section{QLLL: Shearer's Bound Is Tight}\label{sec:ShearerBoundisTight}\label{sec: qlll}
This section aims at proving \Cref{Shearer'sboundistightforQLLL}. 
\Cref{sec-shearer'sboundistight} presents our main results, 
Theorems \ref{shearersboundistight:above} and \ref{Shearer'sboundistightforQLLL}, and illustrate the main idea of our proof.
We prove \Cref{shearersboundistight:above} by induction on the number of left vertices in the interaction graph.
The induced interaction graph, the induced relative dimensions,
the induced dimensions of qudits and the induced subspaces are defined
in Sections \ref{sec-Inducedinteractiongraphandrelativedimensions}, \ref{sec-induceddimensionsof qudits} and \ref{sec-inducedsubspaces}, respectively.
In \Cref{sec-subspaces}, we construct the subspaces in  \Cref{shearersboundistight:above} from the induced subspaces.
Finally, Theorems \ref{shearersboundistight:above} and \ref{Shearer'sboundistightforQLLL} are proved 
in sections \ref{sec-Proof-of-Theorem-shearersboundistight} and \ref{sec-proofShearer'sboundistightforqLLL}, respectively.

\subsection{Shearer's Bound Is Tight for QLLL}\label{sec-shearer'sboundistight}
In this subsection we show that Shearer's bound is tight for QLLL. 
Given an interaction graph $G_B$, 
the base graph of $\GBipartiteGraph$ is defined as $G_D(\GBipartiteGraph)=([m],E_D)$, where $(i_1,i_2) \in E_D$ if and only if the left vertices $i_1,i_2$ are neighbors in $G_B$.
For simplicity, we will use $G_B$ to denote $G_D(G_B)$ if there is no ambiguity.
For example, 
we let $\mathcal{I}(G_B) = \mathcal{I}(G_D(G_B))$, and $\mathsf{Ind}(G_D(G_B)) = \mathsf{Ind}(G_B)$. 
In addition, for any interaction graph $G_B=([m],[n],E_B)$ and $\vec{r} \in {(0,1]}^m$ and $S\subseteq [m]$, we let
$I(G_B,\vec{r},S) = I(G_D(G_B),\vec{r},S)$ and $I(G_B,\vec{r}) = I(G_D(G_B),\vec{r})$.
If $G_B$ and $\vec{r}$ are clear from the context, we will simply denote $I(G_B,\vec{r},S)$ as $I(S)$.
It has been proved that 
Shearer's bound is a lower bound on the relative dimension of the satisfying subspace~\cite{pnas}, more precisely, 
\begin{thm}[Restate of Theorem 1 in \cite{pnas}]\label{thm:pnas}
	Given an interaction graph $G_B = ([m],[n],E_B)$ and rational $\vec{r}=(r_1,\cdots,r_m)\in {(0,1]}^m$, if  $\vec{r}\in \mathcal{I}(G_B)$, then for any subspaces $(V_1,\cdots,V_m)$ of relative dimension $\vec{r}$, $1- \R(\sum_{i=1}^mV_i)\geq I([m])$.
\end{thm}


In addition to \Cref{thm:pnas}, we further show that Shearer's bound can be achieved on the qudits with proper dimensions.

\begin{thm}\label{shearersboundistight:above}
	Given $G_B = ([m],[n],E_B)$ and rational $\vec{r}=(r_1,\cdots,r_m)
	\in{(0,1]}^m$,
	assume that for each left vertex $i\in [m]$, we have $r_i = \frac{p_i}{q_i}$ where $p_i,q_i$ are mutually prime 
	and if $\mathcal{N}(i) = \emptyset$ then $r_i =1$.
    Let $q = \prod_{i\in [m]}q_i$. 
	Then
	\begin{enumerate}[(a)]
		\item\label{item-a-thmshearersboundistight} if $\vec{r} \not\in \mathcal{I}(G_B)$, then there exists some $\vec{d}=(d_1,\cdots,d_n)$
		such that $d_i\leq q^{4^{m}}$ for each $i\in [m]$ and $\Pr[\mathbb{R}(\bigplus_{V\in \vec{V}} V)=1]=1$ for the random instance $\vec{V}$ of the setting $(G_B,\vec{r},\vec{d})$.
		
		\item\label{item-b-thmshearersboundistight} otherwise, $\vec{r} \in \mathcal{I}(G_B)$,
		then there exists some $\vec{d}=(d_1,\cdots,d_n)$ such that
		$d_i\leq q^{8+4^{m}}$ for each $i\in [n]$ and
		$\Pr[\mathbb{R}(\bigplus_{V\in \vec{V}} V)=1-I([m])]=1$
		for the random instance $\vec{V}$ of the setting $(G_B,\vec{r},\vec{d})$.
	\end{enumerate}
\end{thm}


\begin{rem}
The interaction graph $G_B = ([m],[n],E_B)$ in above theorem can be disconnected.
Furthermore, for any left or right vertex $i$ in $G_B$,
$\mathcal{N}(i)$ can be $\emptyset$.
This is for the sake of convenience of the induction argument.
The dimensions of qudits in \Cref{shearersboundistight:above} is not optimised, but we do not expect
to be able to use the current techniques to improve it substantially. 
Our main point is that Shearer's bound can be achieved on some qudits.
\end{rem}

The following theorem further extends \Cref{shearersboundistight:above} to all large enough qudits. 

\theoremstyle{plain}


\newtheorem*{thmShearersboundistight}{Theorem \ref{Shearer'sboundistightforQLLL} }
\begin{thmShearersboundistight}[Restated]
	Given an interaction graph $G_B = ([m],[n],E_B)$ and a rational $\vec{r}\in\left(0,1\right]^m$, for any $\epsilon \in (0,1]$ 
	and any integers $\vec{d}\geq \left(n \cdot t^{1+8m4^{m}},\cdots,n \cdot t^{1+8m4^{m}} \right)$ where 
	$t = \left\lceil\frac{2m+1}{\epsilon}\right\rceil$,
	the random instance $\vec{V}$ of the setting $(G_B,\vec{r},\vec{d})$ satisfies that 
	\begin{enumerate}[(a)]
	    \item \label{Shearer'sboundistightforQLLL-a} if $\vec{r}\not\in \mathcal{I}(G_B)$, then 
	$$\Pr\left[\R\left(\bigplus_{V\in \vec{V}} V\right)\in[1-\epsilon,1]\right]=1;$$
	    \item \label{Shearer'sboundistightforQLLL-b} otherwise, 
	    $$\Pr\left[\R\left(\bigplus_{V\in \vec{V}} V\right) \in \left[1-I([m])-\epsilon,1-I([m])\right]\right]=1.$$
	\end{enumerate}
\end{thmShearersboundistight}

The core in the above two theorems is \Cref{shearersboundistight:above} (\ref{item-a-thmshearersboundistight}).
In the following, we illustrate the main idea of its proof.
A key tool in our proof is the geometrization theorem established by Laumann et al. \cite{laumann2009phase},
which shows that ``almost all'' is equivalent to ``existence".

\begin{thm}[The geometrization theorem, adapted from ~\cite{laumann2009phase}]
	\label{lem: span}
	Given any interaction graph $G_B=([m],[n],E_B)$, any dimension vector $\vec{d}$, and any relative dimension vector $\vec{r}\in [0,1]^m$, if there exists an instance $\vec{U}$ of the setting $(G_B,\vec{r},\vec{d})$ satisfying
	$\R(\bigplus_{V\in \vec{U}} V)=1$,
	then for the random instance $\vec{V}$ of the setting $(G_B,\vec{r},\vec{d})$, we have $\Pr[\mathbb{R}(\bigplus_{V\in \vec{V}} V)=1]=1$.
\end{thm}

The following lemma is used in the proof of \Cref{shearersboundistight:above}.
Given positive integers $k_1,\cdots,k_n$ and an instance of the setting $(G_B,\vec{r},\vec{d})$ spanning the whole space,
one can construct an instance of the setting $(G_B,\vec{r},(k_1d_1,\cdots,k_nd_n))$ spanning the whole space by combining $\prod_{i\in [n]}k_i$ instances of the setting $(G_B,\vec{r},\vec{d})$.

\begin{lemma}\label{prop:multiple}
    Given any $G_B$, $\vec{r}$, $\vec{d}=(d_1,\cdots,d_n)$ and any positive integers $k_1,\cdots,k_n$,
    \begin{enumerate}[(a)]
		\item \label{prop:multiple-item-a} if there exists an instance of the setting $(G_B,\vec{r},\vec{d})$ spanning the whole space, then there exists some instance of the setting $(G_B,\vec{r},(k_1d_1,\cdots,k_nd_n))$ spanning the whole space as well;
		\item \label{prop:multiple-item-b} if the random instance the setting $(G_B,\vec{r},\vec{d})$ spans the whole space, then
	 the random instance of the setting $(G_B,\vec{r},(k_1d_1,\cdots,k_nd_n))$ spans the whole space as well.
	\end{enumerate}
\end{lemma}

Now we prove \Cref{prop:multiple}. 
\begin{proof}
    We only prove (\ref{prop:multiple-item-a}) here,
    then (\ref{prop:multiple-item-b}) is immmediate by \Cref{lem: span}.
	Let $\mathcal{H}_{[n]}=\bigotimes_{i\in [n]}\mathcal{H}_n$ be a Hilbert space where $\ddim(\mathcal{H}_{i})= k_id_i$ for each $i\in [n]$. We decompose the  qudit $\mathcal{H}_{i}$ to $k_i$ orthogonal subspaces $\mathcal{H}_{i}=\bigoplus_{\ell\in[k_i]}\mathcal{H}_{i}^{\ell}$ where $\ddim(\mathcal{H}_{i}^{\ell})=d_i$ for each $i\in [n],\ell\in[k_i]$. 
	Thus for each $\vec{\ell} = \left(\ell_1,\ell_2,\cdots,\ell_n\right)$ where $\ell_i\in [k_i]$ for all $i\in [n]$, we have $(\mathcal{H}_{1}^{\ell_1},\cdots,\mathcal{H}_{n}^{\ell_n})$ are of dimensions $\vec{d}$, 
	and then can be spanned by some subspace set $\vvec{V}_{\vec{\ell}}\sim G_B$ with relative dimensions $\vec{r}$ with respect with 
	$(\mathcal{H}_{1}^{\ell_1},\cdots,\mathcal{H}_{n}^{\ell_n})$. Let $$\vvec{V}=\bigoplus_{\vec{\ell}\in[k_1]\times\cdots\times [k_n]}\vvec{V}_{\vec{\ell}}.$$ 
	It is easy to verify that $\vvec{V}$  is an instance of the setting $(G_B,\vec{r},(k_1d_1,\cdots,k_nd_n)))$ and spans $\mathcal{H}_{[n]}$.
\end{proof}

Combining \Cref{def:shearerbound} with \Cref{thm:shearer1985problem}, we have the following lemma.
\begin{lemma}\label{lemma-eqinshearersbound}
For each $G_B=([m],[n],E_B)$ and $\vec{r} =(r_1,\cdots,r_m)\in {(0,1]}^m$, the following are equivalent: 
(1) $\vec{r}$ is in Shearer's bound for $G_D(G_B)$;
(2) $I(G_B,\vec{r},S)>0$ for each $S\subseteq [m]$;
(3) $\vec{r}\in \mathcal{I}(G_B)$.
\end{lemma}

With Lemmas \ref{prop:multiple}, \ref{lemma-eqinshearersbound} and \Cref{lem: span}, our main idea can be illustrated with following example.

\begin{example}\label{example1}
	Let $G_B=([4],[4],E)$ where $E=\{(i,i),(i,i+1 \pmod 4), i\in[4]\}$ as in \Cref{pic1} (A).
	Let $\vec{r}\triangleq(r_1,r_2,r_3,r_4)=(\frac{1}{3},\frac{1}{3},\frac{1}{4},\frac{1}{4})$.
	Note that $G_D(G_B)$ is a cycle of length 4, hence by \Cref{def: indpoly} we have $I(G_B,\vec{r})=1-\sum_{i=1}^4r_i+r_1\cdot r_3+r_2\cdot r_4=0$. 
    Then one can verify that $\vec{r}\not\in \mathcal{I}(G_B)$ by \Cref{lemma-eqinshearersbound}.
    Moreover, one can verify that for the given $G_B$ and $\vec{r}$,
    the parameter $q$ in \Cref{shearersboundistight:above} is $3\times 3\times 4\times 4 = 144$. 

\begin{figure}
\centering  
\subfloat[$G_B$]{
\begin{tikzpicture}[
roundnode/.style={circle, draw=black!100, fill=white!100, thick, minimum size=5mm},
squarednode/.style={rectangle, draw=black!100, fill=white!100, thick, minimum size=5mm},
node distance=7.5mm
]
\node[squarednode]      (H1)       at(2.7,-0.5)    {$\mathcal{H}_{1}$};
\node[squarednode]      (H2)       at(2.7,-2.5)    {$\mathcal{H}_{2}$};
\node[squarednode]      (H3)       at(2.7,-4.5)    {$\mathcal{H}_{3}$};
\node[squarednode]      (H4)       at(2.7,-6.5)    {$\mathcal{H}_{4}$};
\node[roundnode]        (V1)       at(0,-0.5)    {$V_1$};
\node[roundnode]        (V2)       at(0,-2.5)    {$V_2$};
\node[roundnode]        (V3)       at(0,-4.5)    {$V_{3}$};
\node[roundnode]        (V4)       at(0,-6.5)    {$V_{4}$};

\draw[thick,-] (V1.east) -- (H1.west);
\draw[thick,-] (V1.east) -- (H2.west);
\draw[thick,-] (V2.east) -- (H2.west);
\draw[thick,-] (V2.east) -- (H3.west);
\draw[thick,-] (V3.east) -- (H3.west);
\draw[thick,-] (V3.east) -- (H4.west);
\draw[thick,-] (V4.east) -- (H1.west);
\draw[thick,-] (V4.east) -- (H4.west);

\end{tikzpicture}
}
\hspace{30pt}
\subfloat[$G_1$]{
\begin{tikzpicture}[
roundnode/.style={circle, draw=black!100, fill=white!100, thick, minimum size=5mm},
squarednode/.style={rectangle, draw=black!100, fill=white!100, thick, minimum size=5mm},
node distance=7.5mm
]

\node[squarednode]      (H1U1)       at(8.1,-0.5)    {$\mathcal{H}_{1}$};
\node[squarednode]      (H4U1)       at(8.1,-3.25)    {$\mathcal{H}_{4}$};
\node[squarednode]      (H3U1)       at(8.1,-6)    {$\mathcal{H}_{3}$};
\node[roundnode]        (V1U1)       at(5.4,-0.5)    {$V'_1$};
\node[roundnode]        (V4U1)       at(5.4,-3.25)   {$V'_{4}$};
\node[roundnode]        (V3U1)       at(5.4,-6)    {$V'_{3}$};

\draw[thick,-] (V1U1.east) -- (H1U1.west);
\draw[thick,-] (V3U1.east) -- (H3U1.west);
\draw[thick,-] (V3U1.east) -- (H4U1.west);
\draw[thick,-] (V4U1.east) -- (H4U1.west);
\draw[thick,-] (V4U1.east) -- (H1U1.west);

\end{tikzpicture}
}
\hspace{30pt}
\subfloat[$G_2$]{
\begin{tikzpicture}[
roundnode/.style={circle, draw=black!100, fill=white!100, thick, minimum size=5mm},
squarednode/.style={rectangle, draw=black!100, fill=white!100, thick, minimum size=5mm},
node distance=7.5mm
]

\node[squarednode]      (H3U2)       at(13.5,-0.5)    {$\mathcal{H}_{3}$};
\node[squarednode]      (H4U2)       at(13.5,-3.25)    {$\mathcal{H}_{4}$};
\node[squarednode]      (H1U2)       at(13.5,-6)    {$\mathcal{H}_{1}$};
\node[roundnode]        (V2U2)       at(10.8,-0.5)    {$V'_2$};
\node[roundnode]        (V3U2)       at(10.8,-3.25)    {$V'_{3}$};
\node[roundnode]        (V4U2)       at(10.8,-6)    {$V'_{4}$};

\draw[thick,-] (V2U2.east) -- (H3U2.west);
\draw[thick,-] (V3U2.east) -- (H3U2.west);
\draw[thick,-] (V3U2.east) -- (H4U2.west);
\draw[thick,-] (V4U2.east) -- (H4U2.west);
\draw[thick,-] (V4U2.east) -- (H1U2.west);

\end{tikzpicture}
}
\caption{The interaction graphs in the example}\label{pic1}
\end{figure}
    
    The proof of \Cref{shearersboundistight:above} (\ref{item-a-thmshearersboundistight}) is by induction on the number of left vertices in the interaction graph.
    To illustrate the proof, we will verify 
	\Cref{shearersboundistight:above} (\ref{item-a-thmshearersboundistight}) on the given $G_B$ and $\vec{r}$ based on the induction hypothesis that it has already been proved for each interaction graph where the number of left vertices is no more than $3$.
	By \Cref{lem: span}, it is sufficient to 
	construct an instance $\vec{V}$ spanning the whole space where $\vec{V}\sim \left(G_B,\vec{r},\vec{d}\right)$ 
	for some $\vec{d} \leq \left(144^{4^4},144^{4^4},144^{4^4},144^{4^4}\right)$.
	Our construction is as follows.
	
	Let $ \mathcal{H}_{1},\mathcal{H}_{3},\mathcal{H}_{4}$ be three qudits
	where  $(d_1,d_3,d_4) = (\ddim(\mathcal{H}_{1}),\ddim(\mathcal{H}_{3}),\ddim(\mathcal{H}_{4}))$ will be determined later.
	Let $\mathcal{H}_2\triangleq e_1\oplus e_2\simeq \mathbb{C}^2$
	where $e_1$ and $e_2$ are two orthogonal subspaces of $\mathbb{C}^2$ with dimension 1.
	Define random $V'_1,V'_2,V'_3,V'_4$ as follows.
	\begin{enumerate}
		\item $V'_1=V_1^{\ast}\otimes \mathcal{H}_{\{3,4\}}$ where $V_1^{\ast}$ is a random subspace of $\mathcal{H}_1$ with $\R(V_1^{\ast},\mathcal{H}_1)=\frac{2}{3}$.
		\item $V_2'=V_2^{\ast}\otimes \mathcal{H}_{\{1,4\}}$ where $V_2^{\ast}$ is a random subspace of $\mathcal{H}_3$ with $\R(V_2^{\ast},\mathcal{H}_3)=\frac{2}{3}$.
		\item $V_3'=V_3^{\ast}\otimes \mathcal{H}_1$ where $V_3^{\ast}$ is a random subspace of $\mathcal{H}_3\otimes \mathcal{H}_4$ with $\R(V_3^{\ast},\mathcal{H}_3\otimes \mathcal{H}_4)=\frac{1}{4}$.
		\item $V_4'=V_4^{\ast}\otimes \mathcal{H}_3$ where $V_4^{\ast}$ is a random subspace of $\mathcal{H}_4\otimes \mathcal{H}_1$ with $\R(V_4^{\ast},\mathcal{H}_4\otimes \mathcal{H}_1)=\frac{1}{4}$.
	\end{enumerate}
    Obviously, $V'_1,V'_2,V'_3,V'_4$ are subspaces in $\mathcal{H}_{1,3,4}$.
    Let $\vec{r}' = \left(\frac{2}{3},\frac{1}{4},\frac{1}{4}\right)$.
    One can verify that the interaction graph of $\left(V'_1,V'_3,V'_4\right)$ is $G_1$ as in \Cref{pic1} (B), and $\R\left(\left(V'_1,V'_3,V'_4\right),\mathcal{H}_{1,3,4}\right) = \vec{r}'$.
    Thus, $\left(V'_1,V'_3,V'_4\right)$ is a random instance of $(G_1,\vec{r}',(d_1,d_3,d_4))$.
	Similarly, the interaction graph of $\left(V'_2,V'_3,V'_4\right)$ is $G_2$ as in \Cref{pic1} (C), and $\R\left(\left(V'_2,V'_3,V'_4\right),\mathcal{H}_{1,3,4}\right) = \vec{r}'$.
	We also have 
	$\left(V'_2,V'_3,V'_4\right)$ is a random instance of $(G_2,\vec{r}',(d_1,d_3,d_4))$.
	
	In addition, by \Cref{def: indpoly} we have 
	$$I(G_1,\vec{r}') = 1 - \frac{2}{3} - \frac{1}{4} - \frac{1}{4} + \frac{2}{3}\cdot \frac{1}{4}=0.$$ 
	Combining with \Cref{lemma-eqinshearersbound},
	we have $\vec{r}'\not\in \mathcal{I}(G_1)$.
	Thus by the induction hypothesis, there is some $d_1',d_3'$ and $d_4'$
	where $d'_i\leq (3\times4\times4)^{4^3}$ for each $i\in \{1,3,4\}$
	such that the random instance of the setting $(G_1,\vec{r}',(d_1',d_3',d_4'))$ spans $\mathcal{H}_{1,3,4}$.
	Similarly, by \Cref{def: indpoly} we also have 
	$$I(G_2,\vec{r}') = 1 - \frac{2}{3} - \frac{1}{4} - \frac{1}{4} + \frac{2}{3}\cdot \frac{1}{4}=0.$$ 
	Combining with \Cref{lemma-eqinshearersbound},
	we have $\vec{r}'\not\in \mathcal{I}(G_2)$.
	Thus by the induction hypothesis, there is some $d''_2,d''_3$ and $d''_4$
	where $d'_i\leq (3\times4\times4)^{4^3}$ for each $i\in \{2,3,4\}$
	such that the random instance of the setting $(G_2,\vec{r}'',(d_1'',d_3'',d_4''))$ spans $\mathcal{H}_{1,3,4}$.
	
	Let $d_i = d_i'd_i''$ for each $i\in\{1,3,4\}$.
	By $d'_i,d''_i\leq (3\times4\times4)^{4^3}$ for each $i\in \{1,3,4\}$, we have $d_i \leq 144^{4^4}$.
	Recall that $\left(V'_1,V'_3,V'_4\right)$ is a random instance of $(G_1,\vec{r}',(d_1,d_3,d_4))$ and the random instance of the setting $(G_1,\vec{r}',(d_1',d_3',d_4'))$ spans $\mathcal{H}_{1,3,4}$.
	Thus, by \Cref{prop:multiple} (\ref{prop:multiple-item-b}) and $d_i = d_i' d_i''$ for each $i\in\{1,3,4\}$, 
	we have $\left(V'_1,V'_3,V'_4\right)$ spans the whole space $\mathcal{H}_{1,3,4}$ with probability 1.
	Similarly, one can also verify that $\left(V'_2,V'_3,V'_4\right)$ spans the whole space $\mathcal{H}_{1,3,4}$ with probability 1.
	By the union bound and the definition of $V'_1,V'_2,V'_3,V'_4$, we have there must be some $V^{\star}_1,V^{\star}_2,V^{\star}_3,V^{\star}_4$ such that the following conditions hold:
	\begin{enumerate}
		\item $V_1^{\star}$ is some subspace of $\mathcal{H}_{\{1,3,4\}}$ with $\R(V_1^{\star},\mathcal{H}_{\{1,3,4\}})=\frac{2}{3}$ which acts trivially on $\mathcal{H}_{\{3,4\}}$.
		\item $V_2^{\star}$ is some subspace of $\mathcal{H}_{\{1,3,4\}}$ with $\R(V_2^{\star},\mathcal{H}_{\{1,3,4\}})=\frac{2}{3}$ which acts trivially on $\mathcal{H}_{\{1,4\}}$.
		\item $V_3^{\star}$ is some subspace of $\mathcal{H}_{\{1,3,4\}}$ with $\R(V_3^{\star},\mathcal{H}_{\{1,3,4\}})=\frac{1}{4}$ which acts trivially on $\mathcal{H}_{1}$.
		\item $V_4^{\star}$ is some subspace of $\mathcal{H}_{\{1,3,4\}}$ with $\R(V_4^{\star},\mathcal{H}_{\{1,3,4\}})=\frac{1}{4}$ which acts trivially on $\mathcal{H}_{3}$.
		\item Both $V^{\star}_1,V^{\star}_3,V^{\star}_4$ and $V^{\star}_2,V^{\star}_3,V^{\star}_4$ span the whole space $\mathcal{H}_{\{1,3,4\}}$.
	\end{enumerate}
	
	Let $V_1 = V_1^{\star}\otimes e_1$, $V_2 = V_2^{\star}\otimes e_2$, $V_3 = V_3^{\star}\otimes \mathcal{H}_2$ and $V_4 = V_4^{\star}\otimes \mathcal{H}_2$.
	Let $\vec{V}=(V_1,V_2,V_3,V_4)$.
	One can verify that the interaction graph of $\vec{V}$ is $G_B$ as in \Cref{pic1} (A), and $\R\left(\vec{V},\mathcal{H}_{[4]}\right) =\left(\frac{1}{3},\frac{1}{3},\frac{1}{4},\frac{1}{4}\right) = \vec{r}$.
	In addition, we also have $\ddim(\mathcal{H}_1,\mathcal{H}_2,\mathcal{H}_3,\mathcal{H}_4) = (d_1,2,d_3,d_4) \leq (144^{4^4},144^{4^4},144^{4^4},144^{4^4})$.
	Moreover, by $V^{\star}_1,V^{\star}_3,V^{\star}_4$ span the space $\mathcal{H}_{\{1,3,4\}}$, we have $V_1,V_3,V_4$ span the space $\mathcal{H}_{\{1,3,4\}}\otimes e_1$.
	Similarly, by $V^{\star}_2,V^{\star}_3,V^{\star}_4$ span the space $\mathcal{H}_{\{1,3,4\}}$, we have $V_2,V_3,V_4$ span the space $\mathcal{H}_{\{1,3,4\}}\otimes e_2$.
	Thus, $V_1,V_2,V_3,V_4$ span the whole space $\mathcal{H}_{[4]} = \mathcal{H}_{\{1,3,4\}}\otimes e_1\oplus \mathcal{H}_{\{1,3,4\}}\otimes e_2$.
	In summary, we have $\vec{V}\sim \left(G_B,\vec{r},\vec{d}\right)$ 
	for some $\vec{d} \leq \left(144^{4^4},144^{4^4},144^{4^4},144^{4^4}\right)$ and $\vec{V}$ spans the whole space $\mathcal{H}_{[4]}$.
	Thus, \Cref{shearersboundistight:above} (\ref{item-a-thmshearersboundistight}) holds on the given $G_B$ and $\vec{r}$. 
\end{example}

\subsection{Induced interaction graph and induced relative dimensions}\label{sec-Inducedinteractiongraphandrelativedimensions}
Our proof of \Cref{shearersboundistight:above} (\ref{item-a-thmshearersboundistight}) is by induction on the number of left vertices in the interaction graphs. 
We first define the induced interaction graph and induced relative dimensions.

For any two sets of left vertices $S, T\subseteq [m]$ in $G_B$, define $S \mid T$ in $G_B$ as
$$S\setminus (\{\text{left vertex } i: i\in T \text{ or } \mathcal{N}(G_B,i)\cap \mathcal{N}(G_B,j)\neq \emptyset   \text{ for some } j\in T\}).$$
Intuitively, $S\mid T$ is the set of left vertices in $S$ which share no neighbor with the left vertices in $T$.
For simplicity, we will omit ``in $G_B$'' if $G_B$ is clear from the context.
For each $i\in [m]$ and $S\subseteq [m]$, we will also simply denote the set $S\mid \{i\}$ with $S\mid i$.

\begin{definition}[Induced interaction graph and induced relative dimensions]\label{def-induce-graph}
Given $G_B = ([m],[n],E_B)$ and rational $\vec{r}=(r_1,\cdots,r_m)
	\in{(0,1]}^m$ where $I(G_B,\vec{r},S)>0$ for each $S\subsetneq [m]$
	and $\abs{\mathcal{N}(n)}\geq 2$ for the right vertex $n$, 
assume \emph{w.l.o.g.} $\mathcal{N}(n)=[t]$ where $2 \leq t \leq m$.
Also assume that for each left vertex $i\in [m]$, 
$\mathcal{N}(i) \neq \emptyset$ and $r_i = \frac{p_i}{q_i}<1$ where $p_i,q_i$ are 
mutually prime. 
Let $q = \prod_{i\in [m]}q_i$.

Given a left vertex $i\in [t]$,
let $Y_i = \{i\}\cup \left([m]\setminus [t]\right)$.
Let $G_i$ be the induced graph on the left vertices in $Y_i$
and the right vertices in $[n-1]$.
For each vertex $j$, we will use $\mathcal{N}(j)$ to denote the neighbor of $j$ in $G_B$ and use $\mathcal{N}(G_i,j)$ to denote the neighbor of $j$ in $G_i$ where we specify $G_i$ for clarity.
Similarly, we will also use $I(S)$ to denote 
$I(G_B,\vec{r},S)$ for each subset $S\subseteq [m]$
and use $I(G_i,\vec{\lambda}_i,S)$ to denote the independence polynomial of $G_i,\vec{\lambda}_i$ and $S\subseteq Y_i$ 
where we specify $G_i$ and $\vec{\lambda}_i$ for clarity.
When we use the notation $S \mid T$ for some sets $S,T$, we always means $S \mid T$ in $G_B$.

Let $\vec{\lambda}_i = (\lambda_{ij})_{j\in Y_i}$
where $\lambda_{ij} = r_j$ for each $j\neq i$ and
\begin{align}\label{eq-lambda-ii}
\rho_i = \frac{\sum_{\ell\in[t]} r_{\ell}\cdot I([m] \mid \ell)  }{I([m]\mid i)},\quad \lambda_{ii} = \min \left\{\rho_i,1\right\}.
\end{align}
Assume $\lambda_{ij} = \frac{p_{ij}}{q_{ij}}$ where $p_{ij},q_{ij}$ are mutually prime for each $j\in Y_i$.
\end{definition}
\begin{rem}
In above definition, by $I(G_B,\vec{r},S)>0$ for each $S\subsetneq [m]$,  
we have $\rho_i, \lambda_{ii}$ are properly defined and $0< \lambda_{ii} \leq 1$.
In addition, by $\lambda_{ij} = r_i$ for each $j\neq i$ and $r_i \in (0,1]$, we have $0< \lambda_{ij} \leq 1$ for each $j\neq i$.
Thus, $\vec{\lambda}_i = (\lambda_{ij})_{j\in Y_i}$ is a properly defined relative dimension vector.
We emphasize that if $\mathcal{N}(i) = \{n\}$ for the left vertex $i$,
then $i$ has no neighbor in $G_i$. 
\end{rem}

The following lemma is used in the induction proof of \Cref{shearersboundistight:above} (\ref{item-a-thmshearersboundistight}).
It shows that if $I([m])\leq 0$, then the induced interaction graph $G_i$ and the induced relative dimension vector $\vec{\lambda}_i$ also satisfy the condition in \Cref{shearersboundistight:above} (\ref{item-a-thmshearersboundistight}).

\begin{lemma}\label{lemma-induction}
Given $G_B$, $\vec{r}$, $i$, $G_i$ and $\vec{\lambda}_i$ as in \Cref{def-induce-graph},
assume $I([m])\leq 0$. Then $\vec{\lambda}_i \not \in \mathcal{I}(G_i)$.
Moreover, if $\mathcal{N}(G_i,j) = \emptyset$ for some left vertex $j\in Y_i$ in $G_i$, 
then $\lambda_{ij} = 1$.
\end{lemma}

The following lemma is used in the proof of \Cref{lemma-induction}, which is
immediate by \Cref{def-induce-graph} and \Cref{def: indpoly}.

\begin{lemma}\label{lemma-induction-idpoly}
Given $G_B$, $\vec{r}$, $i$, $Y_i$, $G_i$ and $\vec{\lambda}_i$ as in \Cref{def-induce-graph}
and a left vertex $j\in [t]$, we have
\begin{align*}
I(G_i,\vec{\lambda}_i,[m]\setminus [t]) = I([m]\setminus [t]), \quad I(G_i,\vec{\lambda}_i,[m]\mid j) = I([m]\mid j).
\end{align*}
Thus, we also have
\begin{align*}
\lambda_{ii} = \min \left\{\frac{\sum_{\ell\in[t]} \left(r_{\ell} I(G_i,\vec{\lambda}_i,[m] \mid \ell)\right)  }{I(G_i,\vec{\lambda}_i,[m]\mid i)},1\right\}.
\end{align*}
\end{lemma}

By the definition of independence polynomial, we have the following properties. 
\begin{prop}\label{le:property of ind by def}
	Given $G_B=([m],[n],E_B)$ and $\vec{r} \in {(0,1]}^m$, we have
	\begin{enumerate}[(a)]
		\item \label{item-a-le-propertyindbydef} for each $i\notin S \subsetneq [m]$, $I(S\cup \{i\})=I(S)-r_i\cdot I(S \mid i)$.
		\item \label{item-b-le-propertyindbydef} if $\mathcal{N}(n) = [t]$ for the right vertex $n$, then $I([m]) = I([m]\setminus [t])-\sum_{\ell=1}^t r_{\ell}\cdot I([m]\mid \ell)$.
		\item\label{item-c-le-propertyindbydef} if $I(S)	>0$ for each $S\subsetneq [m]$, then we have   $I(T)\leq 1$ for each $T\subseteq[m]$.
		\item \label{item-d-le-propertyindbydef} if $I(S)	>0$ for each $S\subsetneq [m]$ and $I([m])\leq 0$, then $G_D(G_B)$ is connected.
	\end{enumerate}
\end{prop}
\begin{proof}
By \Cref{def: indpoly},
(\ref{item-a-le-propertyindbydef}) and (\ref{item-b-le-propertyindbydef}) are immediate.
In the following, we prove (\ref{item-c-le-propertyindbydef}) and (\ref{item-d-le-propertyindbydef}).

    \emph{(c).} For each $i$ and $Y$ where 
	$i\notin Y \subsetneq [m]$, by $I(S)
	>0$ for each $S\subsetneq [m]$, we have $I(Y \mid i)>0$.
	Combining with $r_i>0$, we have $r_i I(Y \mid i)>0$.
	In addition, by (\ref{item-a-le-propertyindbydef}) of this proposition, we have 
	$I(Y\cup \{i\})=I(Y)-r_i I(Y \mid i)$.
	Combining with $r_i I(Y \mid i)>0$, we have $I(Y\cup \{i\})<I(Y)$.
	In other words, $I(Y)$ decreases as $Y$ grows. 
	Thus, for any $T \subseteq [m]$, we have
	 $I(T)\leq I(\emptyset)=1$.
	
	\emph{(d).} Assume that $G_D(G_B)$ is not connected for contradiction. 
	Then
	there exists $\emptyset\subsetneq T\subsetneq[m]$ such that $T$ and $[m]\setminus T$ are separate in $G_D(G_B)$.
	In addition, by \Cref{def: indpoly} we have $I([m])=I(T)I([m]\setminus T)\leq 0$. 
	Thus, either $I(T)\leq 0$ or $I([m]\setminus T)\leq 0$, a contradiction with $I(S)
	>0$ for each $S\subsetneq [m]$.
\end{proof}

Now we can prove \Cref{lemma-induction}.
\begin{proof}[Proof of \Cref{lemma-induction}]
At first, we prove that $\vec{\lambda}_i \not \in \mathcal{I}(G_i)$.
By $i \in [t] = \mathcal{N}(n)$ for the right vertex $n$,
we have $([m]\setminus [t]) \mid i = [m] \mid i$.
Thus,
\begin{align*}
    &I(G_i,\vec{\lambda}_i,Y_i) \\
    (\text{by $Y_i =\{i\}\cup \left([m]\setminus[t] \right)$})\quad = &I(G_i,\vec{\lambda}_i,\{i\}\cup \left([m]\setminus[t]\right)) \\
    (\text{by \Cref{le:property of ind by def} (\ref{item-a-le-propertyindbydef})})\quad= &I(G_i,\vec{\lambda}_i,[m]\backslash [t])-\lambda_{ii} I(G_i,\vec{\lambda}_i,([m]\setminus [t])\mid i)\\
    (\text{by \eqref{eq-lambda-ii}})\quad= &I(G_i,\vec{\lambda}_i,[m]\backslash [t])- \frac{\sum_{\ell\in [t]} \left(r_{\ell} I(G_i,\vec{\lambda}_i,[m] \mid \ell)\right)  }{I(G_i,\vec{\lambda}_i,[m]\mid i)}\cdot I(G_i,\vec{\lambda}_i,([m]\setminus [t])\mid i)\\
    (\text{by \Cref{lemma-induction-idpoly}})\quad= &I([m]\backslash [t])-\frac{\sum_{\ell\in [t]} \left(r_{\ell} I([m] \mid \ell)\right)  }{I([m]\mid i)}\cdot I(([m]\setminus [t])\mid i)\\
    (\text{by $([m]\setminus [t]) \mid i = [m] \mid i$})\quad= &I([m]\backslash [t])-\frac{\sum_{\ell\in [t]} \left(r_{\ell} I([m] \mid \ell)\right)  }{I([m]\mid i)}\cdot I([m]\mid i)\\
    = &I([m]\backslash [t])-\sum_{\ell\in [t]}\left(r_{\ell} I([m] \mid \ell)\right)\\
    (\text{by \Cref{le:property of ind by def} (\ref{item-b-le-propertyindbydef})})\quad=&I([m])\\
    \leq&0.
\end{align*}
Combining with \Cref{lemma-eqinshearersbound},
we have $\vec{\lambda}_i \not\in \mathcal{I}(G_i)$.

In the following, we show that if $\mathcal{N}(G_i,j) = \emptyset$ for some left vertex $j\in Y_i$
in $G_i$, then $\lambda_{ij} = 1$.
It is lossless to assume that $j\in [t]$.
Because if $j\in [m]\setminus [t]$, we have 
$\mathcal{N}(G_i,j) = \mathcal{N}(G_B,j)$.
Combining with the assumption $\mathcal{N}(G_B,j) \neq \emptyset$ in \Cref{def-induce-graph},
we also have $\mathcal{N}(G_i,j) \neq \emptyset$.
In  the following, we assume $j\in [t]$.
Combining with $j\in Y_i$,
we have $j= i$.
In addition, if $\mathcal{N}(G_i,i) = \emptyset$ for the left vertex $i$,
we have $\mathcal{N}(G_B,i) = \{n\}$.
Therefore, we have $[m] \mid i = [m]\setminus [t]$ in $G_B$.
Thus, we have 
\begin{align*}
  I([m]\mid i) - \sum_{\ell\in [t]} \left(r_{\ell} I([m] \mid \ell)\right)  = I([m]\backslash [t])-\sum_{\ell\in [t]}\left(r_{\ell} I([m] \mid \ell)\right)=&I([m])\leq0,
\end{align*}
where the last equality is by \Cref{le:property of ind by def} (\ref{item-b-le-propertyindbydef}).
Combining with \eqref{eq-lambda-ii} and $\lambda_{i,i}>0$, we have $\lambda_{i,i} = 1$.
In summary, for each left vertex $j\in Y_i$, 
if $\mathcal{N}(G_i,j) = \emptyset$, then $\lambda_{ij} = 1$.
The lemma is proved.
\end{proof}

\subsection{Induced dimensions of qudits}\label{sec-induceddimensionsof qudits}
In this subsection, we define the induced dimensions of qudits.
Roughly speaking, 
for the induced interaction graph $G_i$ and the induced relative dimension vector $\vec{\lambda}_i$, 
we pick some induced dimension vector $(d_{ij})_{j\in Y_i}$ such that the random instance of the setting $(G_i,\vec{\lambda}_i,(d_{ij})_{j\in Y_i})$
spans the whole space.
The existence of such $(d_{ij})_{j\in Y_i}$ is ensured by applying \Cref{shearersboundistight:above} (\ref{item-a-thmshearersboundistight}) to $G_i$ and $\vec{\lambda}_i$.
The dimension of qudit $d_j$ for $G_B$ and $\vec{r}$ is defined as the product of the induced dimension of corresponding qudit $d_{ij}$ if $j\in [n-1]$, and 
defined by the independence polynomial of $G_B,\vec{r}$ if $j = n$.

\begin{definition}[Dimensions and induced dimensions of qudits]\label{def-induce-dimension}
Given $G_B$, $\vec{r}$, $i$, $Y_i$, $G_i$, $\rho_i$ and $\vec{\lambda}_i$ as in \Cref{def-induce-graph},
let $(a_{ij})_{j\in Y_i}$ be some dimension vector 
such that given any $\mathcal{H}_{1},\cdots,\mathcal{H}_{n-1}$ with $\dim(\mathcal{H}_{1},\cdots,\mathcal{H}_{n-1}) = (a_{ij})_{j\in Y_i}$,
the random subspaces in $\mathcal{H}_{[n-1]}$ of the setting $(G_i,\vec{\lambda}_i,(a_{ij})_{j\in Y_i})$
span the whole space $\mathcal{H}_{[n-1]}$.
Define $(d_{ij})_{j\in Y_i}$ as follows.
If $\rho_i\geq 1$,
let $d_{ij} = q_j$ for each $j\in Y_i$.
Otherwise, let 
$(d_{ij})_{j\in Y_i}$ be $(a_{ij})_{j\in Y_i}$.

Define
\begin{align}\label{eq-def-dj}
d_j = \begin{cases}
\prod_{i\in [t]}d_{ij} \quad & \text{if } j\in [n-1]\\
q^2\sum_{\ell=1}^t \left(r_{\ell} I([m] \mid \ell)\right) & \text{if } j= n.
\end{cases}
\end{align}
\end{definition}
\begin{rem}
We will specify $(a_{ij})_{j\in Y_i}$ for each $i\in[t]$ concretely 
in \Cref{sec-Proof-of-Theorem-shearersboundistight}.
\end{rem}

The following lemma shows that $(d_1,\cdots,d_n)$ is a properly defined dimension vector of qudits and provides an upper bound for the dimension of $d_n$.
\begin{lemma}\label{lemma-d-positive-integer}
Given $d_1,\cdots,d_n$ as in \Cref{def-induce-dimension}, we have $d_1,\cdots,d_n$ are positive integers. 
In addition, we have
$d_n\leq 2q^2$.
\end{lemma}

The following lemma is used in the proof of \Cref{lemma-d-positive-integer}.

\begin{lemma}\label{lem-lower-bound-alphai}
Given $G_B$, $\vec{r}$ and $q$ as in \Cref{def-induce-graph}, for each $S\subsetneq [m]$ and $\ell\in [m]$ we have
\begin{enumerate}[(a)]
    \item \label{lem-lower-bound-alphai-a} $qr_{\ell}$ and $q I(S)$ are positive integers;
    \item \label{lem-lower-bound-alphai-b} $qr_{\ell} I([m] \mid \ell),r_{\ell}q^2 I([m] \mid \ell), q\sum_{j\in S} \left(r_{j} I([m] \mid j)\right),q^2\sum_{j\in S} \left(r_{j} I([m] \mid j)\right)$ are positive integers;
    \item \label{lem-lower-bound-alphai-c} for each $k\in S$, 
there exists some integer $0<\alpha\leq q$ such that
\begin{align}\label{eq-alpha-sumlint}
\alpha(I([m]\mid k))^{-1}\sum_{j\in S} \left(r_{j} I([m] \mid j)\right)
\end{align}
is a positive integer.
\end{enumerate}
\end{lemma}
\begin{proof}
\emph{(a).} By $q>0$ and $r_{\ell}>0$, we have $qr_{\ell}$ is positive.
In addition, by 
$q = \prod_{j\in [m]}q_j$ and $r_{\ell} = \frac{p_{\ell}}{q_{\ell}}$,
we have 
$$qr_{\ell} = p_{\ell}\prod_{j\in [m]\setminus\{\ell\}}q_j.$$
Thus we have $qr_{\ell}$ is an integer by $p_{\ell}$ and $q_j$ for all $j\in [m]$ are integers.
In summary, $qr_{\ell}$ is a positive integer.

For each $S\subsetneq [m]$,
by \Cref{def: indpoly} we have 
\begin{align*}
q I(S) = I(S)\prod_{j\in S }q_j\prod_{k\in [m]
\setminus S}q_k = \sum_{T\in \mathsf{Ind}(G_B)\land T\subseteq S}(-1)^{\abs{T}}\prod_{j\in T}p_j\prod_{\ell\in S\setminus T}q_k\prod_{k\in [m]
\setminus S}q_{\ell}
\end{align*}
is an integer.
In addition, recall that $I(S)>0,q>0$ as in \Cref{def-induce-graph},
we have $qI(S)$ is positive.
Thus, $q I(S) $ is a positive integer.

\emph{(b).} We only show that $qr_{\ell} I([m] \mid \ell)$ is a positive integer for each $\ell \in [m]$.
Then it is immediate that $r_{\ell}q^2 I([m] \mid \ell), q\sum_{j\in S} \left(r_{j} I([m] \mid j)\right)$ and $q^2\sum_{j\in S} \left(r_{j} I([m] \mid j)\right)$ are also positive integers.
Recall that $I(T)>0$ for each $T\subsetneq [m]$ as in \Cref{def-induce-graph},
we have $I([m]\mid \ell)> 0$.
Combining with $r_{\ell}>0$ and $q>0$,
we have $qr_{\ell} I([m] \mid \ell)$ is positive.
In addition, we also have $qr_{\ell} I([m] \mid \ell)$ is an integer.
Let $T \triangleq \{S\mid \left(S\in \mathsf{Ind}(G_B)\right)\land \left(S\subseteq ([m] \mid \ell)\right)\}$. 
By \Cref{def: indpoly}, we have
\begin{align*}
  qr_{\ell} I([m] \mid \ell) &= qr_{\ell}\sum_{S\in T}(-1)^{\abs{S}}\prod_{j\in S}r_j =
    \left(\prod_{i\in [m]}q_i\right)\cdot\frac{p_{\ell}}{q_{\ell}}\cdot\sum_{S\in T}(-1)^{\abs{S}}\prod_{j\in S} \frac{p_{j}}{q_{j}}
    \\&= p_{\ell}\sum_{S\in T}(-1)^{\abs{S}}\prod_{j\in S}p_j\prod_{k\in [m]\setminus (S\cup\{\ell\})}q_k,
\end{align*}
where the last equality is by $\ell\not\in \left([m] \mid \ell\right)$.
Combining with $p_{j},q_j$ are integers for each $i\in [m]$,
we have $qr_{\ell} I([m] \mid \ell)$ is an integer.
In summary, we have $qr_{\ell} I([m] \mid \ell)$ is a positive integer.

\emph{(c).} Let
$\alpha = qI([m]\mid k)$.
Combining with $q I(S)$ is a positive integer for each $S\subsetneq  [m]$ as shown in (\ref{lem-lower-bound-alphai-a}) of this lemma,
we have $\alpha$ is a positive integer.
In addition, by \eqref{eq-alpha-sumlint} we have
\begin{align*}
    &\alpha(I([m]\mid k))^{-1}\sum_{j\in S}\left(r_{j} I([m] \mid j)\right) = \sum_{j\in S} \left(qr_{j} I([m] \mid j)\right).
\end{align*}
Combining with that $qr_{j} I([m] \mid j)$ is a positive integer for each $j \in [m]$ as shown in (\ref{lem-lower-bound-alphai-b}) of this lemma,
we have \eqref{eq-alpha-sumlint} is a positive integer.
\end{proof}

By \Cref{lem-lower-bound-alphai}, we have the following corollary immediately.

\begin{corol}\label{corol-induction-prodq}
$\prod_{j\in Y_{i}} q_{ij} \leq q^2$ for each $i\in [t]$.
\end{corol}
\begin{proof}
By $\lambda_{ij} = r_i$ for each $j\neq i$, we also have 
$q_{ij} = q_j$ for each $j\neq i$.
In addition, 
by combining \eqref{eq-lambda-ii} with \Cref{lem-lower-bound-alphai} (\ref{lem-lower-bound-alphai-c}), 
we have 
there exists some integer $0<\alpha\leq \prod_{j\in [m]}q_j$ such that
$\alpha\lambda_{ii}$
is a positive integer.
Thus, we have $q_{ii}\leq \alpha \leq \prod_{j\in [m]}q_j = q$.
Combining with $q_{ij} = q_j$ for each $j\neq i$,
we have 
\[
\prod_{j\in Y_{i}} q_{ij} \leq q_{ii}\cdot \prod_{j\neq i} q_{ij} \leq q\cdot \prod_{j\neq i} q_{j} \leq  q^2. \qedhere
\]
\end{proof}

Now we can prove \Cref{lemma-d-positive-integer}.

\begin{proof}[Proof of \Cref{lemma-d-positive-integer}.]
At first, we prove that $d_1,\cdots,d_n$ are positive integers.
For each $j\in [n-1]$, $d_j$
is a positive integer is immediate by $t\geq 2$, 
$d_{i,j}$ is a positive integer for each $i\in [t]$ and \eqref{eq-def-dj}.

In the next, we prove that $d_n$ is a positive integer no more than $2q$.
Combining \eqref{eq-def-dj} with \Cref{lem-lower-bound-alphai} (\ref{lem-lower-bound-alphai-b}),
we have
$d_n$ is a positive integer.
Moreover, we have 
\begin{align*} &\sum_{\ell=1}^tr_{\ell}\cdot I([m]\mid \ell)\\
(\text{by \Cref{le:property of ind by def} (\ref{item-b-le-propertyindbydef})})\quad=&I([m]\setminus [t])-I([m])\\      (\text{by \Cref{le:property of ind by def} (\ref{item-a-le-propertyindbydef})}) \quad =&I([m]\setminus [t])-\left(I([m]\backslash\{1\})-r_1\cdot I([m]\mid\{1\})\right)\\
(\text{by $I([m]\backslash\{1\})>0$ and $r_1\leq 1$}) \quad <& I([m]\setminus [t])+ I([m]\mid\{1\})\\
(\text{by \Cref{le:property of ind by def} (\ref{item-c-le-propertyindbydef})}) \quad\leq& 2.
\end{align*}
Thus by \eqref{eq-def-dj} we have
$d_n \leq 2q^2$.
\end{proof}

Given $G_B$, $\vec{r}$ as in \Cref{def-induce-graph}, 
for each $i\in [t]$, we have $r_iq^2  I([m]\mid i)$ is a positive integer by \Cref{lem-lower-bound-alphai} (\ref{lem-lower-bound-alphai-b}). 
Thus, one can decompose a qudit with dimension $d_n$ into $t$ orthogonal subspaces as follows.

\begin{definition}[Induced qudits]\label{def-induce-qudits}
Given $d_1,\cdots,d_n$ as in \Cref{def-induce-dimension},
Let $\mathcal{H}_1,\mathcal{H}_2,\cdots,\mathcal{H}_n$ be qudits where $\ddim(\mathcal{H}_j) = d_j$ for each $j\in [n]$.
Decompose $\mathcal{H}_n$ into $t$ orthogonal subspaces $\mathcal{H}_n^1,\cdots,\mathcal{H}_n^t$ arbitrarily where
\begin{equation}\ \label{eq:dim}
\forall i\in [t], \quad \ddim(\mathcal{H}_n^i) = r_{i} q^2 I([m]\mid i).
\end{equation}
Define
\begin{align*}
T_1 &\triangleq \{\text{left vertex } i\mid (i\in [t])\land \left(r_i\ddim(\mathcal{H}_n) < \ddim(\mathcal{H}_n^i)\right)\land \left(\mathcal{N}(i) \neq \{n\}\right)\},\\
T_2 &\triangleq \{\text{left vertex } i\mid (i\in [t])\land \left(r_i\ddim(\mathcal{H}_n) \geq \ddim(\mathcal{H}_n^i)\right)\}.
\end{align*}
\end{definition}

By \eqref{eq-def-dj} and \eqref{eq:dim} we have $\ddim(\mathcal{H}_n)=d_n=\sum_{\ell=1}^t\ddim(\mathcal{H}_n^\ell)$.
Thus the decomposition in \Cref{def-induce-qudits} is properly defined.
In addition, the following lemma shows that $T_1$, $T_2$ form a partition of the set of left vertices $[t]$.
\begin{lemma}\label{lemma-t1t2t}
$T_1\cap T_2 = \emptyset$ and $T_1\cup T_2 = [t]$.
\end{lemma}
\begin{proof}
$T_1\cap T_2 = \emptyset$ is immediate by their definitions.
In the following, we show that $T_1\cup T_2 = [t]$.
Then the lemma is immediate.
By \Cref{le:property of ind by def} (\ref{item-b-le-propertyindbydef}) we have  
$$I([m]) = I([m]\setminus [t])-\sum_{\ell=1}^t r_{\ell}\cdot I([m]\mid \ell).$$
If $I([m])\leq 0$,
we have 
$$I([m]\setminus [t])\leq\sum_{\ell=1}^t r_{\ell}\cdot I([m]\mid \ell).$$
For each left vertex $i$ where $\mathcal{N}(i) = \{n\}$,
we have 
$$I([m]\mid i) = I([m]\setminus [t])  \leq \sum_{\ell=1}^t r_{\ell}\cdot I([m]\mid \ell).$$
Thus we have 
$$r_{i} q^2 I([m]\mid i)\leq r_iq^2\sum_{\ell=1}^t \left(r_{\ell} I([m] \mid \ell)\right).$$
Combining with \eqref{eq:dim} and \eqref{eq-def-dj}, we have
$$\ddim(\mathcal{H}_n^i) = r_{i} q^2 I([m]\mid i)
    \leq r_iq^2\sum_{\ell=1}^t \left(r_{\ell} I([m] \mid \ell)\right) 
    =r_id_n = r_i\ddim(\mathcal{H}_n).$$
In summary, if $\mathcal{N}(i) = \{n\}$,
we have $\ddim(\mathcal{H}_n^i) \leq r_i\ddim(\mathcal{H}_n)$.
Thus we have 
$\{i\in [t]\mid \mathcal{N}(i) = \{n\}\} \subseteq T_2$.
Thus, we have
\begin{align*}
T_2 &= \{\text{left vertex } i\mid (i\in [t])\land \left(r_i\ddim(\mathcal{H}_n) \geq \ddim(\mathcal{H}_n^i)\right)\}\\
&=\{\text{left vertex } i\mid (i\in [t])\land \left(\left(r_i\ddim(\mathcal{H}_n) \geq \ddim(\mathcal{H}_n^i)\right)\lor \left(\mathcal{N}(i) = \{n\}\right)\right)\}.
\end{align*}
Combining the definition of $T_1$, we have $T_1\cup T_2= [t]$.
\end{proof}

\subsection{Induced subspaces}\label{sec-inducedsubspaces}
In this subsection, we define the induced subspaces.
Given $\mathcal{H}_1,\mathcal{H}_2,\cdots,\mathcal{H}_n$ and $\mathcal{H}_n^1,\cdots,\mathcal{H}_n^t$ as in \Cref{def-induce-qudits},
we define the induced random subspaces 
$V'_1,\cdots,V'_m$ and  
the induced subspaces $V^{\star}_1,\cdots,V^{\star}_m$ as follows.
The intuition is to let $V'_1,\cdots,V'_m$ and $V^{\star}_1,\cdots,V^{\star}_m$
be some random instance and instance of the setting $(G_i,\vec{\lambda}_i,(d_1,\cdots,d_{n-1}))$,
respectively.

\begin{definition}[Induced subspaces]\label{def-induce-subspace}
Given $\mathcal{H}_1,\mathcal{H}_2,\cdots,\mathcal{H}_n$ and $\mathcal{H}_n^1,\cdots,\mathcal{H}_n^t$ as in \Cref{def-induce-qudits}, 
we choose subspaces $V'_1,\cdots,V'_m$ of $\mathcal{H}_{[n-1]}$ randomly for each $i\in [m]$ as follows.
\begin{itemize}
    \item For each $i\in T_1$,
let $V'_i$ be $V_i^{\ast}\otimes \mathcal{H}_{[n-1]\setminus \mathcal{N}(i)}$ where $V_i^{\ast}$ is a random subspace of $\mathcal{H}_{\mathcal{N}(i)\backslash \{n\}}$ with 
$$\mathbb{R}\left(V_i^{\ast},\mathcal{H}_{\mathcal{N}(i)\backslash \{n\}}\right)= \frac{r_i\ddim(\mathcal{H}_n)}{\ddim(\mathcal{H}_n^i)}.$$
\item For each $i\in T_2$, let $V'_i$ be $\mathcal{H}_{[n-1]}$.
\item For each $i\in [m]\setminus [t]$, let $V'_i$ be $V_i^{\ast}\otimes\mathcal{H}_{[n-1]\setminus \mathcal{N}(i)}$ where
$V_i^{\ast}$ is a random subspace of $\mathcal{H}_{\mathcal{N}(i)}$  with $\mathbb{R}(V_i^{\ast},\mathcal{H}_{\mathcal{N}(i)}) = r_i$.
\end{itemize}

Similarly, let $V^{\star}_1,\cdots,V^{\star}_m$ be some subspaces of $\mathcal{H}_{[n-1]}$ such that $\mathbb{R}(V^{\star}_i,\mathcal{H}_{[n-1]}) =\mathbb{R}(V'_i,\mathcal{H}_{[n-1]})$ for each $i\in [m]$ and
\begin{itemize}
    \item For each $i\in T_1$, $V^{\star}_i$ acts trivially on $\mathcal{H}_{[n-1]\setminus \mathcal{N}(i)}$.
    \item For each $i\in T_2$, $V^{\star}_i = \mathcal{H}_{[n-1]}$.
    \item For each $i\in [m]\setminus [t]$, $V^{\star}_i$ acts trivially on $\mathcal{H}_{[n-1]\setminus \mathcal{N}(i)}$.
\end{itemize}
\end{definition}
\begin{rem}
We will specify $V^{\star}_1,\cdots,V^{\star}_m$ concretely 
in \Cref{sec-Proof-of-Theorem-shearersboundistight}.
\end{rem}

The following lemmas show that  $V'_1,\cdots,V'_m$ are properly defined.
Similarly, one can also verify that $V^{\star}_1,\cdots,V^{\star}_m$ are also properly defined.

\begin{lemma}\label{lemma-viuj-integers}
For each $i\in T_1\cup \left([m]\setminus [t]\right)$, 
we have $\ddim(V_i^{\ast})$ is a positive integer.
Thus, $\ddim(V'_i)$ is also a positive integer for each $i\in [m]$.
\end{lemma}

The following two lemmas are used in the proof of \Cref{lemma-viuj-integers}.
\Cref{lemma-ratio-rhoi} is a simple fact by 
\eqref{eq-lambda-ii}, \eqref{eq-def-dj} and \eqref{eq:dim},
and \Cref{lemma-random-instance-dsize} is an easy observation by \Cref{def-instance}.
\begin{lemma}\label{lemma-ratio-rhoi}
For each $i\in [t]$, we have $$\frac{r_i\ddim(\mathcal{H}_n)}{\ddim(\mathcal{H}_n^i)} = \rho_i.$$
\end{lemma}

\begin{lemma}\label{lemma-random-instance-dsize}
Given $i\in [t]$, if $\Pr[\mathbb{R}(\bigplus_{j\in Y_i} V_{ij})=1]=1$ for random subspaces $(V_{ij})_{j\in Y_i}$ of the setting $(G_i,\vec{\lambda}_i,(d_{ij})_{j\in Y_i})$,
then for each $j\in Y_i$, 
$\lambda_{ij}\prod_{k\in \mathcal{N}(G_i,j)}d_{ik}$ is a positive integer.
\end{lemma}

\begin{proof}[Proof of \Cref{lemma-viuj-integers}]
We only prove that $\ddim(V_i^{\ast})$ is a positive integer for each $i\in T_1\cup \left([m]\setminus [t]\right)$.
Then it is immediate that $\ddim(V'_i)$ is also a positive integer for each $i\in [m]$ by its definition.
For each $i\in T_1$, 
by the definition of $T_1$ in \Cref{def-induce-qudits}, we have 
$r_i\ddim(\mathcal{H}_n)< \ddim(\mathcal{H}_n^i)$.
Combining with \Cref{lemma-ratio-rhoi}, we have 
\begin{align}\label{eq-rhoi-le1}
\rho_i = \frac{r_i\ddim(\mathcal{H}_n)}{\ddim(\mathcal{H}_n^i)}<1.
\end{align}
Combining with \Cref{def-induce-dimension}, we have
$\Pr[\mathbb{R}(\bigplus_{j\in Y_i} V_{ij})=1]=1$ for random subspaces $(V_{ij})_{j\in Y_i}$ of the setting $(G_i,\vec{\lambda}_i,(d_{ij})_{j\in Y_i})$.
Combining with \Cref{lemma-random-instance-dsize},
we have $\lambda_{ii}\prod_{j\in \mathcal{N}(G_i,i)}d_{ij}$ is a positive integer.
In addition, by $j\in \mathcal{N}(G_i,i)$ and the definition of $G_i$, we have $j\in [n-1]$.
Combining with \Cref{def-induce-dimension} and $\lambda_{ii}\prod_{j\in \mathcal{N}(G_i,i)}d_{ij}$ is a positive integer,
we have $\lambda_{ii}\prod_{j\in \mathcal{N}(G_i,i)}d_{j}$ is also a positive integer. 
Moreover, by \eqref{eq-rhoi-le1} and \eqref{eq-lambda-ii}, we have
\begin{align}\label{eq-lambdaii-pi}
\lambda_{ii} = \rho_i = \frac{r_i\ddim(\mathcal{H}_n)}{\ddim(\mathcal{H}_n^i)}.
\end{align}
Combining with \Cref{def-induce-subspace},
we have 
$\lambda_{ii} = \mathbb{R}\left(V_i^{\ast},\mathcal{H}_{\mathcal{N}(i)\backslash \{n\}}\right)$.
By $i\in T_1$ and the definition of $G_i$ in \Cref{def-induce-graph}, we have $\mathcal{N}(G_i,i) = \mathcal{N}(i)\setminus \{n\}$.
Thus, we have 
\begin{align*}
&\lambda_{ii}\prod_{j\in \mathcal{N}(G_i,i)}d_{j}  = \mathbb{R}\left(V_i^{\ast},\mathcal{H}_{\mathcal{N}(i)\backslash \{n\}}\right)\prod_{j\in \mathcal{N}(i)\setminus \{n\}}d_{j}
= \mathbb{R}\left(V_i^{\ast},\mathcal{H}_{\mathcal{N}(i)\backslash \{n\}}\right)\ddim\left(\mathcal{H}_{\mathcal{N}(i)\backslash \{n\}}\right)
= \dim\left(V_i^{\ast}\right),
\end{align*}
where the second equality is by \Cref{def-induce-qudits}.
Combining with that $\lambda_{ii}\prod_{j\in \mathcal{N}(G_i,i)}d_{j}$ is a positive integer,
we have $ \dim\left(V_i^{\ast}\right)$ 
is also a positive integer. The lemma is proved.
\end{proof}

The following two lemmas show that $V'_1,\cdots, V'_m$ are an instance of the setting $(G_i,\vec{\lambda}_i,(d_1,\cdots,d_{n-1}))$.

\begin{lemma}\label{lemma-induce-gi}
For each left vertex $i\in [t]$,
$G_i$ is an interaction graph of $(V'_{j})_{j\in Y_i}$.
\end{lemma}

\begin{proof}
Given a left vertex $i\in [t]$, by \Cref{lemma-t1t2t},
we have either $i\in T_1$ or $i\in T_2$.
If $i\in T_1$, we have $\mathcal{N}(G_i,i) = \mathcal{N}(i)\setminus \{n\}$
and 
$$V_i^{'} = V_i^{\ast}\otimes \mathcal{H}_{[n]\setminus \mathcal{N}(i)} = V_i^{\ast}\otimes \mathcal{H}_{[n-1]\setminus \mathcal{N}(G_i,i)}.$$
If $i\in T_2$, we have $\mathcal{N}(G_i,i) = \mathcal{N}(i)\setminus \{n\}$
and $V_i^{'} = \mathcal{H}_{[n-1]}$.
Moreover, for each left vertex $i\in [t]$ and $j\in Y_i\setminus \{i\}$, we have $j\in [m]\setminus [t]$.
Combining with $[t] = \mathcal{N}(n)$ for the right vertex $n$,
we have $j\not\in \mathcal{N}(n)$.
Thus, $n\not \in \mathcal{N}(j)$.
Combining with \Cref{def-induce-graph},
we have $\mathcal{N}(G_i,j) = \mathcal{N}(j)\setminus \{n\} = \mathcal{N}(j)$.
Combining with \Cref{def-induce-subspace}, we have
$$V'_j = V_j^{\ast}\otimes\mathcal{H}_{[n-1]\setminus \mathcal{N}(j)} = V_j^{\ast}\otimes\mathcal{H}_{[n-1]\setminus \mathcal{N}(G_i,j)}.$$
In summary, we always have that 
$\Pi_{V'_j}$ acts trivially on the qudits $\mathcal{H}_{[n-1]\setminus \mathcal{N}(G_i,j)}$ for each $i\in [t]$ and $j\in Y_i$.
Then $G_i$ is an interaction graph of $(V'_{j})_{j\in Y_i}$.     
\end{proof}

\begin{lemma}\label{lemma-induce-rviprime}
For each left vertex $i\in [t]$ and $j\in Y_i$, we have 
\begin{enumerate}[(a)]
    \item $0<\lambda_{ij}\leq 1$;\label{item-lemma-induce-rviprime-lambdaij}
    \item $V_j^{'} = \mathcal{H}_{[n-1]}$ if $\lambda_{ij}=1$;\label{item-lemma-induce-rviprime-Vjprime}
    \item $V_j^{'} = V_j^{\ast}\otimes \mathcal{H}_{[n-1]\setminus \mathcal{N}(G_i,j)}$ where $V_j^{\ast}$ is a random subspace of $\mathcal{H}_{\mathcal{N}(G_i,j)}$ according to the Haar measure with $\mathbb{R}(V_j^{\ast},\mathcal{H}_{\mathcal{N}(G_i,j)}) = \lambda_{ij}$ if $\lambda_{ij}<1$.\label{item-lemma-induce-rviprime-Vjstar}
\end{enumerate}
\end{lemma}
\begin{proof}
By \eqref{eq-lambda-ii}, we have $0<\lambda_{ij}\leq 1$ immediately.
Recall that $Y_i = \{i\}\cup \left([m]\setminus [t]\right)$.
Combining with $j\in Y_i$, we have either $j\in [m]\setminus [t]$ or $j= i$.
In the following, we prove (\ref{item-lemma-induce-rviprime-Vjprime}) and (\ref{item-lemma-induce-rviprime-Vjstar}) of this lemma for the two cases $j\in [m]\setminus [t]$ and $j= i$ separately. Then the lemma is immediate.

At first, we assume $j\in [m]\setminus [t]$.
We have $n\not\in \mathcal{N}(j)$ and then 
$\mathcal{N}(j) = \mathcal{N}(G_i,j)$.
Thus, by $j\in [m]\setminus [t]$ and \Cref{def-induce-subspace}, we have 
$$V'_j = V_j^{\ast}\otimes\mathcal{H}_{[n-1]\setminus \mathcal{N}(j)} = V_j^{\ast}\otimes\mathcal{H}_{[n-1]\setminus \mathcal{N}(G_i,j)}$$ where 
$V_j^{\ast}$ is a random subspace of $\mathcal{H}_{\mathcal{N}(j)} = \mathcal{H}_{\mathcal{N}(G_i,j)}$ according to the Haar measure with 
\begin{align}\label{eq-rviast-hngij}
\mathbb{R}(V_j^{\ast}, \mathcal{H}_{\mathcal{N}(G_i,j)}) = \mathbb{R}(V_j^{\ast}, \mathcal{H}_{\mathcal{N}(j)}) = r_j.
\end{align}
Moreover, by $j\in [m]\setminus [t]$ and $i\in [t]$,
we have $j\neq i$.
Combining with \Cref{def-induce-graph},
we have $\lambda_{ij} = r_j $.
Combining with \eqref{eq-rviast-hngij},
we have 
$\mathbb{R}(V_j^{\ast}, \mathcal{H}_{\mathcal{N}(G_i,j)}) = \lambda_{ij}$.
Then (\ref{item-lemma-induce-rviprime-Vjstar}) of this lemma is proved.
Meanwhile, 
if $r_j = 1$, by \eqref{eq-rviast-hngij} we have 
$\mathbb{R}(V_j^{\ast}, \mathcal{H}_{\mathcal{N}(j)}) = 1$ and then
$V_i^{\ast} = \mathcal{H}_{\mathcal{N}(j)}$.
Combining with \Cref{def-induce-subspace},
we have 
$V'_j =V_j^{\ast}\otimes\mathcal{H}_{[n-1]\setminus \mathcal{N}(j)} = \mathcal{H}_{[n-1]}$.
Then (\ref{item-lemma-induce-rviprime-Vjprime}) of this lemma is proved.
Thus the case $j\in [m]\setminus [t]$ is proved.

In the following, we assume $j = i$.
We have $n\in \mathcal{N}(i) = \mathcal{N}(j)$ and then 
\begin{align}\label{eq-mathcalngij-mathcalni}
\mathcal{N}(G_i,j) = \mathcal{N}(j) \setminus \{n\}.
\end{align}
If $\lambda_{ij} = 1$, we have $\lambda_{ii} = 1$.
Moreover, we claim that $\lambda_{ii} = 1$ if and only if $i\in T_2$.
Thus by $\lambda_{ii} = 1$, we have 
$i\in T_2$.
Combining with \Cref{def-induce-subspace}, we have 
$V_i^{'} = \mathcal{H}_{[n-1]}$.
Combining with $j = i$, we have 
$V_j^{'} = \mathcal{H}_{[n-1]}$.
Then (\ref{item-lemma-induce-rviprime-Vjprime}) of this lemma is proved.
If $\lambda_{ij} < 1$, we have $\lambda_{ii} = \lambda_{ij} < 1$.
Combining with the claim that $\lambda_{ii} = 1$ if and only if $i\in T_2$,
we have $i\not\in T_2$.
Combining with $i\in [t]$ and \Cref{lemma-t1t2t},
we have $i\in T_1$.
Combining with \Cref{def-induce-subspace} and $j=i$, we have
$V'_j = V_j^{\ast}\otimes\mathcal{H}_{[n-1]\setminus \mathcal{N}(j)}$ where 
$V_j^{\ast}$ is a random subspace of $\mathcal{H}_{\mathcal{N}(j)\setminus \{n\}}$
with 
$$\mathbb{R}\left(V_j^{\ast},\mathcal{H}_{\mathcal{N}(j)\backslash \{n\}}\right)= \frac{r_i\ddim(\mathcal{H}_n)}{\ddim(\mathcal{H}_n^i)}.$$
Combining with \eqref{eq-mathcalngij-mathcalni},
we have 
$$V'_j = V_j^{\ast}\otimes\mathcal{H}_{[n-1]\setminus \mathcal{N}(j)} =  V_j^{\ast}\otimes\mathcal{H}_{[n-1]\setminus \mathcal{N}(G_i,j)}$$
where 
$V_j^{\ast}$ is a random subspace of $\mathcal{H}_{\mathcal{N}(G_i,j)}$
with 
\begin{align}\label{eq-definitionrvstar-inproof}
\mathbb{R}\left(V_j^{\ast},\mathcal{H}_{\mathcal{N}(G_i,j)}\right) = \mathbb{R}\left(V_j^{\ast},\mathcal{H}_{\mathcal{N}(j)\backslash \{n\}}\right)= \frac{r_i\ddim(\mathcal{H}_n)}{\ddim(\mathcal{H}_n^i)}.
\end{align}
In addition, by \eqref{eq-def-dj} and $\ddim(\mathcal{H}_n) = d_n$,
we have 
\begin{align}\label{eq-riddimhn}
r_i\ddim(\mathcal{H}_n) = r_id_n = r_iq^2\sum_{\ell=1}^t \left(r_{\ell} I([m] \mid \ell)\right).
\end{align}
Combining with \eqref{eq-definitionrvstar-inproof}, \eqref{eq:dim} and \eqref{eq-riddimhn}, we have 
\begin{equation*}
 \begin{aligned}
\mathbb{R}\left(V_j^{\ast},\mathcal{H}_{\mathcal{N}(G_i,j)}\right)= \frac{r_i\ddim(\mathcal{H}_n)}{\ddim(\mathcal{H}_n^i)} = \frac{r_iq^2\sum_{\ell\in[t]} \left(r_{\ell} I([m] \mid \ell)\right)  }{r_iq^2I([m]\mid i)}=\frac{\sum_{\ell\in[t]} \left(r_{\ell} I([m] \mid \ell)\right)  }{I([m]\mid i)}.  
\end{aligned}  
\end{equation*}
Combining with $\lambda_{ii} = \lambda_{ij} < 1$ and \eqref{eq-lambda-ii}, we have
\begin{equation*}
 \begin{aligned}
\mathbb{R}\left(V_j^{\ast},\mathcal{H}_{\mathcal{N}(G_i,j)}\right)= \frac{\sum_{\ell\in[t]} \left(r_{\ell} I([m] \mid \ell)\right)  }{I([m]\mid i)} = \lambda_{ii} = \lambda_{ij}.  
\end{aligned}  
\end{equation*}
Then (\ref{item-lemma-induce-rviprime-Vjstar}) of this lemma is proved.

At last, we prove the claim that $\lambda_{ii} = 1$ if and only if $i\in T_2$.
Then the case $j=i$ is proved and the lemma is immediate.
If $\lambda_{ii} = 1$,
by \eqref{eq-lambda-ii} we have
\begin{align}\label{eq-immidi}
I([m]\mid i) \leq \sum_{\ell=1}^t \left(r_{\ell} I([m] \mid \ell)\right).
\end{align}
Combining with \eqref{eq:dim} and \eqref{eq-riddimhn}, we have
\begin{equation}
\begin{aligned}\label{eq-ddim-2-ri}
\ddim(\mathcal{H}_n^i) - r_i\ddim(\mathcal{H}_n) &=  r_{i} q^2 I([m]\mid i) - r_iq^2\sum_{\ell=1}^t \left(r_{\ell} I([m] \mid \ell)\right) 
\\&= r_{i} q^2\left(I([m]\mid i)  - \sum_{\ell=1}^t \left(r_{\ell} I([m] \mid \ell)\right)\right).
\end{aligned}   
\end{equation}
Combining with \eqref{eq-immidi}, we have 
$\ddim(\mathcal{H}_n^i) \leq r_i\ddim(\mathcal{H}_n)$.
Combining with \Cref{def-induce-qudits} and $i\in t$,
we have $i\in T_2$.
Conversely,  we also have $\lambda_{ii} = 1$ if $i\in T_2$.
The claim is proved. 
\end{proof}

By Lemmas \ref{lemma-induce-gi}, \ref{lemma-induce-rviprime}, and \Cref{def-instance}, we have the following corollary.
\begin{corol}\label{corol-vprimej-randominstance}
For each left vertex $i\in[t]$,
$(V'_{j})_{j\in Y_i}$ is a random instance of the setting $(G_i,\vec{\lambda}_i,(d_1,\cdots,d_{n-1}))$.
\end{corol}

Similarly, one can also prove the following corollary.
\begin{corol}\label{corol-vstarj-instance}
For each left vertex $i\in[t]$,
$(V^{\star}_{j})_{j\in Y_i}$ is an instance of the setting $(G_i,\vec{\lambda}_i,(d_1,\cdots,d_{n-1}))$.
\end{corol}

\subsection{Subspaces}\label{sec-subspaces}
In this subsection, we define the subspaces.
Given $\mathcal{H}_1,\mathcal{H}_2,\cdots,\mathcal{H}_n,\mathcal{H}_n^1,\cdots,\mathcal{H}_n^t$ as in \Cref{def-induce-qudits} and $V^{\star}_1,\cdots,V^{\star}_m$ as in \Cref{def-induce-subspace},
we define the subspaces $V_1,\cdots,V_m$ as follows.
For each $i\in [m]$, $V_i$ is defined based on the tenser product of $V^{\star}_i$ and the induced qudit $\mathcal{H}_n^i$ or the qudit $\mathcal{H}_n$.
If the relative dimension of the tenser product is less than $r_i$, 
an orthogonal subspace is added to make up the deficit.
The goal is to let $(V_1,\cdots,V_m)$ be an instance of the setting $(G_B,\vec{r},\vec{d})$.

\begin{definition}[Subspaces]\label{def-subspace}
For each $i\in [m]$, define $V_1,\cdots,V_m$ as follows.
\begin{itemize}
    \item If $i\in T_1$, let $V_i =  V^{\star}_i \otimes \mathcal{H}_n^i$.
    \item If $i\in T_2$, let $V_i = \left(V^{\star}_i\otimes\mathcal{H}_n^i\right)\oplus \left(U_i\otimes\mathcal{H}_{[n]\setminus \mathcal{N}(i)}\right)$
    where $U_i$ is an arbitrarily subspace of $\mathcal{H}_{\mathcal{N}(i)\backslash \{n\}}\otimes\left(\bigoplus_{\ell\in[t]\setminus\{i\}}\mathcal{H}_n^{\ell}\right)$ with
	     $$\ddim(U_i) =  \left(r_i\ddim(\mathcal{H}_n) -\ddim(\mathcal{H}_n^{i})\right)\cdot \ddim(\mathcal{H}_{\mathcal{N}(i)\backslash \{n\}}).$$
	\item If $i\in [m]\setminus [t]$, let 
	$V_i = V^{\star}_i\otimes\mathcal{H}_{n}$.
\end{itemize}
\end{definition}

The following two lemmas show that $V_1,\cdots,V_m$ are properly defined.

\begin{lemma}
For each $i\in T_2$, 
$\ddim(U_i)$ is a nonnegative
integer.
In addition, $\ddim(V_i)$ is also a positive integer for each $i\in [m]$.
\end{lemma}

\begin{proof}
We only prove that $\ddim(U_i)$ is a nonnegative
integer for each $i\in T_2$.
Combining with \Cref{def-induce-subspace} and that $\ddim(V^{\star}_i)$ is a positive integer, 
it is immediate that $\ddim(V_i)$ is also a positive integer for each $i\in [m]$. Given $i\in T_2$, we have 
$r_i\ddim(\mathcal{H}_n) \geq  \ddim(\mathcal{H}_n^i)$.
Thus 
$$\ddim(U_i)=\left(r_i\ddim(\mathcal{H}_n) - \ddim(\mathcal{H}_n^i)\right)\cdot \ddim(\mathcal{H}_{\mathcal{N}(i)\backslash \{n\}}) \geq 0.$$
In addition, 
recall that $qr_i$, $q\sum_{\ell=1}^t r_{\ell} I([m] \mid \ell)$ and $qI([m]\mid i)$ are postive integers by \Cref{lem-lower-bound-alphai}.
Thus, by \eqref{eq-def-dj} and \eqref{eq:dim} we have 
\begin{align*}
\ddim(U_i) &= \left(r_i\ddim(\mathcal{H}_n) -\ddim(\mathcal{H}_n^{i})\right)\cdot \ddim(\mathcal{H}_{\mathcal{N}(i)\backslash \{n\}}) \\
&=\left(r_id_n -r_{i} q^2 I([m]\mid i))\right)\cdot \ddim(\mathcal{H}_{\mathcal{N}(i)\backslash \{n\}}) \\
&=  r_iq^2\cdot \ddim(\mathcal{H}_{\mathcal{N}(i)\backslash \{n\}})\cdot \left(- I([m]\mid i) + \sum_{\ell=1}^t \left(r_{\ell} I([m] \mid \ell)\right) \right)\\
&=(qr_i)\cdot \ddim(\mathcal{H}_{\mathcal{N}(i)\backslash \{n\}}) \cdot \left(-q  I([m]\mid i) + q\sum_{\ell=1}^t \left(r_{\ell} I([m] \mid \ell)\right) \right)
\end{align*}
is also an integer.
In summary, we have $\ddim(U_i)$ is a nonnegative integer.
The lemma is proved.
\end{proof}

\begin{lemma}
For each $i\in T_2$,
$$\ddim(U_i)\leq \mathcal{H}_{\mathcal{N}(i)\backslash \{n\}}\otimes\left(\bigoplus_{\ell\in[t]\setminus\{i\}}\mathcal{H}_n^{\ell}\right).$$ 
\end{lemma}
\begin{proof}
Given $i\in T_2$, we have
\begin{align*}
\ddim(U_i) &= \left(r_i\ddim(\mathcal{H}_n) -\ddim(\mathcal{H}_n^{i})\right)\cdot \ddim\left(\mathcal{H}_{\mathcal{N}(i)\setminus \{n\}}\right)\leq  \ddim\left(\mathcal{H}_{\mathcal{N}(i)\setminus \{n\}}\right)\cdot \left(\ddim(\mathcal{H}_n) -\ddim(\mathcal{H}_n^{i})\right)\\
&=\ddim\left(\mathcal{H}_{\mathcal{N}(i)\backslash \{n\}}\right)\cdot\ddim\left(\bigoplus_{\ell\in[t]\setminus\{i\}}\mathcal{H}_n^{\ell}\right) = \ddim\left(\mathcal{H}_{\mathcal{N}(i)\backslash \{n\}}\otimes\left(\bigoplus_{\ell\in[t]\setminus\{i\}}\mathcal{H}_n^{\ell}\right)\right).\qedhere
\end{align*}
\end{proof}

The following two lemmas show that $(V_1,\cdots, V_m)$ is an instance of the setting $(G_B,\vec{r},\vec{d})$.

\begin{lemma}\label{lemma-V12m-GB}
$G_B$ is an interaction graph of $V_1,\cdots,V_m$.
\end{lemma}
\begin{proof}
If $i\in T_1$, we have 
\begin{align*}
V_i =  V'_i \otimes \mathcal{H}_n^i = V_i^{\ast}\otimes \mathcal{H}_{[n]\setminus \mathcal{N}(i)}\otimes \mathcal{H}_n^i.
\end{align*}
If $i\in T_2$, we have 
\begin{align*}
V_i &= V'_i\otimes\mathcal{H}_n^i\oplus U_i\otimes\mathcal{H}_{[n]\setminus \mathcal{N}(i)} = \mathcal{H}_{[n-1]}\otimes\mathcal{H}_n^i\oplus U_i\otimes\mathcal{H}_{[n]\setminus \mathcal{N}(i)} \\&= \mathcal{H}_{[n]\setminus \mathcal{N}(i)}\otimes\left(\mathcal{H}_{\mathcal{N}(i)\setminus \{n\}}\otimes\mathcal{H}_n^i\oplus U_i\right).
\end{align*}
If $i\in [m]\setminus [t]$, we have $V_i = V_i^{\ast}\otimes\mathcal{H}_{[n]\setminus \mathcal{N}(i)}$.
In summary, we always have that 
$\Pi_{V_i}$ acts trivially on the qudits $\mathcal{H}_{[n]\setminus \mathcal{N}(i)}$ for each $i\in [m]$.
Thus we have $G_B$ is one interaction graph of $V_1,\cdots,V_m$.
\end{proof}

\begin{lemma}\label{lemma-V12m-r}
For each left vertex $i\in [m]$, we have $\mathbb{R}(V_i,\mathcal{H}_{[n]}) = r_i$.
\end{lemma}
\begin{proof}
If $i\in T_1$, we have $i\in [t]$.
Then $n\in \mathcal{N}(i)$ and $[n-1]\setminus ([n]\setminus \mathcal{N}(i)) = \mathcal{N}(i)\setminus \{n\}$.
Therefore,
\begin{align*}
\mathbb{R}\left(V_i,\mathcal{H}_{[n]}\right)
&= \mathbb{R}\left(V'_i \otimes \mathcal{H}_n^i,\mathcal{H}_{[n]}\right)
= \mathbb{R}\left(V'_i,\mathcal{H}_{[n - 1]}\right)\times \mathbb{R}\left(\mathcal{H}_n^i,\mathcal{H}_{n}\right)\\
&= \mathbb{R}\left(V_i^{\ast}\otimes \mathcal{H}_{[n]\setminus \mathcal{N}(i)},\mathcal{H}_{[n - 1]}\right)\times \mathbb{R}\left(\mathcal{H}_n^i,\mathcal{H}_{n}\right)\\
&=\mathbb{R}\left(V_i^{\ast},\mathcal{H}_{\mathcal{N}(i)\setminus \{n\}}\right)\times \mathbb{R}\left(\mathcal{H}_n^i,\mathcal{H}_{n}\right)= \frac{r_i\ddim(\mathcal{H}_n)}{\ddim(\mathcal{H}_n^i)}\cdot \frac{\ddim(\mathcal{H}_n^i)}{\ddim(\mathcal{H}_n)} = r_i.
\end{align*}
If $i\in T_2$, by that $U_i$ is a subspace of $\mathcal{H}_{\mathcal{N}(i)\backslash \{n\}}\otimes\left(\bigoplus_{\ell\in[t]\setminus\{i\}}\mathcal{H}_n^{\ell}\right)$, we have 
$\mathcal{H}_{[n-1]}\otimes\mathcal{H}_n^i$ and $U_i\otimes\mathcal{H}_{[n]\setminus \mathcal{N}(i)}$ are  orthogonal.
Thus, we have
\begin{align*}
\mathbb{R}\left(V_i,\mathcal{H}_{[n]}\right) &= \mathbb{R}\left(V'_i\otimes\mathcal{H}_n^i\oplus U_i\otimes\mathcal{H}_{[n]\setminus \mathcal{N}(i)},\mathcal{H}_{[n]}\right)
= \mathbb{R}\left(\mathcal{H}_{[n-1]}\otimes\mathcal{H}_n^i\oplus U_i\otimes\mathcal{H}_{[n]\setminus \mathcal{N}(i)},\mathcal{H}_{[n]}\right)\\
&=\frac{\ddim\left(\mathcal{H}_{[n-1]}\otimes\mathcal{H}_n^i\oplus U_i\otimes\mathcal{H}_{[n]\setminus \mathcal{N}(i)}\right)}{\ddim\left(\mathcal{H}_{[n]}\right)}
=\frac{\ddim\left(\mathcal{H}_{[n-1]}\otimes\mathcal{H}_n^i\right)+\ddim\left( U_i\otimes\mathcal{H}_{[n]\setminus \mathcal{N}(i)}\right)}{\ddim\left(\mathcal{H}_{[n]}\right)}\\
&=\frac{\ddim\left(\mathcal{H}_{[n-1]}\otimes\mathcal{H}_n^i\right)}{\ddim\left(\mathcal{H}_{[n-1]}\otimes\mathcal{H}_{n}\right)}+\frac{\ddim\left( U_i\otimes\mathcal{H}_{[n]\setminus \mathcal{N}(i)}\right)}{\ddim\left(\mathcal{H}_{\mathcal{N}(i)}\otimes \mathcal{H}_{[n]\setminus \mathcal{N}(i)}\right)}
= \frac{\ddim\left(\mathcal{H}_n^i\right)}{\ddim\left(\mathcal{H}_{n}\right)}+\frac{\ddim\left( U_i\right)}{\ddim\left(\mathcal{H}_{\mathcal{N}(i)}\right)}\\
&= \frac{\ddim\left(\mathcal{H}_n^i\right)}{\ddim\left(\mathcal{H}_{n}\right)}+\frac{\left(r_i\ddim(\mathcal{H}_n) -\ddim(\mathcal{H}_n^{i})\right)\cdot \ddim(\mathcal{H}_{\mathcal{N}(i)\backslash \{n\}})}{\ddim\left(\mathcal{H}_{\mathcal{N}(i)}\right)}\\
&=\frac{\ddim\left(\mathcal{H}_n^i\right)}{\ddim\left(\mathcal{H}_{n}\right)}+\frac{\left(r_i\ddim(\mathcal{H}_n) -\ddim(\mathcal{H}_n^{i})\right)}{\ddim\left(\mathcal{H}_{n}\right)} = r_i.
\end{align*}
If $i\in [m]\setminus [t]$,
we also have
\begin{align*}
\mathbb{R}(V_i,\mathcal{H}_{[n]}) =  \frac{\ddim\left(V'_i\otimes\mathcal{H}_{n}\right)}{\ddim(\mathcal{H}_{[n]})} = \frac{\ddim\left(V_i^{\ast}\otimes\mathcal{H}_{[n-1]\setminus \mathcal{N}(i)}\otimes \mathcal{H}_n\right)}{\ddim(\mathcal{H}_{[n]})} = \frac{\ddim\left(V_i^{\ast}\right)}{\ddim(\mathcal{H}_{\mathcal{N}(i)})}  = \mathbb{R}(V_i^{\ast},\mathcal{H}_{\mathcal{N}(i)}) = r_i.
\end{align*}
\end{proof}

By Lemmas \ref{lemma-V12m-GB} and \ref{lemma-V12m-r}, we have the following corollary.
\begin{corol}\label{corol-v1m-setting}
The instance $(V_1,\cdots, V_m)$ is of the setting $(G_B,\vec{r},\vec{d})$.
\end{corol}

The following is an easy property of $V_1,\cdots, V_m$ by \Cref{def-subspace}.

\begin{lemma}\label{lemma-vstar2v-span}
If for each left vertex $i\in[t]$,
$(V^{\star}_{j})_{j\in Y_i}$ span the whole space $\mathcal{H}_{[n-1]}$, then $V_1,\cdots, V_m$ span the whole space $\mathcal{H}_{[n]}$.
\end{lemma}
\begin{proof}
By $(V^{\star}_{j})_{j\in Y_i}$ span the whole space $\mathcal{H}_{[n-1]}$ for each left vertex $i\in[t]$,
we have $(V^{\star}_{j}\otimes \mathcal{H}_n^i)_{j\in Y_i}$
span the whole space $\mathcal{H}_{[n-1]}\otimes \mathcal{H}_n^i$.
In addition, by \Cref{def-subspace} we have 
$V^{\star}_{i}\otimes \mathcal{H}_n^i = V_i$ 
and 
$$V^{\star}_{j}\otimes \mathcal{H}_n^i \subseteq  V^{\star}_{j}\otimes \mathcal{H}_n = V_j$$ 
for each $j\in Y_i\setminus \{i\} = [m]\setminus [t]$. 
Thus, we have $V_1,\cdots, V_m$ span the whole space $\mathcal{H}_{[n]}$.
\end{proof}

\subsection{Proof of Theorem \ref{shearersboundistight:above}}
\label{sec-Proof-of-Theorem-shearersboundistight}
In this subsection, we finish the proof of \Cref{shearersboundistight:above}.

The following lemma will be used in the proof of \Cref{shearersboundistight:above}.
Given $G_B= ([m],[n],E_B)$ and $\vec{r}=(r_1,\cdots,r_m)$ where $\vec{r} \in \mathcal{I}(G_B)$, 
let $G'_B$ be the resulted graph by adding a left vertex $m+1$ where $\mathcal{N}(m+1) = [n]$ to $G_B$ and let $\vec{r}'$ be $(r_1,\cdots,r_m,I(G_B,\vec{r}))$.
With this lemma, one can verify \Cref{shearersboundistight:above} (\ref{item-b-thmshearersboundistight}) on $G_B$ and $\vec{r}$
by applying \Cref{shearersboundistight:above} (\ref{item-a-thmshearersboundistight})
to $G'_B$ and $\vec{r}'$.

\begin{lemma}\label{lem-gbprime-rprime}
Given any $G_B = ([m],[n],E_B)$, and any rational $\vec{r} = (r_1,\cdots,r_m)\in \left[0,1\right]^{m}$,
let $G_B' = ([m+1],[n],E'_B)$ where $E'_B = E_B \cup \{(m+1,1),(m+1,2),\cdots,(m+1,n)\}$, i.e., 
the left vertex $m+1$ has neighbors $\mathcal{N}(m+1) = [n]$. 
If $\vec{r} \in \mathcal{I}(G_B)$, let $\delta = I(G_B,\vec{r})$; otherwise, let $\delta = 0$.
Then for any $r_{m+1}\geq \delta$, 
$\vec{r}' \triangleq\left(r_1,\cdots,r_{m},r_{m+1}\right)\not \in \mathcal{I}(G'_B)$.
In addition, given $\vec{d}=(d_1,\cdots,d_n)$,
if
$\Pr\left[\mathbb{R}\left(\sum_{V\in \vec{U}} V\right)=1\right]=1$
for the random instance $\vec{U}$ of the setting $(G'_B,\vec{r}',\vec{d})$,
then 
$\Pr[\R(\bigplus_{V\in \vec{V}} V) \in [1 - r_{m+1},1-\delta]]=1$
for the random instance $\vec{V}$ of the setting $(G_B,\vec{r},\vec{d})$.
\end{lemma}

\begin{proof}
At first, we prove $\vec{r}' \not \in \mathcal{I}(G'_B)$.
If $\vec{r} \not\in \mathcal{I}(G_B)$,
we have $r_{m+1} \geq \sigma=0$.
Then by $\vec{r} \not\in \mathcal{I}(G_B)$ and the definition of $G'_B$, 
$\vec{r}'=(r_1,\cdots,r_m,r_{m+1}) \not \in \mathcal{I}(G'_B)$ is immediate.
If $\vec{r}\in \mathcal{I}(G_B)$, we have $r_{m+1} \geq \delta = I(G_B,\vec{r})$.
By \Cref{def: indpoly} we have
\begin{align*}
I(G'_B,\vec{r}',[m+1]) &=\sum_{T\in \mathsf{Ind}(G'_B)}(-1)^{\abs{T}}\prod_{i\in T}r_i = -r_{m+1} + \sum_{T\in \mathsf{Ind}(G_B)}(-1)^{\abs{T}}\prod_{i\in T}r_i \\&=  -r_{m+1} + I(G_B,\vec{r})\leq 0.
\end{align*}
Combining with \Cref{lemma-eqinshearersbound},
we have $\vec{r}'\not \in \mathcal{I}(G'_B)$.

In the next, we show that if
$\Pr\left[\mathbb{R}\left(\sum_{V\in \vec{U}} V\right)=1\right]=1$
for the random instance $\vec{U}$ of the setting $(G'_B,\vec{r}',\vec{d})$,
then 
$\Pr[\R(\bigplus_{V\in \vec{V}} V) \in [1 - r_{m+1},1-\delta]]=1$
for the random instance $\vec{V}$ of the setting $(G_B,\vec{r},\vec{d})$.
Then the lemma is immediate.
Given the random instance $V_1,\cdots,V_{m},V_{m+1}$ of the setting $(G'_B,\vec{r}',\vec{d})$ where
\begin{align}\label{eq-prsum}
\Pr\left[\mathbb{R}\left(\sum_{i\in [m+1]} V_i\right)=1\right]=1,
\end{align}
we have
\begin{align*}
    \mathbb{R}\left(\sum_{i\in [m]} V_i\right) \geq \mathbb{R}\left(\sum_{i\in [m+1]} V_i\right) - \mathbb{R}\left( V_{m+1}\right) = \mathbb{R}\left(\sum_{i\in [m+1]} V_i\right) - r_{m+1}.
\end{align*}
Combining with \eqref{eq-prsum},
we have 
\begin{align}\label{eq-prrsumvlowerbound}
\Pr\left[\mathbb{R}\left(\sum_{i\in [m]} V_i\right)\geq 1 - r_{m+1}\right]=1.
\end{align}
In addition, if $\vec{r} \in \mathcal{I}(G_B)$, by \Cref{thm:pnas} we have 
$\mathbb{R}\left(\sum_{i\in [m]} V_i\right)\leq 1 - I(G_B,\vec{r}) = 1 - \delta$.
If $\vec{r} \not\in \mathcal{I}(G_B)$,
we also have 
$\mathbb{R}\left(\sum_{i\in [m]} V_i\right)\leq 1 = 1 - \delta$.
Therefore, we always have 
$\mathbb{R}\left(\sum_{i\in [m]} V_i\right)\leq 1 - \delta$.
Combining with \eqref{eq-prrsumvlowerbound},
we have $\Pr[\R(\bigplus_{V\in \vec{V}} V) \in [1 - r_{m+1},1-\delta]]=1$.
Moreover, by $(V_1,\cdots,V_m,V_{m+1})$ is the random instance of the setting $(G'_B,r',\vec{d})$,
one can verify that $(V_1,V_2,\cdots,V_m)$ is the random instance of the setting $(G_B,\vec{r},\vec{d})$,
The lemma is proved.
\end{proof}

Now we can prove Theorem \ref{shearersboundistight:above},
the main idea of which has been illustrated in \Cref{sec-shearer'sboundistight}.

\begin{proof}[Proof of Theorem \ref{shearersboundistight:above}]
	\emph{(a).} 
	The proof is by induction on the number of left vertices in $G_B$. 
	For the base case that the number of left vertices in $G_B$ is no more than 1, the theorem holds trivially.
	
	For the induction step, given any integer $m\geq 2$,
	we assume that (\ref{item-a-thmshearersboundistight}) of this theorem holds for any interaction graph where the number of left vertices is no more than $m - 1$. In the following, we prove that the theorem also holds for any graph $G_B$ where the number of left vertices is $m$. 

If $\vec{r} \not\in \mathcal{I}(G_B)$, by \Cref{lemma-eqinshearersbound} 
we have there is some 
$R\subseteq[m]$ such that $I(R)\leq0$.
Let $S\subseteq R$ be a set such that $I(S)\leq 0$ and $I(T)>0$ for each $T\subsetneq S$.
In the following, we consider the two possible cases $S\subsetneq [m]$ and $S = [m]$ separately.

At first we assume $S\subsetneq [m]$.
Let $G'$ be the interaction graph obtained from $G_B$ by deleting the left vertices $[m]\setminus S$ and the edges connected to these vertices.
Let $\vec{r}_S$ be $(r_i)_{i\in S}$.
By $I(G',\vec{r}_S,S)=I(S)\leq 0$ and \Cref{lemma-eqinshearersbound},
we have $\vec{r}_S \not\in \mathcal{I}(G')$.
Let $\alpha = \prod_{i\in S}q_i$ and $\beta = \prod_{i\not\in S}q_i$.
Let $\mathcal{N}(i)$ denote $\mathcal{N}(G_B,i)$ for simplification.
For each left vertex $i\in S$,
we have $\mathcal{N}(G',i) = \mathcal{N}(i)$.
Combining with the assumption $r_i =1$ if $\mathcal{N}(i) = \emptyset$,
we have $r_i =1$ if $\mathcal{N}(G',i) = \emptyset$.
Combining with $\vec{r}_S \not\in \mathcal{I}(G')$ and the induction hypothesis,
we have there exists some $\vec{d}=(d_1,d_2,\cdots,d_n)$ where 
$d_i\leq \alpha^{4^{m}}$ for each $i\in [n]$ such that
random subspaces of the setting $(G',\vec{r}_S,\vec{d})$ span the whole space.
Thus, there exists a subspace set of the setting $(G',\vec{r}_S,\vec{d})$ spanning the whole space.
Combining with \Cref{prop:multiple} (\ref{prop:multiple-item-a}), 
we have for any set of qudits
$\mathcal{H}_1,\mathcal{H}_2,\cdots,\mathcal{H}_n$ where $\ddim(\mathcal{H}_{1},\cdots,\mathcal{H}_{n}) = \beta\vec{d}$, there exists a subspace set $(V'_i)_{i\in S}$ of $\mathcal{H}_{[n]}$ such that
$(V'_i)_{i\in S}\sim (G',\vec{r}_S,\beta \vec{d})$ and $(V'_i)_{i\in S}$ spans the whole space $\mathcal{H}_{[n]}$.
Moreover, by $\alpha = \prod_{i\in S}q_i$, $\beta = \prod_{i\not\in S}q_i$ and 
$d_j\leq \alpha^{4^{m}}$ for each $j\in [n]$,
we have 
\begin{align}\label{eq-betadi-uppperbound}
\beta d_j\leq \alpha^{4^{m}}\cdot \beta \leq \prod_{i\in [m]}q_i^{4^{m}} = q^{4^{m}}.
\end{align}

We define a subspace set $(V^{\star}_1,V^{\star}_2,\cdots,V^{\star}_m)$ of $\mathcal{H}_{[n]}$ as follows.
For each $i\in S$, let 
$V^{\star}_i = V'_{i}$.
For each $i\in [m]\setminus S$,
let $V^{\star}_i = V^{\ast}_i\otimes \mathcal{H}_{[n]\setminus \mathcal{N}(i)}$ where $V^{\ast}_i$ is an arbitrary subspace of $\mathcal{H}_{\mathcal{N}(i)}$ with 
$\mathbb{R}\left(V^{\ast}_i,\mathcal{H}_{\mathcal{N}(i)}\right) = r_i$.
First of all, we show that $\ddim(V^{\ast}_i)$ is an integer for each $i\in [m]\setminus S$. Thus $(V^{\star}_1,V^{\star}_2,\cdots,V^{\star}_m)$ is well defined.
Given $i\in [m]\setminus S$, if $r_i = 1$, then by $\mathbb{R}\left(V^{\ast}_i,\mathcal{H}_{\mathcal{N}(i)}\right) = r_i$,
we have $\ddim(V^{\ast}_i) = \ddim(\mathcal{H}_{\mathcal{N}(i)})$ is an integer.
In the following, we assume  $i\in [m]\setminus S$ and $r_i < 1$.
Combining with the assumption if $\mathcal{N}(i) = \emptyset$ then $r_i =1$,
we have $\mathcal{N}(i) \neq \emptyset$.
Assume \emph{w.l.o.g.} $j\in \mathcal{N}(i)$ for some right vertex $j\in [n]$.
By the definitions of $r_i$, $\beta$ and $\mathcal{H}_1,\mathcal{H}_2,\cdots,\mathcal{H}_n$,
we have 
\begin{equation*}
\begin{aligned}
\ddim\left(V^{\ast}_i\right) &= r_i\cdot \ddim\left(\mathcal{H}_{\mathcal{N}(i)}\right) = r_i\cdot\prod_{k\in \mathcal{N}(i)}\ddim\left(\mathcal{H}_{k}\right) =r_i\cdot\prod_{k\in \mathcal{N}(i)}\left(\beta d_k\right) = \frac{p_i}{q_i}\cdot\prod_{j\in \mathcal{N}(i)}\left(\beta d_j\right)\\
&= \frac{p_i\beta d_j}{q_i}\prod_{k\in \mathcal{N}(i)\setminus \{j\}}\left(\beta d_k\right)
=\frac{p_id_j}{q_i}\left(\prod_{\ell \not\in S}q_{\ell}\right)\prod_{k\in \mathcal{N}(i)\setminus \{j\}}\left(\beta d_k\right)
=p_id_j\left(\prod_{\ell \not\in S\cup\{i\}}q_{\ell}\right)\prod_{k\in \mathcal{N}(i)\setminus \{j\}}\left(\beta d_k\right),
\end{aligned}
\end{equation*}
where the last equality is by that $i\in [m]\setminus S$.
Therefore, by $\beta, p_{\ell}, d_{k}, q_{\ell}$ are integers for each $\ell\in [m]$ and $k\in [n]$,
we have $\ddim(V^{\ast}_i)$ is also an integer.

In the next, we show that $(V^{\star}_1,V^{\star}_2,\cdots,V^{\star}_m)$ is an instance of the setting $(G_B,\vec{r},\beta\vec{d})$ spanning the whole space $\mathcal{H}_{[n]}$.
Thus, (\ref{item-a-thmshearersboundistight}) of this theorem is immediate by \Cref{lem: span} and 
\eqref{eq-betadi-uppperbound}.
At first, we have $\ddim(\mathcal{H}_{1},\cdots,\mathcal{H}_{n}) = \beta\vec{d}$.
Meanwhile, given a left vertex $i\in [m]$, 
we have either $i\in S$ or $i\in [m]\setminus S$.
If $i\in S$, we have $V^{\star}_i = V'_{i}$.
Combining with $(V'_i)_{i\in S}\sim (G',\vec{r}_S,\beta \vec{d})$,
we have 
$$V^{\star}_i = V'_{i} = U_i\otimes \mathcal{H}_{[n]\setminus \mathcal{N}(G',i)} = U_i\otimes \mathcal{H}_{[n]\setminus \mathcal{N}(i)},$$
where $U_i$ is some subspace of $\mathcal{H}_{\mathcal{N}(i)}$ with 
$\mathbb{R}\left(U_i,\mathcal{H}_{\mathcal{N}(i)}\right) = r_i$.
Thus $\mathbb{R}\left(V^{\star}_i, \mathcal{H}_{[n]}\right) =\mathbb{R}\left( V'_{i}, \mathcal{H}_{[n]}\right) = r_i$.
Otherwise, $i\in [m]\setminus S$.
We also have $V^{\star}_i = V^{\ast}_i\otimes \mathcal{H}_{[n]\setminus \mathcal{N}(i)}$ and 
$\mathbb{R}\left(V^{\star}_i, \mathcal{H}_{[n]}\right) = \mathbb{R}\left(V^{\ast}_i,\mathcal{H}_{\mathcal{N}(i)}\right) = r_i$.
Thus, $G_B$ is an interaction graph of $(V^{\star}_1,V^{\star}_2,\cdots,V^{\star}_m)$ 
and $\mathbb{R}\left(\vvec{V}, \mathcal{H}_{[n]}\right) = \vec{r}$.
In addition, by that $(V'_i)_{i\in S}$ spans the whole space $\mathcal{H}_{[n]}$ and that $V^{\star}_i = V'_{i}$ for each $i\in S$, we have $(V^{\star}_1,V^{\star}_2,\cdots,V^{\star}_m)$ spans the whole space $\mathcal{H}_{[n]}$ immediately.
The case $S\subsetneq [m]$ is proved.

In the following we assume $S=[m]$. 
Recall that $S$ satisfies $I(S)\leq 0$ and $I(T)>0$ for each $T\subsetneq S$. Combining with $S=[m]$ and \Cref{le:property of ind by def} (\ref{item-d-le-propertyindbydef}),
we have $G_D(G_B)$ is connected. 
Also recall that $m\geq 2$,
we have there must be a right vertex in $G_B$ with at least two neighbors in the left vertices. Without loss of generality, we assume the right vertex $n$ is such a vertex and $\mathcal{N}(n) = [t]$  where $t\geq2$.
In addition, by $m\geq 2$ and $I(T)>0$ for each $T\subsetneq S = [m]$,
we have $r_i <1$ for each left vertex $i$ of $G_B$.
Otherwise, $r_i \geq 1$. 
We have $I(\{i\})= 1 - r_i \leq 0$ by \Cref{def: indpoly}, a contradiction with that $I(T)>0$ for each $T\subsetneq [m]$.
For each left vertex $i\in [m]$,
combining $r_i <1$ with the assumption that $r_i =1$ if $\mathcal{N}(i) = \emptyset$ in the theorem,
we have $\mathcal{N}(i) \neq \emptyset$.
Thus, all the assumptions in \Cref{def-induce-graph} are satisfied.
One can define the induced interaction graphs,
the induced qudits with induced dimensions, and
 the induced subspaces with the induced relative dimensions as in 
Definitions \ref{def-induce-graph}, \ref{def-induce-dimension}, \ref{def-induce-qudits} and 
\ref{def-induce-subspace},
and define the dimensions of qudits $\vec{d}=(d_1,\cdots,d_n)$ and the subspaces $(V_1,\cdots,V_m)$ as in 
Definitions \ref{def-induce-dimension} and \ref{def-subspace}.
By \Cref{corol-v1m-setting}, we have 
$(V_1,\cdots,V_m)$ is of the setting $(G_B,\vec{r},\vec{d})$.
Moreover, 
we claim that $d_i\leq q^{4^{m}}$ for each $i\in [m]$ and  
$V_1,\cdots,V_m$ spans the whole space.
Then (\ref{item-a-thmshearersboundistight}) of this theorem is immediate by \Cref{lem: span}.
In the following, we prove the claim.

At first, we prove that $d_i\leq q^{4^{m}}$ for each $i\in [m]$.
Recall the definitions of $G_i$ and $\vec{\lambda}_i$
in \Cref{def-induce-graph}.
By $I([m]) = I(S)\leq 0$ and \Cref{lemma-induction},
for each $i\in [t]$ we have $\vec{\lambda}_i \not \in \mathcal{I}(G_i)$ and
if $\mathcal{N}(G_i,j) = \emptyset$ for some left vertex $j\in Y_i$ in $G_i$, 
then $\lambda_{ij} = 1$. 
Combining with the induction hypothesis,
we have there exists some $(a_{ij})_{j\in Y_i}$
such that $a_{ij}\leq q^{4^{m-t+1}}$ for each $j\in Y_i$ and the random instance of the setting $(G_i,\vec{\lambda}_i,(a_{ij})_{j\in Y_i})$ spans the whole space.
Combining with the definition of $(d_{ij})_{j\in Y_i}$ in \Cref{def-induce-dimension}, 
we have if $\rho_i<1$,
then $d_{ij} = a_{ij} \leq q^{4^{m-t+1}}$ for each $j\in Y_i$;
otherwise $\rho_i\geq1$, we also have $d_{ij} = q \leq q^{4^{m-t+1}}$
Combining with \eqref{eq-def-dj}, we have for each $j\in [n-1]$,
\begin{align}
    d_{j} = \prod_{i\in t}d_{ij}\leq q^{t\cdot 4^{m-t+1}} \leq q^{4^{m}},
\end{align}
where the last inequality is by $t\geq 2$.
Combining with \Cref{lemma-d-positive-integer}, we have
$d_{j}\leq q^{4^{m}}$ for each $j\in [n]$.

In the following, we show that $V_1,\cdots,V_m$ spans the whole space. Then (\ref{item-a-thmshearersboundistight}) of this theorem is proved.
Recall that for each $i\in [t]$, the random instance of the setting $(G_i,\vec{\lambda}_i,(a_{ij})_{j\in Y_i})$ spans the whole space.
If $\rho_i<1$, we have $d_{ij} = a_{ij}$ by \Cref{def-induce-dimension}, then the random instance of the setting $(G_i,\vec{\lambda}_i,(d_{ij})_{j\in Y_i})$ spans the whole space.
If $\rho_i\geq 1$, we have $r_{ii} = 1$ by \eqref{eq-lambda-ii}.
Thus we also have the random instance of the setting $(G_i,\vec{\lambda}_i,(d_{ij})_{j\in Y_i})$ spans the whole space.
Combining with \eqref{eq-def-dj} and \Cref{prop:multiple} (\ref{prop:multiple-item-b}),
we have the random subspaces of the setting $(G_i,\vec{\lambda}_i,(d_{j})_{j\in Y_i})$ spans the whole space.
Combining with \Cref{corol-vprimej-randominstance},
we have $(V'_{j})_{j\in Y_i}$ as in \Cref{def-induce-subspace}, is a random instance of the setting $(G_i,\vec{\lambda}_i,(d_{j})_{j\in Y_i})$ and spans the whole space.
In other words, $\Pr[\mathbb{R}(\bigplus_{j\in Y_i} V'_{j},\mathcal{H}_{[n-1]})=1]=1$.
By the union bound, we have 
$$\Pr\left(\bigvee_{i\in [t]}\left[\mathbb{R}\left(\bigplus_{j\in Y_i} V'_{j},\mathcal{H}_{[n-1]}\right)=1\right]\right) \geq 1 - \sum_{i\in [t]}\left(1 - \Pr\left[\mathbb{R}\left(\bigplus_{j\in Y_i} V'_{j},\mathcal{H}_{[n-1]}\right)=1\right]\right) = 1.$$
Thus, we have there are some spaces $(V^{\star}_{j})_{j\in [m]}$, as in \Cref{def-induce-subspace}, satisfies that 
$\mathbb{R}(\bigplus_{j\in Y_i} V^{\star}_{j},\mathcal{H}_{[n-1]})=1$.
Combining with \Cref{lemma-vstar2v-span}, we have that
the subspaces $V_1,\cdots,V_m$ span the whole space $\mathcal{H}_{[n]}$, which finishes the proof of (\ref{item-a-thmshearersboundistight}).

\emph{(b).} Given $G_B = ([m],[n],E_B)$ and $\vec{r} = (r_1,\cdots,r_m)$ where
$\vec{r} \in \mathcal{I}(G_B)$,
by \Cref{lemma-eqinshearersbound},
we have $I(G_B,\vec{r},[m])>0$.
Let $G_B' = ([m+1],[n],E'_B)$ where $E'_B = E_B \cup \{(m+1,1),(m+1,2),\cdots,(m+1,n)\}$ and $\vec{r}'$ be $(r_1,\cdots,r_m,r_{m+1})$ where $r_{m+1} =  I(G_B,\vec{r},[m])>0$. 
By \Cref{lem-gbprime-rprime},
we have $\vec{r}'\not \in \mathcal{I}(G'_B)$.
Let $q_{m+1}$ be the minimum positive integer $k$ such that $kr_{m+1}$ is an integer.
Thus, we have 
$q_{m+1}\leq q$ if $qr_{m+1}$ is an integer.
In addition, by \Cref{def: indpoly} we have 
\begin{align*}
qr_{m+1}  = r_{m+1}\prod_{i\in [m]}q_i = 
I(G_B,\vec{r},[m])\prod_{i\in [m]}q_i
= \sum_{T\in \mathsf{Ind}(G_B)}(-1)^{\abs{T}}\prod_{i\in T}p_i\prod_{i\in [m]
\setminus T}q_i.
\end{align*}
Combining with $p_i,q_i$ are integers for each $i\in [m]$,
we have $qr_{m+1}$ is also an integer.
Thus we have 
$q_{m+1} \leq q$.
Therefore, we have 
$$\prod_{i\in [m+1]}q_i \leq q^2.$$
Recall that $\vec{r}'\not \in \mathcal{I}(G'_B)$.
Applying (\ref{item-a-thmshearersboundistight}) of this theorem to $G'_B$ and $\vec{r}'$, we have there is some 
$\vec{d}=(d_1,\cdots,d_n)$ such that 
$$d_i\leq \prod_{i\in [m+1]}q^{4^{m+1}}_i \leq q^{8+4^{m}}$$
for each $i\in [n]$
and
$\Pr\left[\mathbb{R}\left(\sum_{V\in \vec{U}} V\right)=1\right]=1$
for the random instance $\vec{U}$ of the setting $(G'_B,\vec{r}',\vec{d})$.
Combining with \Cref{lem-gbprime-rprime},
we have 
\begin{align}\label{eq-prsumvi-lower}
\Pr\left[\mathbb{R}\left(\sum_{i\in [m]} V_i\right)\in[ 1 - r_{m+1},1-I(G_B,\vec{r},[m])]\right]=1,
\end{align}
for the random instance $(V_1,V_2,\cdots,V_m)$ of the setting $(G_B,\vec{r},\vec{d})$.
Combining with $r_{m+1} = I(G_B,\vec{r},[m])$,
we have 
\begin{align*}
\Pr\left[\mathbb{R}\left(\sum_{i\in [m]} V_i\right)= 1 - I(G_B,\vec{r},[m])\right]=1,
\end{align*}
Recall that $d_i\leq q^{8+4^{m}}$
for each $i\in [n]$.
The conclusion is immediate.
\end{proof}

\subsection{Proof of Theorem \ref{Shearer'sboundistightforQLLL}}\label{sec-proofShearer'sboundistightforqLLL}
In this subsection, we prove Theorem \ref{Shearer'sboundistightforQLLL}. 
Theorem \ref{Shearer'sboundistightforQLLL} provides an interval of $\R\left(\bigplus_{V\in \vec{V}} V\right)$ for the random instance $\vec{V}$ of some given setting.
By \Cref{lem-gbprime-rprime}, it is sufficient to show that the random instance $\vec{V}'$ of another setting spans the whole space.
Given any $G_B$, $\vec{r}=(r_1,\cdots,r_m)$, $\vec{d}=(d_1,\cdots,d_n)$ and any positive integer $t$,
to prove that the random instance of  
the setting $(G_B,\vec{r},\vec{d})$
spans the whole space, 
\Cref{lemreducedimensionqudit} shows that it is sufficient to prove that 
the random instance of  
the setting $(G_B,\vec{r}',\vec{d}')$
spans the whole space for some $\vec{r}' = (r'_1,\cdots,r'_m)$ and $\vec{d}'$,
where $\vec{r}'$ is tailored from $\vec{r}$
such that $r'_i$ is a multiple of $1/t$ for each $i\in [m]$, and $tn\vec{d}'\leq \vec{d}$.
Because $\vec{r}'$ is tailored to multiples of $1/t$,
by \Cref{shearersboundistight:above} (\ref{item-a-thmshearersboundistight}) 
we have there exists some $\vec{d}'$ which is upper bounded by a function of $t$ and $m$ such that the random instance of the setting $(G_B,\vec{r}',\vec{d}')$
spans the whole space.
Thus Theorem \ref{Shearer'sboundistightforQLLL} is proved.

The following lemma is the core of the proof of Theorem \ref{Shearer'sboundistightforQLLL}.
\begin{lemma}\label{lemreducedimensionqudit}
For any interaction graph $G_B = ([m],[n],E_B)$, any positive integer $t$, any rational $\vec{r}= (r_1,\cdots,r_m)\in\left[0,1\right]^m$,
let $\vec{r}' =(r'_1,\cdots,r'_{m})$ 
where 
\begin{align}\label{eq-bound-riti}
    r'_i = \max\left\{\frac{\left\lfloor tr_i - 1\right\rfloor}{t},0\right\}
\end{align}
for each $i\in [m]$.
If there exists some $\vec{d}'$ such that 
the random instance of  
the setting $(G_B,\vec{r}',\vec{d}')$
spans the whole space,
then for any $\vec{d}\geq tn\vec{d}'$, the random instance of  
the setting $(G_B,\vec{r},\vec{d})$
spans the whole space.
\end{lemma}
\begin{proof}
By \Cref{lem: span}, to prove this lemma, it is sufficient to construct an instance $(V_1,V_2,\cdots,V_m)\sim (G_B,\vec{r},\vec{d})$ spanning the whole space 
$\mathcal{H}_{[n]}$ where $\ddim\left(\mathcal{H}_{1},\cdots,\mathcal{H}_{n}\right) = \vec{d}$.
Suppose $\vec{d}' = (d'_1,\cdots,d'_n)$ and $\vec{d} = (d_1,\cdots,d_n)$.
The construction of $(V_1,V_2,\cdots,V_m)$ is as follows.
\begin{itemize}
	\item For each $j\in [n]$, decompose $\mathcal{H}_j$ into two orthogonal subspaces $\mathcal{H}_j=\mathcal{H}_j^a\oplus\mathcal{H}_j^b$ with 
	\begin{align*}
	\dim(\mathcal{H}_j^a)=d'_j \cdot \left\lfloor \frac{d_j}{d'_j}\right\rfloor,\quad \dim(\mathcal{H}_j^b)=d_j - \dim(\mathcal{H}_j^a).
	\end{align*}
	We emphasize that $\dim(\mathcal{H}_j^b)=0$ if $d_j$ is a multiple of $d'_j$.
    Note that $\dim(\mathcal{H}_j^a)$ is a multiple of $d'_j$. 
    Thus, by \Cref{prop:multiple} (\ref{prop:multiple-item-a}) and that the random instance of the setting $(G_B,\vec{r}',\vec{d}')$ spans the whole space, we have there is an instance $\left(V^{\ast}_1\otimes \mathcal{H}_{[n]\setminus \mathcal{N}(1)}^a,\cdots,V^{\ast}_{m}\otimes \mathcal{H}_{[n]\setminus \mathcal{N}(m)}^a\right)\sim G_B$ spanning  $\mathcal{H}_{[n]}^a$ with relative dimension $\vec{r}'$ to $\mathcal{H}_{[n]}^a$.
	\item For any $i\in[m]$ where $r'_i>0$, let 
	\begin{align*}
	U_i = \bigplus_{j \in \mathcal{N}(i)} \left(\mathcal{H}_j^b\otimes\mathcal{H}_{\mathcal{N}(i)\setminus\{j\}}\right),\quad V_i=\left(V^{\ast}_i\oplus U_i\oplus V^{\star}_i\right)\otimes \mathcal{H}_{[n]\setminus \mathcal{N}(i)}
	\end{align*}
	where $V^{\star}_i$
	is an arbitrary subspace of $\mathcal{H}_{\mathcal{N}(i)}$ which is orthogonal with $V^{\ast}_i\oplus U_i$  and of dimension
	\begin{align*}
	\ddim(V^{\star}_i) = r_i \ddim\left(\mathcal{H}_{\mathcal{N}(i)}\right) - \ddim\left(V^{\ast}_i\oplus U_i\right).
	\end{align*}
	We remark that it is possible that $\dim(V^{\star}_i) = 0$.
	\item For any $i\in[m]$ where $r'_i = 0$, let 
	$V_i=V^{\star}_i\otimes \mathcal{H}_{[n]\setminus \mathcal{N}(i)}$
	where $V^{\star}_i$
	is an arbitrary subspace of $\mathcal{H}_{\mathcal{N}(i)}$ of dimension
	$r_i \ddim\left(\mathcal{H}_{\mathcal{N}(i)}\right)$.
\end{itemize}

At first, we show that $\ddim(V^{\star}_i)$ is a nonnegative integer for each $i\in [m]$.
Thus, $(V_1,V_2,\cdots,V_m)$ is properly defined.
Recall that when we talk about the instance of the setting $(G_B,\vec{r},\vec{d})$,
we always assume that the Hilbert space $\mathcal{H}_{[n]}$ with the dimension vector $\vec{d}$ can admit an instance with the interaction graph $G_B$ and the
relative dimension vector $\vec{r}$ as in \Cref{def-instance}.
Thus, we have $r_i \prod_{j\in \mathcal{N}(i)}d_j$ is an integer for each $i\in [m]$,
otherwise, some subspace in the instance must has a fractional dimension, which is a contradiction.
By $r_i \prod_{j\in \mathcal{N}(i)}d_j$ is an integer for each $i\in [m]$, we have $\ddim(V^{\star}_i)$ is also an integer.
In the next, we show that $\ddim(V^{\star}_i)$ is nonnegative for each $i\in [m]$. 
If $r'_i = 0$, this conclusion is immediate.
In the following, we assume $r'_i > 0$.
Recall that $(V^{\ast}_1\otimes \mathcal{H}^a_{[n]\setminus \mathcal{N}(1)},\cdots,V^{\ast}_m\otimes \mathcal{H}^a_{[n]\setminus \mathcal{N}(m)})\sim G_B$ has relative dimension $\vec{r}'$ to $\mathcal{H}_{[n]}^a$.
Thus, we have $V^{\ast}_i$ is a subspace of $\mathcal{H}_{\mathcal{N}(i)}^a$
with relative dimension $r'_i$ to $\mathcal{H}_{\mathcal{N}(i)}^a$.
Therefore, we have 
\begin{align}\label{eq-dim-vasti}
    \ddim(V^{\ast}_i) =r'_i\cdot\ddim(\mathcal{H}_{\mathcal{N}(i)}^a) \leq r'_i\cdot\ddim(\mathcal{H}_{\mathcal{N}(i)}).
\end{align}
Meanwhile, for each $j\in [n]$, by the definition of $\mathcal{H}_j^b$ we have $\dim(\mathcal{H}_j^b) \leq d'_j$. 
Combining with $\ddim(\mathcal{H}_{j})= d_j\geq  tn\cdot d'_j$,
we have 
$\mathbb{R}\left(\mathcal{H}_j^b,\mathcal{H}_j\right) \leq \frac{1}{tn}$.
Thus, by the definition of $U_i$ we have
\begin{align*}
&\mathbb{R}\left(U_i,\mathcal{H}_{\mathcal{N}(i)}\right)
=\bigplus_{j \in \mathcal{N}(i)}\left(\mathbb{R}\left(\mathcal{H}_j^b\otimes\mathcal{H}_{\mathcal{N}(i)\setminus\{j\}},\mathcal{H}_{\mathcal{N}(i)}\right)\right)
= \bigplus_{j \in \mathcal{N}(i)}\mathbb{R}\left(\mathcal{H}_j^b,\mathcal{H}_j\right) 
\leq \frac{n}{tn} = \frac{1}{t}.
\end{align*}
Therefore,
\begin{align}\label{eq-ref-ddimui}
&\quad\ddim\left(U_i\right) = \mathbb{R}\left(U_i,\mathcal{H}_{\mathcal{N}(i)}\right)\cdot\ddim\left(\mathcal{H}_{\mathcal{N}(i)}\right)\leq \frac{\ddim\left(\mathcal{H}_{\mathcal{N}(i)}\right)}{t}.
\end{align}
In addition, by \eqref{eq-bound-riti} we have 
$r_i\geq r'_i + \frac{1}{t}$.
Combining with \eqref{eq-dim-vasti} and \eqref{eq-ref-ddimui}, we have
\begin{align*}
&\ddim\left(V^{\ast}_i\oplus U_i\right) = \ddim(V^{\ast}_i) +\ddim\left(U_i\right) = \left(r'_i + \frac{1}{t}\right)\ddim\left(\mathcal{H}_{\mathcal{N}(i)}\right)
\leq r_i\ddim\left(\mathcal{H}_{\mathcal{N}(i)}\right).
\end{align*}
Combining with $\ddim(V^{\star}_i) = r_i \ddim\left(\mathcal{H}_{\mathcal{N}(i)}\right) - \ddim\left(V^{\ast}_i\oplus U_i\right)$,
we have $\ddim(V^{\star}_i)\geq 0.$
In summary, $\ddim(V^{\star}_i)$ is a nonnegative integer for each $i\in [m]$.

In the next, we show that $(V_1,V_2,\cdots,V_m)\sim (G_B,\vec{r},\vec{d})$. 
By the definition of $(V_1,V_2,\cdots,V_m)$,
we have $G_B$ is an interaction graph of $(V_1,V_2,\cdots,V_m)$ immediately. In addition, 
if $r'_i=0$,
we have 
$$\R(V_i,\mathcal{H}_{[n]}) = \R\left(V^{\star}_i\otimes \mathcal{H}_{[n]\setminus \mathcal{N}(i)},\mathcal{H}_{[n]}\right) = \R\left(V^{\star}_i,\mathcal{N}(i)\right) = r_i.$$
If $r'_i>0$,
by the definitions of $V_i^\ast,U_i$ and $V^{\star}_i$, we have $V_i^\ast,U_i,V^{\star}_i$
are orthogonal each other.
Thus, we have 
\begin{equation*}
\begin{aligned}
\ddim\left(V^{\ast}_i\oplus U_i\oplus V^{\star}_i\right) &= \ddim\left(V^{\ast}_i\right) +  \ddim\left( U_i\right) +  \ddim\left( V^{\star}_i\right) \\&= r_i \ddim\left(\mathcal{H}_{\mathcal{N}(i)}\right) - \ddim\left(V^{\ast}_i\oplus U_i\right) + \ddim\left( U_i\right) +  \ddim\left( V^{\star}_i\right) 
\\&= r_i \ddim\left(\mathcal{H}_{\mathcal{N}(i)}\right).
\end{aligned}
\end{equation*}
Therefore,
\begin{align*}
    \R(V_i,\mathcal{H}_{[n]}) &= \R\left(\left(V^{\ast}_i\oplus U_i\oplus V^{\star}_i\right)\otimes \mathcal{H}_{[n]\setminus \mathcal{N}(i)},\mathcal{H}_{[n]}\right) = \R\left(V^{\ast}_i\oplus U_i\oplus V^{\star}_i,\mathcal{H}_{\mathcal{N}(i)}\right) \\&= \frac{\ddim\left(V^{\ast}_i\oplus U_i\oplus V^{\star}_i\right)}{\ddim\left(\mathcal{H}_{\mathcal{N}(i)}\right)}= \frac{r_i\ddim\left(\mathcal{H}_{\mathcal{N}(i)}\right)}{\ddim\left(\mathcal{H}_{\mathcal{N}(i)}\right)}  = r_i.
\end{align*}
Thus we have $\R((V_1,V_2,\cdots,V_m),\mathcal{H}_{[n]}) = \vec{r}$.
Moreover, one can verify that $\ddim(\mathcal{H}_{1},\cdots,\mathcal{H}_{n}) = \vec{d}$.
In summary, we have $(V_1,V_2,\cdots,V_m)\sim (G_B,\vec{r},\vec{d})$.

At last, we verify that $(V_1,V_2,\cdots,V_m)$ spans the whole space $\mathcal{H}_{[n]}$.
Then this lemma is proved.
Define
$$L = \{i\in [m]\mid r'_i>0\},\quad T = \bigcup_{i\in L}\mathcal{N}(i).$$
Recall that the instance $\left(V^{\ast}_1\otimes \mathcal{H}_{[n]\setminus \mathcal{N}(1)}^a,\cdots,V^{\ast}_{m}\otimes \mathcal{H}_{[n]\setminus \mathcal{N}(m)}^a\right)\sim G_B$ has relative dimension $\vec{r}'$ to $\mathcal{H}_{[n]}^a$ and spans  $\mathcal{H}_{[n]}^a$.
Thus, we have $\ddim(V^{\ast}_i) = r'_i = 0$ for each $i\not\in L$ and
\begin{align*}
\mathcal{H}_{[n]}^a\subseteq \sum_{i\in [m]}\left(V^{\ast}_i\otimes \mathcal{H}_{[n]\setminus \mathcal{N}(i)}^a\right).
\end{align*}
Therefore we have 
\begin{align*}
\mathcal{H}_{[n]}^a\subseteq \sum_{i\in [m]}\left(V^{\ast}_i\otimes \mathcal{H}_{[n]\setminus \mathcal{N}(i)}^a\right) 
= \sum_{i\in L}\left(V^{\ast}_i\otimes \mathcal{H}^a_{[n]\setminus \mathcal{N}(i)}\right) \subseteq \sum_{i\in L}\left(V^{\ast}_i\otimes \mathcal{H}_{[n]\setminus \mathcal{N}(i)}\right).
\end{align*}
Moreover, by the definitions of $L$ and $T$,
we have for each $i\in L$, $V^{\ast}_i\otimes \mathcal{H}_{[n]\setminus \mathcal{N}(i)}$ acts trivially on the qudits $\mathcal{H}_{[n]\setminus T}$.
Thus, we have 
\begin{align}\label{eq-hna}
\mathcal{H}_{T}^a\otimes  \mathcal{H}_{[n]\setminus T}\subseteq \sum_{i\in L}\left(V^{\ast}_i\otimes \mathcal{H}_{[n]\setminus \mathcal{N}(i)}\right).
\end{align}
In addition, by the definitions of $L$, $T$ and $U_i$, we have
\begin{equation}\label{eq-sumjint-hbj}
\begin{aligned}
    &\quad\sum_{j\in T}\left(\mathcal{H}^{b}_{j}\otimes \mathcal{H}_{[n]\setminus \{j\}}\right)
    = \quad\sum_{j\in \bigcup_{i\in L}\mathcal{N}(i)}\left(\mathcal{H}^{b}_{j}\otimes \mathcal{H}_{[n]\setminus \{j\}}\right) =
    \sum_{i\in L}\sum_{j\in \mathcal{N}(i)}\left(\mathcal{H}^{b}_{j}\otimes \mathcal{H}_{[n]\setminus \{j\}}\right)
    \\& =\sum_{i\in L}\sum_{j\in \mathcal{N}(i)}\left(\mathcal{H}^{b}_{j}\otimes \mathcal{H}_{\mathcal{N}(i)\setminus \{j\}}\otimes \mathcal{H}_{[n]\setminus \mathcal{N}(i)}\right)
     =\sum_{i\in L}\left(\mathcal{H}_{[n]\setminus \mathcal{N}(i)}\otimes \sum_{j\in \mathcal{N}(i)}\left(\mathcal{H}^{b}_{j}\otimes \mathcal{H}_{\mathcal{N}(i)\setminus \{j\}} \right)\right)
     \\&=\sum_{i\in L}\left(\mathcal{H}_{[n]\setminus \mathcal{N}(i)}\otimes U_i\right).
\end{aligned}
\end{equation}
Combining with \eqref{eq-hna} and \eqref{eq-sumjint-hbj}, we have
\begin{align*}
    \mathcal{H}_{[n]} &= \bigotimes_{j\in [n]}\mathcal{H}_{j} = 
    \mathcal{H}_{[n]\setminus T} \otimes \left(\bigotimes_{j\in T}\left(\mathcal{H}^{a}_{j}\oplus \mathcal{H}^{b}_{j}\right)\right)\\
    &\subseteq \mathcal{H}_{[n]\setminus T}\otimes \mathcal{H}^{a}_{T} +  \mathcal{H}_{[n]\setminus T} \otimes \sum_{j\in T}\left(\mathcal{H}^{b}_{j}\otimes\left(\bigotimes_{k\in T\setminus \{j\}}\left(\mathcal{H}^{a}_{j}\oplus \mathcal{H}^{b}_{j}\right)\right)\right)
    \\&= \mathcal{H}_{[n]\setminus T}\otimes \mathcal{H}^{a}_{T} +\sum_{j\in T}\left(\mathcal{H}^{b}_{j}\otimes\mathcal{H}_{[n]\setminus \{j\}}\right)
    \\&
    \subseteq 
    \left(\sum_{i\in T}\left(V^{\ast}_i\otimes \mathcal{H}_{[n]\setminus \mathcal{N}(i)}\right)\right) + \left(\sum_{j\in T}\left(\mathcal{H}_{[n]\setminus \mathcal{N}(j)}\otimes U_j\right)\right)
    \\&= \sum_{i\in T}\left((V^{\ast}_i\oplus U_i)\otimes \mathcal{H}_{[n]\setminus \mathcal{N}(i)}\right)\subseteq \sum_{i\in [m]}V_i.
\end{align*}
In other words, $(V_1,V_2,\cdots,V_m)$ spans the whole space $\mathcal{H}_{[n]}$,
which finishes the proof.
\end{proof}

Now we can prove \Cref{Shearer'sboundistightforQLLL}.
\begin{proof}[Proof of Theorem \ref{Shearer'sboundistightforQLLL}]
Suppose $\vec{r} = (r_1,\cdots,r_m)$.
Let $G_B' = ([m+1],[n],E'_B)$ where $E'_B = E_B \cup \{(m+1,1),(m+1,2),\cdots,(m+1,n)\}$.
Define $\vec{r}' =(r'_1,\cdots,r'_m,r'_{m+1}) $ as follows.
Let $r'_i$ be defined as \eqref{eq-bound-riti}
for each $i\in [m]$.
Define
\begin{equation*}
\begin{aligned}
    r'_{m+1} &\triangleq
\begin{cases}
I(G_B,(r'_1,\cdots,r'_m)),  &\text{if } (r'_1,\cdots,r'_m) \in \mathcal{I}(G_B),\\
0, &\text{otherwise. }
\end{cases}
\end{aligned}
\end{equation*}
By \Cref{lemma-eqinshearersbound}, we have $ I(G_B,(r'_1,\cdots,r'_m)) \geq 0$ if $(r'_1,\cdots,r'_m) \in \mathcal{I}(G_B)$.
Thus, we always have $r'_{m+1} \geq 0$.
Define
\begin{equation*}
\begin{aligned}
    \delta &= 
\begin{cases}
I(G_B,\vec{r}), \quad \quad \quad &\text{if } \vec{r} \in \mathcal{I}(G_B),\\
0, &\text{otherwise.}
\end{cases}
\end{aligned}
\end{equation*}
Similarly to $r'_{m+1}$, we also have $\delta \geq 0$.
Let $r_{m+1} = r'_{m+1} + 1/t$.

We claim that $\delta \leq r_{m+1}\leq \delta + \epsilon$ and that 
there is some $\vec{d}' = (d'_1,\cdots,d'_n)$ where 
$d'_j\leq t^{8m\cdot 4^{m}}$
for each $j\in [n]$ such that 
the random instance of  
the setting $(G'_B,\vec{r}',\vec{d}')$
spans the whole space.
Thus, by \Cref{lemreducedimensionqudit} we have
for any 
$$\vec{d}\geq \left(n \cdot t^{1+8m4^{m}},\cdots,n \cdot t^{1+8m4^{m}} \right)\geq tn\vec{d}',$$ the random instance of  
the setting $(G'_B,(r_1,\cdots,r_{m+1}),\vec{d})$
spans the whole space.
By \Cref{lem-gbprime-rprime} and $\delta \leq r_{m+1}$, we have
$\Pr\left[\R\left(\bigplus_{V\in \vec{V}} V\right) \in [1 - r_{m+1},1-\delta]\right]=1$
for the random instance $\vec{V}$ of the setting $(G_B,\vec{r},\vec{d})$.
Combining with $r_{m+1}\leq \delta + \epsilon$,
we have 
\begin{align}\label{eq-prrsumv-largerthan-1minusdeltaepsilon}
    \Pr\left[\R\left(\bigplus_{V\in \vec{V}} V\right) \in [1 - \delta-\epsilon,1-\delta]\right]\geq \Pr\left[\R\left(\bigplus_{V\in \vec{V}} V\right) \in [1 - r_{m+1},1-\delta]\right]=1.
\end{align}
If $\vec{r}\not\in \mathcal{I}(G_B)$, we have $\delta = 0$. Then (\ref{Shearer'sboundistightforQLLL-a}) of this theorem is immediate by \eqref{eq-prrsumv-largerthan-1minusdeltaepsilon}.
If $\vec{r}\in \mathcal{I}(G_B)$, we have $\delta = I(G_B,\vec{r})$. Then (\ref{Shearer'sboundistightforQLLL-b}) of this theorem is also immediate by \eqref{eq-prrsumv-largerthan-1minusdeltaepsilon}.
In the following, we prove the claims.
Then the theorem is proved.

At first, we show the claim $\delta \leq r_{m+1}\leq \delta + \epsilon$.
Firstly, if $\vec{r} \not\in \mathcal{I}(G_B)$,
we have $\delta = 0$.
Then $r_{m+1} = r'_{m+1} + 1/t \geq 0 = \delta$.
If $\vec{r} \in \mathcal{I}(G_B)$,
we have $\delta = I(G_B,\vec{r})$.
By \eqref{eq-bound-riti} we have $r'_i\leq r_i$ for each $i\in [m]$.
Then $(r'_1,\cdots,r'_m) \leq \vec{r}$.
We have 
$$I(G_B,(r'_1,\cdots,r'_m))\geq I(G_B,\vec{r})= \delta,$$
because 
$I(G_D,\vec{p}_1)\leq I(G_D,\vec{p}_2)$ for any $G_D$ and any $\vec{p}_2 \leq \vec{p}_1$ where $\vec{p}_1\in  \mathcal{I}(G_D)$.
In addition, by $(r'_1,\cdots,r'_m) \leq \vec{r}$ and $\vec{r} \in \mathcal{I}(G_B)$,
we also have $(r'_1,\cdots,r'_m) \in \mathcal{I}(G_B)$
and then $r'_{m+1} = I(G_B,(r'_1,\cdots,r'_m))$.
Thus, we have 
$$r_{m+1} = r'_{m+1} + 1/t = I(G_B,(r'_1,\cdots,r'_m)) + 1/t \geq \delta + 1/t \geq \delta.$$
In summary, we always have $r_{m+1}\geq \delta$.
In the following, we prove $r_{m+1}\leq \delta + \epsilon$.
By $r_{m+1} = r'_{m+1} + 1/t$,
it is sufficient to show 
$r'_{m+1} \leq \delta + \epsilon - 1/t$.
If $(r'_1,\cdots,r'_m) \not\in \mathcal{I}(G_B)$,
we have $r'_{m+1} = 0\leq \delta + \epsilon - 1/t$.
In the following, we assume $(r'_1,\cdots,r'_m) \in \mathcal{I}(G_B)$.
Then $r'_{m+1} = I(G_B,(r'_1,\cdots,r'_m))$.
Suppose $r'_{m+1}>\delta + \epsilon - 1/t$ for contradiction.
We have 
\begin{align}\label{eq-rprimemplus1largeepsilon}
   r'_{m+1} = I(G_B,(r'_1,\cdots,r'_m)) > \delta + \epsilon - 1/t.
\end{align}
In addition, let $(A_1,\cdots,A_{m})\sim (G_B,(r_1,\cdots,r_m))$ be a set of events such that 
$\Pr\left(\cap_{i\in [m]}\Neg{A_i}\right)=\delta$.
By \Cref{thm:shearer1985problem},
such events $(A_1,\cdots,A_{m})$ exist.
Let $y_1,\cdots,y_m$ be independent random variables 
drawn from [0,1] uniformly where $\{y_1,\cdots,y_m\}$ are independent of $\{A_1,\cdots,A_{m}\}$.
For each $i\in [m]$,
let $A'_i$ be the event that $A_i$ happens and  $y_i\leq r'_i/r_i$.
Thus, for each $i\in [m]$ we have
\begin{equation}\label{eq-probaprimi-probaminusaprime}
    \Pr(A'_i) = \Pr(A_i)\cdot \frac{r'_i}{r_i} = r'_i, \quad \Pr(A_i\setminus A'_i) = \Pr(A_i) - \Pr(A'_i) = r_i - r'_i.
\end{equation}
In addition, one can verify that $(A'_1,\cdots,A'_{m})\sim(G_B,(r_1,\cdots,r_m))$.
Combining with $(r'_1,\cdots,r'_m) \in \mathcal{I}(G_B)$ and \Cref{thm:shearer1985problem},
we have 
\begin{equation}\label{eq-probcapnotaprimei}
    \Pr\left(\bigcap_{i\in [m]}\Neg{A'_i}\right)\geq I\left(G_B,\left(r'_1,\cdots,r'_m\right)\right).
\end{equation}
Thus, we have 
\begin{align*}
    &\quad\Pr\left(\bigcap_{i\in [m]}\Neg{A_i}\right) \\
(\text{by the inclusion–exclusion principle })\quad  &\geq \Pr\left(\bigcap_{i\in [m]}\Neg{A'_i}\right) - \Pr\left(\bigcup_{i\in [m]}\left(A_i\setminus A'_i\right)\right)\\
(\text{by the union bound }) \quad   &\geq \Pr\left(\bigcap_{i\in [m]}\Neg{A'_i}\right) 
    - \sum_{i\in [m]}\Pr\left(A_i \setminus A'_i\right) \\
(\text{by \eqref{eq-probcapnotaprimei} and  \eqref{eq-probaprimi-probaminusaprime}}) \quad     &\geq I\left(G_B,\left(r'_1,\cdots,r'_m\right)\right) - \sum_{i\in [m]}\left(r_i - r'_i\right) 
    \\
(\text{by \eqref{eq-rprimemplus1largeepsilon} and  \eqref{eq-bound-riti}}) \quad    &> \delta + \epsilon - \frac{1}{t} - \frac{2m}{t}\\
\left(\text{by $t = \left\lceil\frac{2m+1}{\epsilon}\right\rceil$}\right) \quad    &\geq \delta
\end{align*}
That is, $\Pr\left(\cap_{i\in [m]}\Neg{A_i}\right) >\delta$,
which is contradictory with $\Pr\left(\cap_{i\in [m]}\Neg{A_i}\right)=\delta$.
In summary, we always have $\delta \leq r_{m+1}\leq \delta + \epsilon$.

In the next, we show the claim that there is some $\vec{d}' = (d'_1,\cdots,d'_n)$ where $d'_j\leq t^{8m\cdot 4^{m}}$
for each $j\in [n]$ such that 
the random instance of  
the setting $(G'_B,\vec{r}',\vec{d}')$
spans the whole space.
At first, we prove $\vec{r}'\not\in \mathcal{I}(G'_B)$.
If $(r'_1,\cdots,r'_m) \in \mathcal{I}(G_B)$, we have $r'_{m+1} =I(G_B,(r'_1,\cdots,r'_m))$.
Thus, by \Cref{def: indpoly} we have 
\begin{align*}
I(G'_B,\vec{r}') &= \sum_{T\in \mathsf{Ind}(G'_B)}(-1)^{|T|}\prod_{i\in T}r'_i =  -r'_{m+1} + \sum_{T\in \mathsf{Ind}(G_B)}(-1)^{|T|}\prod_{i\in T}r'_i \\
&= -I(G_B,(r'_1,\cdots,r'_m)) + I(G_B,(r'_1,\cdots,r'_m)) = 0.
\end{align*}
Combining with \Cref{lemma-eqinshearersbound}, we have $\vec{r}'\not\in \mathcal{I}(G'_B)$.
If $(r'_1,\cdots,r'_m) \not\in \mathcal{I}(G_B)$,
we have $r'_{m+1} = 0$ and then $\vec{r}'= (r'_1,\cdots,r'_m,0)$. 
Combining with $(r'_1,\cdots,r'_m) \not\in \mathcal{I}(G_B)$ and the definition of $G'_B$,
we also have $\vec{r}'= (r'_1,\cdots,r'_m,0)\not\in \mathcal{I}(G'_B)$.

Let $q_{i}$ be the minimum positive integer $k$ such that $kr'_{i}$ is an integer for each $i\in [m+1]$. We show that
\begin{equation}\label{eq-prod-qi}
\prod_{i\in [m+1]}q_{i} \leq t^{2m}
\end{equation}
By \eqref{eq-bound-riti}, we have $q_{i}\leq t$ for each $i\in [m]$.
In the next, we show $q_{m+1} \leq t^{m}$.
Then \eqref{eq-prod-qi} is immediate.
If $(r'_1,\cdots,r'_m) \in \mathcal{I}(G_B)$, 
we have $r'_{m+1} = 
I(G_B,(r'_1,\cdots,r'_m))$.
Combining with \Cref{def: indpoly} we have 
\begin{align*}
r'_{m+1}t^{m}  = 
I(G_B,(r'_1,\cdots,r'_m))t^{m}
= \sum_{T\in \mathsf{Ind}(G_B)}(-1)^{\abs{T}}\prod_{i\in T}\max\left\{\left\lfloor tr_i -1\right\rfloor,0\right\}\prod_{i\in [m]
\setminus T}t.
\end{align*}
Combining with $t$ is an integer, we have $r'_{m+1}t^{m}$ is also an integer.
If $(r'_1,\cdots,r'_m) \not\in \mathcal{I}(G_B)$, we have $r'_{m+1} = 0$,
then $r'_{m+1}t^{m}$ is also an integer.
Thus, we always have $q_{m+1} \leq t^{m}$.
Combining with $q_{i}\leq t$ for each $i\in [m]$,
\eqref{eq-prod-qi} is proved.

If $(r'_1,\cdots,r'_{m+1}) \in (0,1]^{m+1}$, combining $\vec{r}'\not\in \mathcal{I}(G'_B)$, \eqref{eq-prod-qi} with \Cref{shearersboundistight:above} (\ref{item-a-thmshearersboundistight}),
we have the there is some $\vec{d}' = (d'_1,\cdots,d'_n)$ where 
$$d'_j\leq \prod_{i\in [m+1]}q_{i}^{4^{m+1}} \leq t^{2m\cdot 4^{(m+1)}} = t^{8m\cdot 4^{m}}$$
for each $j\in [n]$ such that 
the random instance of  
the setting $(G'_B,\vec{r}',\vec{d}')$
spans the whole space.
If $r'_i = 0$ for some $i\in [m+1]$,
by applying \Cref{shearersboundistight:above} (\ref{item-a-thmshearersboundistight})
to the subgraph of $G'_B$ induced by the left vertices $\{i\in [m+1]\mid r'_i>0\}$ and the right vertices $[n]$,
one can also verify that such $\vec{d}'$ exists.
Then the claims are proved and the theorem is immediate.
\end{proof}

\section{CLLL: Beyond Shearer's Bound}\label{sec:tightregionforclll}
In this section, we study the interior of commuting LLL. 

\begin{definition}[Commuting Interior]\label{defcommutinginterior}
The \emph{commuting interior} of an interaction bipartite graph $\GBipartiteGraph=([m],[n],E_B)$,
denoted by $\CInterval(\GBipartiteGraph)$, is the set $\{$rational $\vec{r}\in [0,1)^m$: $\RR\big(\bigplus_{\Subspace\in \SubspaceSet} \Subspace \big)<1$  for any commuting subspace set $\SubspaceSet\sim \GBipartiteGraph$ with $\RR(\SubspaceSet) = \vec{r}\}$. 
\end{definition}
The following is a basic property of the commuting interior.
\begin{prop}\label{le:monotonicity} If $\vec{r}\in\CInterval(G_B)$, then $\vec{r}'\in \CInterval(G_B)$ for any rational $\vec{r}'\leq \vec{r}$.
\end{prop}
\begin{proof} 
Without loss of generality, we assume that $\vec{r}' = (r_1-\epsilon,r_2,\cdots,r_m)$. By contradiction, suppose that there exists a commuting subspace set $\SubspaceSet'\sim \GBipartiteGraph$ with $\RR(\SubspaceSet') = \vec{r}'$ such that $\RR\big(\bigplus_{\Subspace'\in \SubspaceSet'} \Subspace' \big) = 1$. In the following, we will construct another commuting subspace set $\SubspaceSet\sim \GBipartiteGraph$ with $\RR(\SubspaceSet) = \vec{r}$ such that $\RR\big(\bigplus_{\Subspace\in \SubspaceSet} \Subspace \big) = 1$, which is contradicted to the condition $\vec{r}\in\CInterval(G_B)$. 

W.l.o.g., we assume the edge $(1,1)$ exists in $G_B$. Since $\vec{r}$ and $\vec{r}'$ are both rational vectors, $\frac{\epsilon}{1-r'_1}$ is a rational number. Suppose $\frac{\epsilon}{1-r'_1} = \frac{a}{b}$ where $a$ and $b$ are integers. Let $\mathcal{H}'_1,\cdots,\mathcal{H}'_n$ denote the qudits that $\SubspaceSet'=\{V'_1,\cdots,V'_m\}$ acts on. Define $\mathcal{H}_1 = \mathcal{H}'_1 \otimes \mathcal{H}_1^c$ where $\ddim{(\mathcal{H}_1^c)} = b$, and $\mathcal{H}_j = \mathcal{H}'_j$ for any $2\leq j\leq n$. We construct a subspace set $\SubspaceSet=\{V_1,\cdots,V_m\}$ acting on the qudits $\mathcal{H}_1,\cdots,\mathcal{H}_n$ as follows. Define $V_1 =V'_1\otimes \mathcal{H}_1^c+\mathcal{H}_{[n]}'\otimes W$ where $W$ is an arbitrary $a$-dimensional subspace of $\mathcal{H}_1^c$. For each $2\leq i\leq m$, let $V_i = V'_i\otimes \mathcal{H}_1^c $. We can easily check that $\SubspaceSet$ is commuting, $\RR(\SubspaceSet)=\vec{r}$, and $\RR\big(\bigplus_{\Subspace\in \SubspaceSet} \Subspace \big) = 1$.
\end{proof}

As shorthand, we write $\Interior(G_B)$ for $\Interior(G_D(G_B))$. According to Theorem 1 in \cite{pnas} (see also Theorem \ref{thm:pnas}), $\CInterval$ contains all rational vectors in the Shearer's bound. That is, for any rational $\vec{r}$, if $\vec{r}\in \Interior(G_B)$, then $\vec{r}\in \CInterval(G_B)$. 
Here, we particularly care about whether $\CInterval(G_B)\subseteq\Interior(G_B)$, i.e., whether Shearer's bound is tight for CLLL.
\begin{definition}[Gap]\label{defgap}
An interaction bipartite graph $\GBipartiteGraph$ is called \emph{gapless for CLLL} if $\CInterval(G_B)\subseteq\Interior(G_B)$, otherwise it is called \emph{gapful}. Similarly, we can also define \emph{gapless/gapful} for VLLL. We do not mention ``for CLLL" or ``for VLLL" if it is clear from context.
\end{definition}

The rest of Section \ref{sec:tightregionforclll} is organized as follows. Section \ref{CLLLeualVLLL} shows that CLLL equals VLLL on a large class of bipartite graphs. Section \ref{sec:gapdecision} provides a sufficient and necessary condition for gap existence. Section \ref{sec:rules} develops a set of reduction rules for inferring gap existence of a bipartite graph from known ones. Finally, based on these results, Section \ref{sec:gapexistence} provides an almost complete characterization of gapless/gapful bipartite graphs.


\subsection{Solitary Qudits Are Classical}\label{CLLLeualVLLL}
The main result of this subsection is Theorem \ref{thm:solitaryisclassical}, which says that solitary qudits can be restricted to be classical without changing the interior. As a corollary, for $G_B$ where all right vertices are solitary, the interior of CLLL equals that of VLLL (Theorem \ref{thm:2-loc}).
\begin{definition}[Solitary Qudits]
Given an interaction bipartite graph $G_B=([m],[n],E_B)$, we say a right vertex $j\in[n]$ is \emph{solitary} if for any $i_1,i_2\in \mathcal{N}(j)$ we have $\mathcal{N}(i_1)\cap\mathcal{N}(i_2)=\{j\}$.
\end{definition}


\begin{definition}[Classical Qudits]
Let $\mathcal{H}$ be a qudit and $\{|\ell \rangle:\ell\in [d]\}$ be its computational basis. $\mathcal{H}$ is said to be classical if every Hamiltonian $H$ commutes with $|\ell\rangle\langle\ell|$ for any $\ell\in[d]$. 
\end{definition}

\begin{thm}\label{thm:solitaryisclassical}
	Given an interaction bipartite graph $G_B=([m],[n],E_B)$, if there exists a commuting subspace set $\SubspaceSet\sim G_B$ spanning the whole space, then there exists another commuting subspace set $\SubspaceSet'\sim G_B$ with $\RR(\SubspaceSet') = \RR(\SubspaceSet)$ spanning the whole space and satisfying that for each solitary $j\in[n]$, $\mathcal{H}_j$ is classical.
\end{thm}

The idea of Theorem \ref{thm:solitaryisclassical} is to dissect the structure of commuting local Hamiltonians by using Bravyi and Vyalyi's Structure Lemma~\cite{bravyi2005commutative}.
\begin{lem}[Structure Lemma, adapted from~\cite{bravyi2005commutative}]\label{lem:structure}
	Suppose $\mathcal{X}$, $\mathcal{Y}$, $\mathcal{Z}$ are complex Euclidean spaces, $\Pi_V$ and $\Pi_W$ are projection operators acting on $\mathcal{X}\otimes \mathcal{Y}$ and $\mathcal{Y}\otimes \mathcal{Z}$ respectively. If $\Pi_{V}$ and $\Pi_{W}$ commute, then $\mathcal{Y}$ can be decomposed to some orthogonal subspaces $\mathcal{Y}=\bigoplus_\ell\mathcal{Y}_\ell=\bigoplus_\ell \mathcal{Y}_{\ell 1}\otimes\mathcal{Y}_{\ell 2}$ such that for any $\ell$:
	\begin{enumerate}
		\item $\Pi_V$ and $\Pi_W$ preserve $\mathcal{Y}_\ell$;
		\item Restricted to $\mathcal{Y}_\ell$, $\Pi_{V}$ and $\Pi_{W}$ act non-trivially only on $\mathcal{Y}_{\ell1}$ and $\mathcal{Y}_{\ell2}$ respectively.
	\end{enumerate}
	In other words, $V$ can be decomposed as $V=\bigoplus_{\ell} V_{\ell}\otimes\mathcal{Y}_{\ell 2}$ where $V_{\ell}\subseteq \mathcal{X}\otimes\mathcal{Y}_{\ell 1}$, and $W$ can be decomposed as $W=\bigoplus_{\ell} W_{\ell}\otimes\mathcal{Y}_{\ell 1}$ where $W_{\ell}\subseteq \mathcal{Y}_{\ell 2}\otimes\mathcal{Z}$.
\end{lem}

The following lemma will be also used in the proof of Theorem \ref{thm:solitaryisclassical}.
\begin{lem}\label{lem:whethercontainproduct}
Suppose $\mathcal{X}_1$, $\mathcal{Y}_1$, $\mathcal{X}_2$, and $\mathcal{Y}_2$ are complex Euclidean spaces. $V_1$ is a subspace of $\mathcal{X}_1\otimes \mathcal{Y}_1$ satisfying that $V_1\not\supseteq W_1\otimes \mathcal{Y}_1$ for any nonempty subspace $W_1\subseteq \mathcal{X}_1$. $V_2$ is a subspace of $\mathcal{X}_2\otimes \mathcal{Y}_2$ satisfying that $V_2\not\supseteq W_2\otimes \mathcal{Y}_2$ for any nonempty subspace $W_2\subseteq \mathcal{X}_2$. Then  $V_1\otimes \mathcal{X}_2\otimes \mathcal{Y}_2+V_2\otimes \mathcal{X}_1\otimes \mathcal{Y}_1\not\supseteq W\otimes \mathcal{Y}_1\otimes \mathcal{Y}_2$ for any nonempty subspace $W\subseteq\mathcal{X}_1\otimes\mathcal{X}_2$.
\end{lem}
\begin{proof}
We first introduce some notations. For a vector $|v\rangle \in \mathcal{X}\otimes\mathcal{Y}$, we define $S_{\mathcal{X}}(|v\rangle)$ as the span of its Schmidt bases for $\mathcal{X}$. In other words, if $|v\rangle=\sum_{i}|i_x\rangle\otimes |i_y\rangle$ is the Schmidt decomposition, then $S_{\mathcal{X}}(|v\rangle)=\text{Span}(|1_x\rangle,|2_x\rangle,\cdots)$. For a subspace $V\subseteq\mathcal{X}\otimes\mathcal{Y}$, we define $S_{\mathcal{X}}(V):=\text{Span}(S_{\mathcal{X}}(|v\rangle):|v\rangle\in V)$ as the span of $S_{\mathcal{X}}(|v\rangle)$ over all $|v\rangle\in V$.
\begin{claim}
$V\not\supseteq W\otimes \mathcal{Y}$ for any nonempty subspace $W\subseteq \mathcal{X}$ if and only if $S_{\mathcal{X}}(V^\perp)=\mathcal{X}$.
\end{claim}
\begin{proof}
$\Longleftarrow:$ Suppose that $V\supseteq W\otimes \mathcal{Y}$ for some nonempty $W\subseteq \mathcal{X}$. Then $V^\perp\subseteq W^{\perp} \otimes \mathcal{Y}$ and $S_{\mathcal{X}}(|v\rangle)\subseteq W^\perp$ for any $|v\rangle\in V^\perp$. So $S_{\mathcal{X}}(V^\perp)\subseteq W^\perp \subsetneq \mathcal{X}$.

$\Longrightarrow:$ Suppose $S_{\mathcal{X}}(V^\perp)\subseteq W^\perp$ for some nonempty $W\subseteq \mathcal{X}$. Then $S_{\mathcal{X}}(|v\rangle)\subseteq W^\perp$ for any $|v\rangle\in V^\perp$. Obviously, we have $V^\perp\subseteq W^{\perp} \otimes \mathcal{Y}$, which implies that  $V\supseteq W\otimes \mathcal{Y}$.
 \end{proof}

Note that $(V_1\otimes \mathcal{X}_2\otimes \mathcal{Y}_2+V_2\otimes \mathcal{X}_1\otimes \mathcal{Y}_1)^\perp=V_1^\perp\otimes V_2^\perp$. Then it is easy to check that $S_{\mathcal{X}_1\otimes\mathcal{X}_2}(V_1^\perp\otimes V_2^\perp)=S_{\mathcal{X}_1}(V_1^\perp)\otimes S_{\mathcal{X}_2}(V_2^\perp)$, which is $\mathcal{X}_1\otimes\mathcal{X}_2$ according to the above claim. Finally, we can conclude the proof by applying the above claim again.
\end{proof}

\begin{proof}[Proof of Theorem \ref{thm:solitaryisclassical}]
Fix an arbitrary solitary $j\in[n]$ and w.l.o.g. assume $\mathcal{N}(j) = [k]$. Let $\mathcal{H}_1,\cdots,\mathcal{H}_n$ denote the qudits that $\SubspaceSet=\{V_1,\cdots,V_m\}$ acts on. By applying Lemma \ref{lem:structure} iteratively, 
we decompose $\mathcal{H}_j$ to some orthogonal subspaces $\mathcal{H}_j=\bigoplus_\ell\mathcal{H}_{j\ell}=\bigoplus_\ell \mathcal{H}_{j\ell 1}\otimes\mathcal{H}_{j\ell 2}\otimes\cdots\otimes\mathcal{H}_{j\ell k}$ such that: for each $i\in[k]$, $V_i$ can be decomposed as 
$$V_i=\bigoplus_\ell V_{i\ell}=\bigoplus_{\ell} V^{\lloc}_{i\ell}\otimes\left(\bigotimes_{i'\neq i}\mathcal{H}_{j\ell i'}\right)\otimes \mathcal{H}_{[n]\setminus \mathcal{N}(i)}$$
where $V^{\lloc}_{i\ell}\subseteq \mathcal{H}_{\mathcal{N}(i)\setminus\{j\}}\otimes\mathcal{H}_{j\ell i}$. Furthermore, we can always decompose $V_{i\ell}^{\lloc}$ to two orthogonal subspaces $V_{i\ell}^{\lloc}=V_{i\ell,1}^{\lloc}\oplus V_{i\ell,2}^{\lloc}$ such that
\begin{itemize}
\item[(i)] $V_{i\ell,1}^{\lloc}=W \otimes\mathcal{H}_{j\ell i} \text{ for some nonempty subspace $W \subseteq \mathcal{H}_{\mathcal{N}(i)\setminus\{j\}}$}$, and
\item[(ii)] $V_{i\ell,2}^{\lloc}\not\supseteq W\otimes\mathcal{H}_{j\ell i}$ for any nonempty subspace $W\subseteq \mathcal{H}_{\mathcal{N}(i)\setminus\{j\}}$.
\end{itemize}  
We define
\begin{align*}
V_{i\ell,1}:=V^{\lloc}_{i\ell,1}\otimes\left(\bigotimes_{i'\neq i}\mathcal{H}_{j\ell i'}\right)\otimes \mathcal{H}_{[n]\setminus \mathcal{N}(i)} \quad
\mbox{ and }\quad  
V_{i\ell,2}:=V_{i\ell,2}^{\lloc}\otimes\left(\bigotimes_{i'\neq i}\mathcal{H}_{j\ell i'}\right)\otimes \mathcal{H}_{[n]\setminus \mathcal{N}(i)}.
\end{align*}

We construct a subspace set $\SubspaceSet'=\{V_1',\cdots,V_m'\}\sim G_B$ as follows. For each $i\in[k]$, define $V_i'=\sum_{\ell} V_{i\ell,1}$. For each $i\notin [k]$, define $V_i'=V_i$. Note that restricted to $\mathcal{H}_{j\ell}$, $V_{i\ell,1}$ acts non-trivially only on the qudit $\mathcal{H}_{\mathcal{N}(i)\setminus\{j\}}$. Thus $\SubspaceSet'$ uses $\mathcal{H}_{j}$ as a classical variable (rotate the computational basis if needed).

\begin{claim} $V_{i_1}'$ and $V_{i_2}'$ commute for any $i_1, i_2\in[m]$.



\end{claim}
\begin{proof} We only need to prove the claim in the following two cases.

{\it Case I: $i_1\neq i_2\in [k]$}. It suffices to show that $V_{i_1\ell_1,1}$ and $V_{i_2\ell_2,1}$ commute for any $\ell_1$ and $\ell_2$. If $\ell_1\neq \ell_2$, then $V_{i_1\ell_1,1}\subseteq \mathcal{H}_{j\ell_1}\otimes \mathcal{H}_{[n]\setminus\{j\}}$ and $V_{i_2\ell_2,1}\subseteq \mathcal{H}_{j\ell_2}\otimes \mathcal{H}_{[n]\setminus\{j\}}$ are orthogonal and thus commute. If $\ell_1=\ell_2=\ell$, then restricted to $\mathcal{H}_{j\ell}$, $V_{i_1\ell,1}$ and $V_{i_2\ell,1}$ act non-trivially on $\mathcal{H}_{\mathcal{N}(i_1)\setminus\{j\}}$ and $\mathcal{H}_{\mathcal{N}(i_2)\setminus\{j\}}$ respectively. Noticing that $\left(\mathcal{N}(i_1)\setminus\{j\}\right)\cap\left(\mathcal{N}(i_1)\setminus\{j\}\right)=\emptyset$ since $j$ is solitary, we have that $V_{i_1\ell,1}$ and $V_{i_2\ell,1}$ commute. 

{\it Case II: $i_1\in[k]$ and $i_2\notin[k]$}. Let $\Pi_{\ell}$, $\Pi_{i_1\ell}$, and $\Pi_{i_1\ell,1}$ denote the projectors on $\mathcal{H}_{j\ell}\otimes\mathcal{H}_{[n]\setminus\{j\}}$, $V_{i_1\ell}$, and $V_{i_1\ell,1}$ respectively. Noting that $ \Pi_{i_1\ell}=\Pi_{V_{i_1}}\cdot\Pi_{\ell}$ and $\Pi_{\ell}$ acts non-trivially only on the qudit $\mathcal{H}_j$, we have
$$\Pi_{i_1\ell}\cdot \Pi_{V_{i_2}}=\Pi_{V_{i_1}}\cdot\Pi_{\ell}\cdot \Pi_{V_{i_2}}=\Pi_{V_{i_1}}\cdot \Pi_{V_{i_2}}\cdot \Pi_{\ell}=\Pi_{V_{i_2}}\cdot\Pi_{V_{i_1}}\cdot \Pi_{\ell}=\Pi_{V_{i_2}}\cdot\Pi_{i_1\ell},$$
i.e., $\Pi_{i_1\ell}$ and $\Pi_{V_{i_2}}$ commute.
Let $\{|b\rangle\}_b$ denote the computational basis of $\mathcal{H}_{j\ell i_1}$. A key observation is that $\Pi_{i_1\ell,1}$ is the product of $\langle b| \Pi_{i_1\ell}|b\rangle \otimes \mathsf{id}$ over all $|b\rangle$'s, where $\mathsf{id}$ denotes the identity operator on $\mathcal{H}_{j\ell i_1}$. Since $\Pi_{i_1\ell}$ and $\Pi_{V_{i_2}}$ commute, every $\langle b| \Pi_{i_1\ell}|b\rangle \otimes \mathsf{id}$ commutes with $\Pi_{V_{i_2}}$. Then $\Pi_{i_1\ell,1}$ also commutes with $\Pi_{V_{i_2}}$. Thus $\Pi_{V_{i_1}'}=\sum_{\ell} \Pi_{i_1\ell,1}$ and $\Pi_{V_{i_2}'}=\Pi_{V_{i_2}}$ commute.
\end{proof}
\begin{claim}
$\SubspaceSet'$ spans the whole space.
\end{claim}
\begin{proof}
Let $W_{i\ell}$ be the subspace of $\mathcal{H}_{\mathcal{N}(i)\setminus\{j\}}$ such that $V_{i\ell,1}^{\lloc}=W_{i\ell}\otimes\mathcal{H}_{j\ell i}$.
For simplicity of notations, we let $$U_\ell=\sum_{i\in[k]}W_{i\ell} \otimes \mathcal{H}_{[n]\setminus \mathcal{N}(i)}+\sum_{i\notin [k]} V_i^{\lloc}\otimes \mathcal{H}_{[n]\setminus(\mathcal{N}(i)\cup\{j\})},$$
Note that $\sum_{i=1}^m V'_i=\bigoplus_{\ell} U_\ell\otimes \mathcal{H}_{j\ell}$ and $\sum_{i=1}^m V_i =\bigoplus_{\ell}\left(U_\ell\otimes \mathcal{H}_{j\ell}+\sum_{i=1}^k V_{i\ell,2}\right)$. 

By contradiction, we assume that $\SubspaceSet'$ does not span the whole space, i.e., there exists some $\ell$ such that $U_\ell\subsetneq \mathcal{H}_{[n]\setminus \{j\}}$. In the next paragraph, we will show that $U_\ell\otimes \mathcal{H}_{j\ell}$ and $\sum_{i=1}^k V_{i\ell,2}$ commute. 
Then, because $U_\ell\otimes \mathcal{H}_{j\ell}+\sum_{i=1}^k V_{i\ell,2}=\mathcal{H}_{[n]\setminus\{j\}}\otimes\mathcal{H}_{j\ell}$, we have $\sum_{i=1}^k V_{i\ell,2}\supseteq U_\ell^\perp\otimes \mathcal{H}_{[n]\setminus \{j\}}$. However, as $\{V_{i\ell,2}\}_{i=1}^k$ act non-trivially on disjoint qudits, by applying Lemma \ref{lem:whethercontainproduct} iteratively, we have that $\sum_{i=1}^k V_{i\ell,2}\not\supseteq U\otimes \mathcal{H}_{j\ell}$ for any nonempty subspace $U\subseteq \mathcal{H}_{[n]\setminus\{j\}}$. A contradiction.

To show that $U_\ell\otimes \mathcal{H}_{j\ell}$ and $\sum_{i=1}^k V_{i\ell,2}$ commute, it suffices to show that for any $i\in[k]$ and $i'\in[m]$,
\begin{itemize}
    \item[(a)] $W_{i'\ell}\otimes \mathcal{H}_{[n]\setminus \mathcal{N}(i')}$ commute with $V_{i\ell,2}$ if $i'\in[k]$; and 
    \item[(b)] $V_{i'}^{\lloc}\otimes \mathcal{H}_{[n]\setminus(\mathcal{N}(i')\cup\{j\})}$ commute with $V_{i\ell,2}$ if $k+1\leq i'\leq m$.
\end{itemize}
Case (a) is obvious because restricted to $\mathcal{H}_{j\ell}$, $W_{i'\ell} \otimes \mathcal{H}_{[n]\setminus \mathcal{N}(i')}$ and $V_{i\ell,2}$ act non-trivially on disjoint set of qudits, namely $\mathcal{H}_{\mathcal{N}(i')\setminus \{j\}}\otimes \mathcal{H}_{j\ell i'}$ and $\mathcal{H}_{\mathcal{N}(i)\setminus \{j\}}\otimes \mathcal{H}_{j\ell i}$ respectively. For Case (b), the proof of the previous claim presents that both $V_{i\ell}$ and $V_{i\ell,1}$ commute with $V_{i'}$ for any $k\leq i'\leq m$. So $V_{i\ell,2}$ also commutes with $V_{i'}$, as $V_{i\ell,2}$ is the orthogonal complement of $V_{i\ell,1}$ in $V_{i\ell}$. Let $\Pi_{\ell}$ denote the projector on $\mathcal{H}_{j\ell}\otimes \mathcal{H}_{[n]\setminus\{j\}}$, then the projector projecting on $V_{i'}^{\lloc}\otimes \mathcal{H}_{[n]\setminus(\mathcal{N}(i')\cup\{j\})}$ is $\Pi_{V_{i'}}\Pi_\ell$. Now it is obvious that $V_{i'}^{\lloc}\otimes \mathcal{H}_{[n]\setminus(\mathcal{N}(i')\cup\{j\})}$ and $V_{i\ell,2}$ commute by noting that $\Pi_\ell$ and $V_{i\ell,2}$ commute. 
\end{proof}

Note that $\RR(\SubspaceSet')$ may be smaller than $\RR(\SubspaceSet)$. We can easily address this problem while keeping $\mathcal{H}_j$ still classical by playing the trick used in the proof of Proposition \ref{le:monotonicity}. 
Finally, we apply the above procedure to each solitary qudits one by one in an arbitrary order. It is straightforward to verify that a classical qudit keeps classical when applying the above procedure to other solitary qudits. Finally, all solitary qudits will be classical. 
\end{proof}

The following theorem is an immediate corollary of Theorem \ref{thm:solitaryisclassical}.
\begin{thm}\label{thm:2-loc} For any interaction bipartite graph $G_B$ where all the right vertices are solitary, $\CInterval(G_B)$ is the set consisting of all rational vectors in $\VInterval(G_B)$.
\end{thm}
In particular, if $G_B$ is a cycle of length at least 6, i.e., $G_B$ is isomorphic to a bipartite graph $([\ell],[\ell],E_\ell)$ where $E_\ell=\{(i,i),(i,i+1):i\in[\ell-1]\}\cup\{(\ell,\ell),(\ell,1)\}$ and $\ell\geq 3$, then $G_B$ is gapful for VLLL \cite{he2017variable}. By Theorem \ref{thm:2-loc}, $G_B$ is also gapful for CLLL.


\subsection{A Sufficient and Necessary Condition for Gap Decision}\label{sec:gapdecision}

The main result of this subsection is Theorem \ref{Conj:GapGeom}, a sufficient and necessary condition for deciding whether Shearer's bound is tight for CLLL on a given interaction bipartite graph. Theorem \ref{Conj:GapGeom} enables to decide whether a gap existence without computing the Shearer's bound. In particular, Theorem \ref{Conj:GapGeom} will be used in the proof of Theorem \ref{thm:reductionrules}, which is about reduction rules.

\begin{definition}[Quantum Exclusiveness]
A subspace set $\SubspaceSet=\{V_1,\cdots,V_m\}$ is called \emph{exclusive} with respect to an interaction bipartite graph $\GBipartiteGraph$, if $\SubspaceSet \sim \GBipartiteGraph$ and $\Subspace_{i}\perp \Subspace_{i'}$ for any edge $(i,i')$ in the base graph. We will omit ``with respect to $\GBipartiteGraph$" if it is clear from the context. Note that an exclusive subspace set must be commuting.
\end{definition}

\begin{thm}\label{Conj:GapGeom}
Given a connected interaction bipartite graph $\GBipartiteGraph$, the following two conditions are equivalent:
\begin{itemize}
\item[(a)] For any rational $\vec{r}\in\CInterval(\GBipartiteGraph)$, there is an exclusive subspace set with interaction bipartite graph $\GBipartiteGraph$ and relative dimension vector $\vec{r}$.
\item[(b)] $\GBipartiteGraph$ is gapless for CLLL.
\end{itemize}
\end{thm}

Theorem \ref{Conj:GapGeom} is in fact a leveraging of Theorem 5 in \cite{he2017variable}, which is for VLLL, to CLLL. The following two lemmas, Lemma \ref{le:exclusionisworst} and \ref{le:WorstPlacement}, about exclusive classical event sets will be used.

\begin{definition}[Classical exclusiveness]
A classical event set $\EventSet$ is said to be exclusive with respect to a graph $G_D$, if $G_D$ is a dependency graph of $\EventSet$ and $\Pr(A_i\cap A_{i'})=0$ for any edge $(i,i')$ in $G_D$. We do not mention ``with respect to $G_D$" if it is clear from context.
\end{definition}

The \emph{abstract boundary} of a graph $\GDependencyGraph$, denoted by $\Boundary(\GDependencyGraph)$, is the set $\{\vec{p}: (1-\epsilon)\vec{p}\in \Interior(\GDependencyGraph) \textrm{ and } (1+\epsilon)\vec{p}\notin \Interior(\GDependencyGraph) \textrm{ for any }\epsilon\in (0,1)\}$.

\begin{lem}[Theorem 1 in \cite{shearer1985problem}]\label{le:exclusionisworst}
Given a graph $\GDependencyGraph$ and $\vec{p}\in\Interior(\GDependencyGraph)\cup\Boundary(\GDependencyGraph)$. Among all event sets $\EventSet\sim\GDependencyGraph$ with $\Pr(\EventSet)=\vec{p}$, there is an exclusive one such that $\Pr(\cup_{\Event\in \EventSet} \Event )$ is maximized.
\end{lem}

\begin{lem}[Lemma 29 in \cite{he2017variable}]\label{le:WorstPlacement}
Suppose that $\GDependencyGraph$ is a dependency graph of event sets $\EventSet$ and $\EventSetB$, $\Pr(\EventSet)=\Pr(\EventSetB)$, and $\EventSet$ is exclusive. Then $\Pr(\cup_{\Event\in \EventSet} \Event)\geq\Pr(\cup_{B\in \EventSetB} B)$. Moreover, when $G_D$ is connected, the equality holds if and only if $\EventSetB$ is exclusive.
\end{lem}

The following basic fact about quantum exclusiveness will also be used.
\begin{lem}\label{le:exclusivemonotonicity}
Given an interaction bipartite graph $\GBipartiteGraph=([m], [n], E_B)$ and a rational vector $\vec{r}\in[0,1)^m$, if there is an exclusive subspace set $\SubspaceSet\sim \GBipartiteGraph$ with $\RR(\SubspaceSet) = \vec{r}$, then for any rational $\vec{r'} \leq \vec{r}$, there is an exclusive subspace set $\SubspaceSet'\sim \GBipartiteGraph$ with $\RR(\SubspaceSet') = \vec{r'}$.
\end{lem}
\begin{proof}
Without loss of generality, we assume that $\vec{r'} = (r_1-\epsilon,r_2,\cdots,r_m)$ and the edge $(1,1)$ exists in $G_B$. Since $\vec{r'}$ and $\vec{r}$ are both rational, so is $\epsilon/r_1$. Suppose $\epsilon/r_1 = a/b$ where $a$ and $b$ are integers. Let $\mathcal{H}_1,\cdots,\mathcal{H}_n$ be the qudits that $\SubspaceSet=\{V_1,\cdots,V_m\}$ acts on. Define $\mathcal{H}'_1 = \mathcal{H}_1 \otimes \mathcal{H}_1^c$ where $\dim{(\mathcal{H}_1^c)} = b$, and $\mathcal{H}'_i = \mathcal{H}_i$ for any $i\geq 2$. We construct a subspace set $\SubspaceSet'=\{V_1',\cdots,V_m'\}$ acting on the qudits $\mathcal{H}_1',\cdots,\mathcal{H}_n'$ as follows. Let $V'_1 =V_1\otimes W$ where $W$ is an arbitrary $(b-a)$-dimensional subspace of $\mathcal{H}_1^c$. For each $i\geq 2$, let $V'_i = V_i\otimes \mathcal{H}_1^c$. One can easily check that $\SubspaceSet'\sim \GBipartiteGraph$, $\RR(\SubspaceSet') = \vec{r'}$, and $\SubspaceSet'$ is exclusive.
\end{proof}

Now, we are ready to prove Theorem \ref{Conj:GapGeom}.

\begin{proof}[Proof of Theorem \ref{Conj:GapGeom}] Let $G_D$ denote the base graph of $G_B$.

\underline{$(a)\Rightarrow(b)$}: Given any rational $\vec{r}\in\CInterval(\GBipartiteGraph)$, suppose there is an exclusive subspace set $\SubspaceSet \sim \GBipartiteGraph$ with $\RR(\SubspaceSet)=\vec{r}$. Because $\vec{r}\in\CInterval(\GBipartiteGraph)$, we have $\RR(\bigplus_{\Subspace\in\SubspaceSet}\Subspace)<1$. In the following, we will construct an exclusive classical event set $\EventSet$ with dependency graph $G_D$ and $\Pr(\EventSet)=\vec{r}$ such that $\Pr(\cup_{A\in\EventSet} A)=\RR(\sum_{\Subspace\in\SubspaceSet}\Subspace)<1$. 
By Lemma \ref{le:WorstPlacement}, we have $\Pr(\cup_{B\in\EventSetB} B)<1$ for every event set $\EventSetB\sim G_D$ with $\Pr(\EventSetB)=\vec{r}$. Thus $\vec{r}\in \Interior(\GBipartiteGraph)$, which concludes the proof of this direction.

The event set $\EventSet=(A_1,\cdots,A_m)$ is constructed as follows. Because $\SubspaceSet$ is commuting, $V_1,\cdots,V_m$ can be diagonalized with respect to the same orthonormal basis, denoted by $\{|e_\ell\rangle:\ell\in[t]\}$. Let $(\Omega,\mathcal{F},\Pr)$ be a probability space where $\Omega = \{|e_\ell\rangle:\ell\in[t]\}$,  $\mathcal{F} = 2^\Omega$, and $\Pr(|e_i\rangle) = 1/t$ for each $i$. Define $A_i=\{|e_\ell\rangle:|e_\ell\rangle\in V_i\}$. Obviously, $\Pr(A_i)=\RR(V_i)$ and $\Pr(\cup_{A\in\EventSet} A)=\RR(\sum_{\Subspace\in\SubspaceSet}\Subspace)<1$. Moreover, for any $i\in[m]$ and $S\subseteq [m]\setminus (\Neighbor(G_D,i)\cup\{i\})$, we have $\Pr(A_i\cap (\cup_{i'\in S} A_{i'}))=\RR(V_i\cap (\sum_{i'\in S} V_{i'}))=\RR(V_i)\cdot \RR(\sum_{i'\in S} V_{i'})=\Pr(A_i)\cdot\Pr(\cup_{i'\in S} A_{i'})$, thus $A_i$ is independent with $\{A_{i'}:i'\notin \Neighbor(G_D,i)\cup\{i\}\}$. So $G_D(G_B)$ is a dependency graph of $\EventSet$. Finally, for any edge $(i,i')$ in $G_D$, $\Pr(A_i\cap A_{i'})=\RR(V_i\cap V_{i'})=0$, thus $\EventSet$ is exclusive.

\vspace{1ex}
\noindent\underline{$(b)\Rightarrow(a)$}: We assume that $G_B$ is gapless, i.e., $\CInterval(G_B)\subset \Interior(G_D)$. Given a rational $\vec{r}\in \CInterval(\GBipartiteGraph)$, choose an arbitrary rational vector $\vec{q}\in\Boundary(G_D)$ satisfying that $\vec{q}\geq \vec{r}$. As $\Boundary(G_D)\cap \Interior(G_D)=\emptyset$ and $\CInterval(G_B)\subset \Interior(G_D)$, it has that $\vec{q}\notin\CInterval(G_B)$. By Definition \ref{defcommutinginterior}, there exists a commuting subspace set $\SubspaceSet\sim G_B$ with $\RR(\SubspaceSet)=\vec{q}$ such that $\RR\left(\bigplus_{V\in\SubspaceSet}V \right)=1$. In the following, we will show that $\SubspaceSet$ is exclusive, which then concludes the proof by Theorem \ref{le:exclusivemonotonicity}. 

We define a classical event set $\EventSet=\{A_1,\cdots,A_m\}$ corresponding to $\SubspaceSet$ as in the proof of the direction $(a)\Rightarrow (b)$. By Lemma \ref{le:exclusionisworst}, there is an exclusive event set $\EventSetB\sim G_D$ with $\Pr(\EventSetB)=\vec{q}$. Because $G_D$ is connected and $\Pr(\cup_{A\in\EventSet} A)=\RR(\sum_{\Subspace\in\SubspaceSet}\Subspace)=1\geq \Pr(\cup_{B\in\EventSetB} B)$, we have that $\Event$ is exclusive by Lemma \ref{le:WorstPlacement}.  Finally, noticing that $V_i\perp V_{i'}$ if and only if $\Pr(A_i\cap A_{i'})=0$, we conclude that $\SubspaceSet$ is exclusive.  
\end{proof}


\subsection{Reduction Rules}\label{sec:rules}
To infer gap existence of a bipartite graph from known ones, a set of reduction rules are established for VLLL ~\cite{he2017variable}. With these reduction rules, various bipartite graphs are shown to be gapful/gapless for VLLL. In this subsection, we leverage these reduction rules to CLLL. 

We consider the following six types of operations on an interaction bipartite graph  $G_B=([m],[n],E_B)$.
\begin{itemize}
\item Delete-$R$-Leaf: Delete a vertex $j\in [n]$ on the right side with $|\mathcal{N}(j)|\leq 1$, and remove the incident edge if any.
\item Duplicate-$L$-Vertex: Given a vertex $i\in[m]$ on the left side, add a vertex $i'$ to the left side, and add edges incident to $i'$ so that $\mathcal{N}(i')=\mathcal{N}(i)$.
\item Duplicate-$R$-Vertex:  Given a vertex $j\in [n]$ on the right side, add a vertex $j'$ to the right side, and add some edges incident to $j'$ so that $\Neighbor(j')\subseteq \Neighbor(j)$.
\item Delete-Edge: Delete an edge from $E_B$ provided that the base graph $G_D$ remains unchanged.
\item Delete-$L$-Vertex: Delete a vertex $i\in[m]$ on the left side, and remove all the incident edges.
\item  Delete-$L$-Leaf: Delete a vertex $i\in [m]$ on the left side with $|\Neighbor(i)|\leq 1$, and remove the incident edge if any.
\end{itemize}


The following theorem summarizes how these operations influence the existence of gaps.
\begin{thm}\label{thm:reductionrules}
Given an interaction bipartite graph $\GBipartiteGraph=([m],[n],E_B)$, we have
\begin{itemize}
\item[(a)] if $G_B$ is gapful for CLLL, then it remains gapful after applying Delete-Edge;
\item[(b)] if $G_B$ is gapless for CLLL, then it remains gapless after applying Delete-$L$-Vertex;
\item[(c)] $G_B$ is gapful if and only it is gapful after applying Delete-$R$-Leaf, Duplicate-$L$-Vertex, Duplicate-$R$-Vertex, or Delete-$L$-Leaf;
\end{itemize}
\end{thm}

\begin{proof} Throughout this proof, we only consider bipartite graphs which are connected. This restriction does not lose
generality for the following reason. The abstract interior (and commuting interior respectively) of a disconnected bipartite graph is exactly the direct product of the abstract interiors (and commuting interior respectively) of its connected components, which are also bipartite graphs. So a bipartite graph is gapless if and only if each of its connected component is gapless.

\underline{Part (a).} Let $G_B'$ be the bipartite graph obtained from $G_B$ after applying Delete-Edge. Because the base graph is unchanged, $\Interior(G_B)=\Interior(G_B')$. Moreover, it is obvious that $\CInterval(G_B')\subseteq\CInterval(G_B)$. So Part (a) is immediate by Definition \ref{defgap}.

\underline{Part (b).} W.l.o.g., we delete the left vertex $m$ from $G_B$ and obtain a bipartite graph $G_B'$. Let $G_D$ and $G_D'$ denote the base graphs of $G_B$ and $G_B'$ respectively. A key observation is that $I(G_D',\vec{x})\equiv I\left(G_D,(\vec{x}^T,0)\right)$. So $\Interior(G_D')=\{\vec{p}\in\Interior(G_D):p_m=0\}$. Moreover, it is easy to see that $\CInterval(G_B')=\{\vec{p}\in\CInterval(G_B):p_m=0\}$ by Definition \ref{defcommutinginterior}. Now, Part (b) is obvious. 

\underline{Part (c).}  {\it Duplicate-R-Vertex:} Trivial, because both the abstract interior and commuting interior remain unchanged after applying  Duplicate-R-Vertex.

{\it Duplicate-L-Vertex}: W.l.o.g., we duplicate the left vertex $m$ and obtain a bipartite graph $G_B'$. That is, $G_B'=([m+1],[n],E_B')$ where $E_B'=E_B\cup \{(m+1,j):j\in\mathcal{N}(m)\}$. Let $G_D$ and $G_D'$ denote the base graphs of $G_B$ and $G_B'$ respectively. It is straightforward to verify that $I(G_D',\vec{x})\equiv I\left(G_D,(x_1,\cdots,x_{m-1},x_{m}+x_{m+1})\right)$. So $\Interior(G_D')=\{\vec{p}:(p_1,\cdots,p_{m-1},p_m+p_{m+1})\in\Interior(G_D)\}$. Moreover, it is easy to see that $\CInterval(G_B')=\{\vec{p}:(p_1,\cdots,p_{m-1},p_m+p_{m+1})\in\CInterval(G_B)\}$ by Definition \ref{defcommutinginterior}. Then this case is obvious. 

{\it Delete-$R$-Leaf}: W.l.o.g., suppose that the vertex $n$ on the right side satisfies that $|\mathcal{N}(n)|\leq 1$. The nontrivial case is when $|\mathcal{N}(n)|=1$. W.l.o.g., we assume $\Neighbor(n)=\{m\}$. Furthermore, as $G_B$ is connected, the left vertex $m$ has at least two incident edges, and we assume the edge $(m,n-1)$ exists in $G_B$ without loss of generality. Let $G_B'=([m],[n-1],E_B'\})$ denote the bipartite graph obtained from $G_B$ by deleting the vertex $n$ and its incident edge $(m,n)$. On one hand, since $G_B'$ can be obtained from $G_B$ by applying Delete-Edge and then deleting the isolated right vertex $n$, so according to Part (a), if $G_B$ is gapful, then so is $G_B'$. On the other hand, let $\GBipartiteGraph''=([m],[n],E_B'')$ be another bipartite graph where $E_B''=E_B\cup \{(i,n):i\in\mathcal{N}(n-1)\}$. Because $G_B$ can be obtained from $G_B''$ by applying Delete-Edge operations, so according to Part (a), if $G_B$ is gapless, then so is $G_B''$. Moreover, $G_B''$ can be obtained from $G_B'$ by applying Duplicate-$R$-Vertex, so if $G_B''$ is gapless, then so is $G_B'$.


{\it Delete-$L$-Leaf}: W.l.o.g., suppose that the vertex $m$ on the left side satisfies that $|\mathcal{N}(m)|\leq 1$. The nontrivial case is when $|\mathcal{N}(m)|=1$. W.l.o.g., we assume $\Neighbor(m)=\{n\}$ and let $G_B'=([m-1],[n],E_B'\})$ denote the resulted bipartite graph after removing the vertex $m$ and its incident edge $(m,n)$. Besides, we assume $\Neighbor(G_B,n) = \{1,2,\cdots,k,m\}$. As $G_B$ is connected, we have $k\geq 1$.

$\GBipartiteGraph'$ is gapless $\Rightarrow$ $\GBipartiteGraph$ is gapless: Given any rational $\vec{r}= (r_1,\cdots,r_{m}) \in\CInterval(\GBipartiteGraph)$, we will construct an exclusive $\SubspaceSet\sim G_B$, which implies $\GBipartiteGraph$ is gapless by Theorem \ref{Conj:GapGeom}. Define 
\[
\vec{r}' = \left(\frac{r_1}{1 - r_{m}},\frac{r_2}{1 - r_{m}}, \cdots, \frac{r_k}{1 - r_{m}}, r_{k+1},\cdots,r_{m-1}\right).
\] 

We claim that $\vec{r}' \in \CInterval(\GBipartiteGraph')$. By contradiction, suppose that $\vec{r}' \notin \CInterval(\GBipartiteGraph')$. As $\vec{r}'$ is rational, there is a commuting subspace set $\SubspaceSet'\sim\GBipartiteGraph'$ with $\RR(\SubspaceSet')=\vec{r}'$ satisfying $\RR\left(\bigplus_{i\in[m-1]} V_i' \right)=1$. Let $\mathcal{H}_1'\otimes\cdots\otimes\mathcal{H}_n'$ denote the qudits that $\SubspaceSet'$ acts on. We construct a commuting subspace set $\SubspaceSet\sim\GBipartiteGraph$ as follows. Suppose $1 - r_{m} = a/b$ where $a$ and $b$ are integers. Let $\mathcal{H}_n = \mathcal{H}'_n \otimes \mathcal{H}_n^c$ where $\dim{(\mathcal{H}_n^c)} = b$, and $\mathcal{H}'_j = \mathcal{H}_j$ for any $j\in [n-1]$. Pick an arbitrary $a$-dimensional subspace $W$ of $\mathcal{H}_n^c$ and let
\begin{itemize}
\item $V_i =V_i' \otimes W$ for each $i\in[k]$,
\item $V_i =V_i'\otimes \mathcal{H}_n^c $ for each $k<i\leq m-1$, and
\item $V_{m}= W^{\perp}\otimes\left(\bigotimes_{i=1}^n \mathcal{H}'_i\right)$.
\end{itemize}
It is easy to verify that $\SubspaceSet$ is commuting, $\SubspaceSet \sim \GBipartiteGraph$,  $
\RR\left(\sum_{i\in[m]} V_i \right)=1$, and $
\RR(\SubspaceSet) = \vec{r}$. So $\vec{r}\notin\CInterval(G_B)$. A contradiction.

Because $\vec{r}'\in \CInterval(\GBipartiteGraph')$ and $\GBipartiteGraph'$ is gapless, according to Theorem \ref{Conj:GapGeom}, there is an exclusive subspace set $\SubspaceSet'\sim \GBipartiteGraph'$ where $\RR(\SubspaceSet')=\vec{r}'$. We construct $\SubspaceSet$ from $\SubspaceSet'$ as above. It is easy to verify that $\SubspaceSet\sim G_B$, $\RR(\SubspaceSet) = \vec{r}$, and $\SubspaceSet$ is exclusive.

$\GBipartiteGraph$ is gapless $\Rightarrow$ $\GBipartiteGraph'$ is gapless: Note that $G_B'$ can be obtained from $G_B$ by applying the Delete-L-Vertex (deleting the left vertex $m$). So $G_B'$ is gapless according to Part (b).
\end{proof}
With these reduction rules, it is easy to see all trees are gapless, including 1-D chains \cite{movassagh2010unfrustrated} and regular trees \cite{shearer1985problem,heilmann1972theory,coudron2012unfrustration,pnas}.

\begin{thm}\label{thm:treesgapless}
	An interaction bipartite graph $\GBipartiteGraph$ is gapless for CLLL if it is a tree.
\end{thm}

\begin{proof} Suppose $G_B=([m],[n],E_B)$ is a tree. We can obtain the bipartite graph $([1],[1], \{(1,1)\})$ from $G_B$ by applying Delete-$L$-Leaf or Delete-$R$-Leaf repeatedly. Obviously, $G_B'$ is gapless. So $G_B$ is also gapless according to Theorem \ref{thm:reductionrules}.
\end{proof}

\subsection{An Almost Complete Characterization of Gap Existence}\label{sec:gapexistence}

In this subsection, we prove Theorem \ref{cor:basegraphCyclegapful}, which gives an almost complete characterization of gap existence except when the base graph has only 3-cliques.
\begin{thm}\label{cor:basegraphCyclegapful}
	Given an interaction bipartite graph $G_B$, $G_B$ is gapless for CLLL if its base graph is a tree, and gapful for CLLL if its base graph has an induced cycle of length at least 4.
\end{thm}
\begin{proof}
Let $G_D$ denote the base graph of $G_B$. If $G_D$ is a tree, then $G_B$ is gapless for VLLL by Theorem 6 in~\cite{he2017variable}, so it is also gapless for CLLL.

Suppose $G_D$ has an induced cycle of length at least 4. W.l.o.g., we assume the induced subgraph of $G_D$ on $[\ell]$, where $\ell\geq 4$, is a cycle $([\ell],E_C)$ where $E_C=\{(i,i+1):i\in[\ell-1]\}\cup \{(\ell,1)\}$. We apply Delete-$L$-Vertex on $G_B$ to delete all left vertices not in $[\ell]$ and then apply Delete-$R$-Leaf to delete all right vertices whose degree are at most 1, then we obtain a bipartite graph $G_B'=([\ell],S,E_B')$ such that for any $j\in S$, either $\mathcal{N}(G_B',j)=\{i,i+1\}$ for some $i\in[\ell-1]$ or $\mathcal{N}(G_B',j)=\{\ell,1\}$. So all $j\in S$ are solitary. Moreover, by Theorem 7 in \cite{he2017variable}, such $G_B'$ is gapful for VLLL, so $G_B'$ is also gapful for CLLL according to Theorem \ref{thm:2-loc}. Finally, we can conclude $G_B$ is gapful for CLLL according to Theorem \ref{thm:reductionrules} (b) and (c).  
\end{proof}

\section*{Acknowledgement}
We are grateful to Andr{\'a}s Gily{\'e}n for his manuscript about VLLL and telling us the problem whether Shearer's bound is tight for QLLL. We also thank the anonymous referees for their valuable comments.


\bibliographystyle{plain}

\end{document}